\documentclass[12pt]{article}

\RequirePackage{graphicx}
\RequirePackage{amsthm}
\RequirePackage[cmex10]{amsmath}
\RequirePackage{natbib}
\RequirePackage[colorlinks,citecolor=black,urlcolor=black]{hyperref}
 
\newtheorem{notation}{Notation}
\newtheorem{lemma}{Lemma}
\newtheorem{fact}{Fact}
\newtheorem{claim}{Claim}
\newtheorem{remark}{Remark}
\newtheorem*{claim*}{Claim}
\newtheorem{theorem}{Theorem}

\newtheorem{example}{Example}
\newtheorem{corollary}{Corollary}
\newtheorem{assumption}{Assumption}
\newtheorem{proposition}{Proposition}
\newtheorem{definition}{Definition}

\usepackage[top=1.18in, bottom=1.18in, left=1.2in, right=1.2in]{geometry}

\newtheorem{assumptionalt}{Assumption}[assumption]

\newenvironment{assumptiona}[1]{
  \renewcommand\theassumptionalt{#1}
  \assumptionalt
}{\endassumptionalt}

\usepackage{txfonts}
\usepackage{amsfonts}
\usepackage{epsfig,setspace}
\usepackage{amsthm,array,geometry, wasysym,enumerate}
\usepackage{sgame}
\usepackage[font=small,labelfont=bf]{caption}
\usepackage{subcaption}
\usepackage{comment}
\usepackage[normalem]{ulem}
\usepackage{amsmath,amssymb,amsthm}
\usepackage{setspace,graphicx}
\usepackage{lmodern}
\usepackage{slantsc}%
\usepackage[cmtip,all]{xy}
\usepackage{accents}
\usepackage[cmyk]{xcolor}

\usepackage{pgfplots}
\pgfplotsset{compat=1.16}

\usepackage{subcaption} 
\usepgfplotslibrary{fillbetween}
\usetikzlibrary{patterns}
\usetikzlibrary{arrows}
\usetikzlibrary{shapes,decorations} 
\usetikzlibrary{positioning,chains,fit,calc} 

\usepackage{ifthen}
\usepackage{xfrac}
\usepackage{epstopdf}
\usepackage{times} 
\usepackage{titlesec}

\usetikzlibrary{calc,arrows,automata,shapes.misc,shapes.arrows,chains,matrix,positioning,scopes,decorations.pathmorphing,shadows}

\newcommand{\citeapos}[1]{\citeauthor{#1}'s (\citeyear{#1})}

\newcommand{\co}{\mathrm{co}}

\newcommand{\ext}{\mathrm{ext}}
\newcommand{\marg}{\mathrm{marg}}
\newcommand{\ddd}{\mathrm{ d}}
\newcommand{\dd}{\mathrm{d}}
\newcommand{\supp}{\mathrm{supp}}
\newcommand{\bbE}{\mathbb E}
\newcommand{\bbExp}[1]{\bbE \left[#1\right]}

\newcommand{\real}{\mathbb{R}}
\newcommand{\sequence}[1]{(#1)_{n=1}^{\infty}}

\newcommand{\longsquiggly}{\xymatrix{{}\ar@{~>}[r]&{}}}

\newcommand{\feas}{\mathrm{feas}}

\newcommand{\neigh}{\mathcal{N}}

\newcommand{\Menu}{{\mathcal{A}}}
\newcommand{\menu}{B}

\newcommand{\scr}{s}
\newcommand{\scrb}{t}
\newcommand{\scrc}{r}
\newcommand{\SCR}{S}
\newcommand{\SCRset}{T}
\newcommand{\SCRrat}{R}
\newcommand{\SCRU}{{\mathcal S}}
\newcommand{\scru}{t}
\newcommand{\ut}{u}
\newcommand{\utb}{v}
\newcommand{\UT}{{\mathcal U}}
\newcommand{\UTB}{{\mathcal V}}
\newcommand{\posteriorut}{v}

\newcommand{\prior}{{\mu_0}}
\newcommand{\posterior}{\mu}
\newcommand{\posteriorb}{\tilde\mu}
\newcommand{\Info}{{\mathcal P}}

\newcommand{\maxval}{V}

\newcommand{\Infof}{{\Info^F}} 
\newcommand{\DO}{{\Delta\Omega}}
\newcommand{\andinf}{\cup\{\infty\}}
\newcommand{\ereal}{\real\andinf}
\newcommand{\extreal}{\overline{\real}}
\newcommand{\erealp}{\real_+\andinf}
\newcommand{\extrealp}{\overline{\real}_+}
\newcommand{\cost}{\kappa}
\newcommand{\dcost}[1]{{d^{+}_{#1}\cost}}
\newcommand{\icost}{C}
\newcommand{\icostd}{c}
\newcommand{\icostdd}{{\nabla c}}
\newcommand{\icostD}{\mathcal{C}}
\newcommand{\KL}{K}
\newcommand{\transf}{\psi}
\newcommand{\icostdset}{\tilde{\icostD}}
\newcommand{\ip}{p}
\newcommand{\ipb}{q}
\newcommand{\mps}{m}
\newcommand{\mip}{Q}
\newcommand{\lp}[1]{{L^{#1}(\prior)}}

\newcommand{\Dom}{\SCR^\cost}
\newcommand{\Domicost}{\Info^{\icost}}

\newcommand{\suppf}{\sigma}
\newcommand{\mult}{\gamma}
\newcommand{\wt}{\tau}
\newcommand{\func}{f}
\newcommand{\funcb}{g}
\newcommand{\Func}{{\mathcal F}}
\newcommand{\nuis}{\lambda}

\newcommand{\dscr}{\beta}
\newcommand{\data}{\mathcal{D}}
\newcommand{\dgraph}{G}
\newcommand{\dgraphv}{V}
\newcommand{\dgraphE}{E}
\newcommand{\DMenu}{\mathcal{B}}
\newcommand{\param}{\theta}
\newcommand{\paramb}{\eta}
\newcommand{\paramc}{\zeta}
\newcommand{\refact}{\underline{a}}
\newcommand{\cycle}{\chi}
\newcommand{\Cycle}{\mathcal{X}}
\newcommand{\cyclev}{\ell}
\newcommand{\Cyclev}{\mathcal{L}}

\newcommand{\acomment}[1]{}

\newcommand{\samplepoints}{10}
\makeatletter
  \renewcommand\@seccntformat[1]{\csname the#1\endcsname.{\hskip.7em\relax}} 
\makeatother

\title{{\bf Predicting Choice from Information Costs
}\footnote{We would like to thank Ben Brooks, Mark Dean, Tommaso Denti, Mira Frick, Daniel Gottlieb, Ryota Iijima, Navin Kartik, R. Vijay Krishna, Ce Liu, Daniel Martin, Stephen Morris, Pietro Ortoleva, John Rehbeck, Marzena Rostek, Kareen Rozen, and Omer Tamuz for their comments. Tianhao Liu provided excellent research assistance.}
}

\author{
\begin{minipage}{0.3\textwidth}\centering  
Elliot Lipnowski\footnote{\texttt{e.lipnowski@columbia.edu}} \\ \centering \it \small Columbia University
\end{minipage}                  
\begin{minipage}{0.3\textwidth}\centering 
Doron Ravid\footnote{\texttt{dravid@uchicago.edu}}  \\ \centering \it \small University of Chicago 
\end{minipage} 
}

\newcommand{\argmax}{\mathrm{argmax}}

\date{\vspace{0.8cm} \today}
\begin{document}

\maketitle
\begin{abstract}

\noindent 

An agent acquires a costly flexible signal before making a decision. We explore to what degree knowledge of the agent's  information costs helps predict her behavior. We establish an impossibility result: learning costs alone generate no testable restrictions on choice without also imposing constraints on actions' state-dependent utilities. By contrast, choices from a menu often uniquely pin down the agent's decisions in all submenus. To prove the latter result, we define \emph{iteratively differentiable} cost functions, a tractable class amenable to first-order techniques. Finally, we construct tight tests for a multi-menu data set to be consistent with a given cost.
\vspace{.2cm}

\end{abstract}
\newpage 

\begin{spacing}{1}
\onehalfspacing


\section{Introduction}
Sparked by \citeauthor{sims2003implications}' (\citeyear{sims1998stickiness,sims2003implications,sims2006rational}) studies on rational inattention, the last two decades have seen a growing interest in models of costly flexible learning. Compared with the traditional framework of decision-making under uncertainty, these models postulate that the agent's behavior depends on an additional parameter: the cost of information acquisition. The appropriate values for this parameter have been the subject of intense inquiry. Some studies, such as \cite{caplin2015revealed}, \cite{de2017rationally}, \cite{dean2019experimental}, \cite{dewan2020estimating},  and \cite{caplin2020rational}, seek to answer this question empirically, developing tools for identifying, measuring, and testing the agent's cost function in experiments. Other studies take an axiomatic approach, advocating for classes of parameterized cost functions based on their characterizing features \citep[e.g.,][]{pomatto2020cost,hebert2021neighborhood,caplin2021rationally}. An alternative string of papers explores which cost functions can be microfounded via dynamic learning \citep[e.g.,][]{morris2019wald,bloedel2020cost, hebert2021time} or as a physical cost of performing experiments \citep{denti2021experimental}. In this paper, we ask a complementary question: In what ways can one use the agent's information-acquisition cost to predict her behavior?

To answer our question, we study a canonical flexible-learning model in which an agent chooses from a finite set of alternatives, the benefits from which depend on a stochastic state. Before making her decision, the agent chooses what signal to acquire about this state. Learning comes at a cost, which the agent subtracts from the expected benefit she derives from her final decision. 
After choosing her information, the agent observes a signal realization and takes an action. Following the literature \citep{caplin2015revealed}, we refer to the resulting probability in which the agent takes each action conditional on the state as the agent's \emph{stochastic choice rule} (SCR).

We begin by studying one's ability to forecast the agent's choices using only her cost of information acquisition. Proposition~\ref{prop: dense set of rationalizable SCRs} shows this exercise is essentially futile: holding the agent's costs fixed, one can approximate any finite-cost SCR arbitrarily well with SCRs that are optimal for some utility function. Therefore, other than ruling out some SCRs as technologically infeasible, the agent's learning costs make no predictions that could be falsified with finite data. 

Proposition~\ref{prop: dense set of rationalizable SCRs} leaves open the possibility that the approximating SCRs are rationalized via indifference or with a small set of utilities. For example, if learning is free, the agent optimally takes all actions with positive probability in all states only if her action does not influence her payoffs. In such cases, an analyst may be able to obtain meaningful predictions by ruling out a knife-edge set of utilities or by using other considerations to refine the set of optimal SCRs. Theorem~\ref{thm: dense set of uniquely rationalizable SCRs} shows such refinements cannot meaningfully increase the set of restrictions imposed on the agent's behavior by her cost function, provided that this function is strictly increasing in informativeness and the set of feasible information structures is sufficiently flexible. 

The above-mentioned results imply the agent's cost function on its own imposes few restrictions on the agent's behavior in a single menu. However, we show the agent's costs can significantly restrict the way the agent behaves \emph{across} menus. To make this point, we focus on a novel class of smooth cost functions that we call \emph{iteratively differentiable}. We prove such cost functions allow one to solve for the agent's optimal SCR given a fixed utility function using a first-order condition (Lemma~\ref{lem: posterior separable approximation characterizes optimality}). One can also apply this condition to do the reverse, namely, to solve for the utility functions that make a fixed SCR optimal. In Theorem~\ref{thm: cross-choice predictions}, we use the ability to invert utility from choices to show that, for finite-valued iteratively differentiable costs satisfying an infinite-slope condition, one can often use the agent's actions in one menu to pin down her behavior in all submenus. Thus, under the theorem's conditions, once the analyst knows the agent's cost of learning, the agent's choice when facing the grand menu is the sole remaining degree of freedom in the agent's behavior. 

That the agent's cost function restricts her behavior across menus is useful not only for predicting her actions, but also for testing whether or not her decisions are consistent with a given cost function. In section~\ref{sec: cycle test} we develop such a test, assuming the given cost function is iteratively differentiable and finite-valued. Our test takes as its starting point a data set describing the agent's chosen SCR in a collection of menus. From this data set, our test constructs a bipartite graph with nodes labeled by actions and menus, and with an edge between an action and a menu if the agent sometimes takes said action when faced with that menu. We identify each cycle in this graph with a set of equations, and show these equations are satisfied for all cycles if and only if the data set is compatible with the agent having a common utility and the pre-specified cost function. Moreover, it suffices to check the equations associated with a small set of cycles, called a \emph{cycle basis}. We also prove our test is tight for cost functions satisfying the prerequisites of Theorem~\ref{thm: cross-choice predictions}: testing the equations corresponding to any collection of cycles that does not include a cycle basis is insufficient for proving consistency of a data set. 

Several studies consider the problem of testing whether a data set is consistent with costly information acquisition, and if so, whether one can use these data to make inferences about the agent's cost of information \citep[e.g.,][]{caplin2015revealed,caplin2017rationally,de2017rationally,caplin2020rational,chambers2020costly,dewan2020estimating,denti2022posterior,lin2022stochastic}.
This literature typically assumes that, in addition to observing the agent's choices, the analyst also knows the agent's utility function, which is either observed directly, or inferred from decision problems in which information acquisition plays no role. Our paper contributes to this literature by delineating the kind of data needed for testing various hypotheses regarding the agent's information-acquisition costs. 
In particular, we show that without utility information, one cannot test such hypothesis using the agent's behavior in a single menu. We also provide a method for using variation in  the agent's menu to check whether or not the agent faces a given hypothetical cost function. 

Other papers have advocated the use of various cost functions based on non-behavioral justifications. For example, \cite{mensch2018cardinal}, \cite{pomatto2020cost}, and \cite{hebert2021neighborhood}  pin down particular functional forms for the agent's learning costs by stating properties directly on the agent's cost function rather than on the decisions that it generates.\footnote{Relatedly, \cite{frankel2019quantifying} study axiomatically the appropriate measures of the ex-post value of information.} Other papers, such as \cite{morris2019wald}, \cite{bloedel2020cost}, and \cite{hebert2021time}, pin down specific functional forms by asking which static cost functions can be micro-founded as coming from a dynamic learning process.\footnote{A related study is the one by \cite{denti2021experimental}, which obtain conditions under which a cost function defined over the agent's distribution of posterior beliefs for all priors can be micro-founded as a prior-independent cost over experiments.} Our paper complements this literature by studying the restrictions imposed by a given cost function, and providing a blueprint for generating predictions and testing hypotheses on the agent's information acquisition costs without the need for a comprehensive revealed preference analysis.

The paper also provides various technical results on models with costly information acquisition, including: regularity properties of the indirect cost function over SCRs (Lemma~\ref{lemma: Properties of Indirect Cost}); an algebraic sufficient condition for an optimal SCR to be uniquely optimal (Proposition~\ref{prop: linear independence implies unique rationalizability}); a Langrangian first-order condition for optimality of a given SCR under smooth costs (Proposition~\ref{prop: multiplier result for iteratively differentiable case}); a reduction of the agent's problem with differentiable costs to that under its posterior-separable approximations (Lemma~\ref{lem: posterior separable approximation characterizes optimality}); and a generic uniqueness result for optimal SCRs (Lemma~\ref{lem: unique prediction with known benefits}). These results apply to cost functions and environments not covered by the extant literature. Despite the potential usefulness of these technical results, we nevertheless see our main contribution as the study of the testable content of information costs.

\section{Model}\label{sec: model}

An agent makes a decision from a finite set $A$ of actions with $|A|>1$. The payoff from each action depends on a payoff state $\omega$ belonging to some compact metric space $\Omega$ and distributed according to some probability measure $\prior\in\DO$.\footnote{We view any Polish space $Y$ as a measurable space, endowed with its Borel sigma-algebra. Let $\Delta Y$ denote the set of probability measures on this measurable space, and interpret linear combinations of elements of $\Delta Y$ as pointwise linear combinations. Unless otherwise stated, we endow $\Delta Y$ with the weak* topology generated by continuous bounded functions, which in turn makes $\Delta Y$ a compact metrizable space if $Y$ is compact. We let $\supp(\gamma)$ denote the support of $\gamma$ for any $\gamma\in\Delta Y$. 
}
Without loss of generality, we take $\prior$ to have full support. 
The agent's utility from choosing action $a\in A$ in state $\omega\in\Omega$ is $\ut_a(\omega)$. We assume $\ut_{a} \in \lp{1}$, 
meaning the agent's expected payoff $\bbExp{\ut_a}$ from every action $a\in A$ is well defined, and refer to $\ut:=(\ut_{a})_{a\in A} \in \UT:=\lp{1}^A$ as the agent's \textbf{utility function}.\footnote{For any random variable $f:\Omega\to\real$, let $\bbE [f]:=\int f(\omega)\ \prior(\dd\omega)$ whenever the latter Lebesgue integral is well defined. 
Recall the space $\lp{1}$ is the set of all measurable functions $f:\Omega \to \real$ such that $\bbE |f|<\infty$, 
modulo $\prior$-almost sure equivalence, equipped with the $\lp{1}$ norm, $\Vert f \Vert_{1} := \bbE |f|.$
} 

\paragraph*{\emph{Information Policies.}}

Before taking her action, our agent chooses what signal to observe in order to learn about $\omega$. We assume our agent is Bayesian, and model the agent's information via the distribution of her posterior beliefs, $\ip \in \Delta\DO$. Specifically, we allow the agent to choose any $\ip$ whose average equals the prior, $\int \posterior \ \ip (\dd\posterior) = \prior$, which is equivalent to letting her select any $\ip$ that originates from some signal structure.\footnote{See, for example, \cite{Aumann1995}, \cite{Kamenica2011}, or \cite{benoit2011apparent}. For the sake of completeness, Appendix~\ref{app: sec: bayes plausibility} provides a proof that applies to the case of infinite states.\label{footnote: splitting lemma}} We refer to any $\ip$ that averages back to the prior as an \textbf{information policy}, and denote the set of all information policies by $\Info$. Occasionally, we will refer to the set $\Infof\subseteq\Info$ of \textbf{simple information policies}, namely those with finite support. Relatedly, a \textbf{simply drawn posterior} is a posterior that belongs to the support of some simple information policy.\footnote{Equivalently, simply drawn posteriors are the elements of $\{\posterior\in\DO:\ \epsilon\posterior\leq\prior \text{ for some } \epsilon>0\}$. To see the equivalence, note any element $\posterior$ of this set (as witnessed by $\epsilon\in(0,1)$) is in the support of $\ip=\epsilon\delta_\posterior+(1-\epsilon)\delta_{\tfrac{\prior-\epsilon\posterior}{1-\epsilon}}$.} 

We order information policies via informativeness in the sense of \cite{Blackwell1953}. Specifically, we say $\ip \in \Info$ is \textbf{more informative} than $\ipb \in \Info$ (written as $\ip \succeq \ipb$) if $\ip$ is a mean-preserving spread of $\ipb$.\footnote{That is, some measurable $\mps:\DO \to \Delta \DO$ exists such that $\int \posteriorb \ \mps(\dd\posteriorb|\posterior) = \posterior$ for all $\posterior \in \DO$, and $\ip(D) = \int \mps(D|\posterior) \ \ipb(\dd\posterior)$ for all Borel $D\subseteq \DO$.} An information policy $\ip$ is \textbf{strictly more informative} than $\ipb$ (written $\ip \succ \ipb$) if $\ip \succeq \ipb$ and $\ip \neq \ipb$. Intuitively, $\ip$ is more informative than $\ipb$ if observing $\ip$ is equivalent to observing $\ipb$ and an additional signal. For a review of the connection between this information ranking and other notions of informativeness, see \cite{khan2019information}.

\paragraph*{\emph{Information Costs.}}
Information comes at a cost that is summarized via a function taking values in the extended nonnegative reals,
\[
\icost:\Info\to\extrealp:=\erealp.
\]
By letting $\icost$ take a value of infinity, we allow it to encode constraints on the set of feasible information policies. We assume $\icost$ is a lower-semicontinuous convex function that is \textbf{proper}---that is, not globally equal to $\infty$---and let $\Domicost = \icost^{-1}(\mathbb{R})$ denote its \textbf{(effective) domain}, namely, the nonempty set of information policies that can be induced at finite cost. We also require $\icost$ to be \textbf{monotone}, meaning $\icost(\ip) \geq \icost(\ipb)$ whenever $\ip \succeq \ipb$. Sometimes, we focus on costs that are \textbf{strictly monotone}, by which we mean $\icost(\ip) > \icost(\ipb)$ holds whenever $\ip \succ \ipb$ and $\ipb \in \Domicost$. As we explain in section~\ref{sec: discussion}, assuming $\icost$ is convex and monotone is without loss of generality. 

\paragraph*{\emph{The Agent's Problem.}}
After choosing her information, the agent observes her signal realization and takes an action. An \textbf{action strategy} is a measurable mapping, $\alpha:\DO\to\Delta A$, where $\alpha(a|\posterior)$ is the probability the agent chooses action $a$ when her posterior belief is $\posterior$. A  \textbf{strategy} is an information policy $\ip$ paired with an action strategy $\alpha$. The agent's payoff from the strategy $(\ip,\alpha)$, given benefit $\ut\in\UT$, is 
\[
\int_\DO \int_A \ut_a \ \alpha(\dd a|\posterior)\ \ip(\dd \posterior)-\icost(\ip)
.\]
Observe some strategy yields a finite value, because some information policy yields finite cost. We say $(\ip,\alpha)$ is $\ut$\textbf{-optimal} if it maximizes the above objective among all possible strategies.

\paragraph*{\emph{Stochastic Choice Rules.}} 
We follow \cite{caplin2015revealed} and summarize the agent's behavior via a state-dependent stochastic choice rule. Within our more general setting, such a rule consist of a vector of (essentially bounded) mappings,
\[
\scr = (\scr_{a})_{a\in A},
\] 
where $s_{a}\in \lp{\infty}$ gives the conditional probability the agent takes action $a$ given the state.\footnote{Recall $\lp{\infty}$ is the space of all measurable functions $f:\Omega \rightarrow \mathbb{R}$ (identified up to $\prior$-almost sure equality) that are bounded $\prior$-almost surely, equipped with the $\lp{\infty}$ norm, $\Vert f \Vert_{\infty}:=\text{ess }\sup |f|.$} Because probabilities are positive and add up to $1$, we must have $\scr_{a} \geq \mathbf{0}$ for all $a$ and $\sum_{a\in A} \scr_{a} = \mathbf1.$ We refer to every $\scr \in \SCRU:=\lp{\infty}^{A}$ that satisfies these constraints as a (state-dependent) \textbf{stochastic choice rule (SCR)},\footnote{
Because $A$ is finite, the Radon-Nikodym Theorem implies that the probability $\scr_{a}(\omega)$ of action $a$ conditional on state $\omega$ is well-defined up to $\prior$-almost sure equivalence. Consequently, a stochastic choice rule naturally lives in $\SCRU=\lp{\infty}^{A}$, which does not distinguish between functions that agree $\prior$-almost surely. An equivalent formalism (given a disintegration theorem) would have stochastic choice rules living in the set of probability measures over $A\times\Omega$ with marginal $\prior$ on $\Omega$.} and denote the set of all SCRs by $\SCR$. 

The support of an SCR $\scr$ is the set of actions it generates with positive probability, $\supp \ s = \{a\in A: \scr_{a} \neq \mathbf{0}\}.$ An SCR $\scr$ has \textbf{full support} if it uses all actions, that is, if $\supp \ s = A.$ The SCR has \textbf{conditionally full support} if it uses all actions in all states, meaning $\scr_a$ is $\prior$-almost surely strictly positive for every $a\in A$.

A strategy $(\ip,\alpha)$ \textbf{induces} an SCR $\scr$ if for every action $a\in A$ and event $\hat\Omega \subseteq \Omega$,
$$\int \alpha(a|\posterior) \ \posterior(\hat\Omega) \ \ip(\dd \posterior) =\bbExp{\mathbf1_{\hat\Omega}\ \scr_a}.
$$ 
An information policy $\ip\in\Info$ \textbf{can induce} $\scr\in\SCR$ if $\ip$ can describe the information the agent receives in some strategy that results in $\scr$, that is, if $(\ip,\alpha)$ induces $\scr$ for some $\alpha$.

\paragraph*{\emph{Rationalizability}} We are interested in understanding which stochastic choice rules are optimal for a \emph{given} objective, which are optimal for \emph{some} objective, and which can be \emph{uniquely} optimal. Given a utility function $\ut \in \UT$, we say $\scr$ is $\ut$-\textbf{rationalizable} if it is induced by some $\ut$-optimal strategy $(\ip,\alpha)$. A stochastic choice rule is \textbf{uniquely} $\ut$-\textbf{rationalizable} if it is the only $\ut$-rationalizable SCR. An SCR that is $\ut$-rationalizable for \emph{some} $\ut$ is \textbf{rationalizable}. We also say $\scr$ is \textbf{uniquely rationalizable} if it is uniquely $\ut$-rationalizable for some $\ut$.\footnote{In other words, $\scr$ is uniquely rationalizable if some $\ut$ exists such that $\scr$ is the only $\ut$-rationalizable SCR; this notion is distinct from requiring that $\ut$ be the unique utility that rationalizes $\scr$.}

\subsection*{Example Cost Functions}\label{sec: examples}
In this subsection, we provide some example cost functions satisfying our hypotheses.

\begin{example}[Mutual Information Costs] \label{ex: Mutual Info Costs}
Let $\KL: \Delta \Omega \rightarrow \extreal$ be the Kullback-Leibler divergence from the prior $\prior$,
\begin{equation*}
\KL (\posterior) =
\begin{cases}
\int \log \frac{\dd \posterior}{\dd \prior}(\omega) \ \posterior (\dd \omega) & \text{if }\posterior \ll \prior,\\
\infty & \text{otherwise.}
\end{cases}
\end{equation*}
By \cite{posner1975random}, $\KL$ is lower semicontinuous. The \textbf{mutual information} cost function is given by 
\begin{equation*}
    \icost(\ip) = \int \KL (\posterior) \ \ip (\dd \posterior).    
\end{equation*}
This cost function was first introduced by Sims (\citeyear{sims1998stickiness,sims2003implications,sims2006rational}) and has served as the workhorse cost function of the literature on rational inattention. See \cite{Csiszar1974}, \cite{matejka2015rational}, \cite{caplin2019rational}, and \cite{denti2020note} for characterizations of optimal behavior under mutual information costs.
\end{example}

\begin{example}[Log-Likelihood Ratio Costs] \label{ex: LLR costs}
\cite{pomatto2020cost} proposed and axiomatized the Log-likelihood ratio (LLR) cost function. The LLR cost function is defined for a finite $\Omega$, and is parameterized by a matrix of positive numbers, $\boldsymbol{\param}\in \mathbb{R}^{\Omega \times \Omega}_{+}$ such that $\boldsymbol{\param}_{\omega,\omega}=0$ for all $\omega\in \Omega$. Given $\boldsymbol{\param}$, define the function $\icostd^{LLR}_{\boldsymbol{\param}}: \Delta\Omega \rightarrow \extreal$ via
\[
\icostd^{LLR}_{\boldsymbol{\param}}(\posterior) = \sum_{\omega,\omega'}\mathbf{\param}_{\omega,\omega'} \frac{\posterior(\omega)}{\prior(\omega)} \log \frac{\posterior(\omega)}{\posterior(\omega')}.
\]
Then the $\boldsymbol{\param}$-LLR cost function is given by 
\[ 
\icost^{LLR}_{\boldsymbol{\param}}(\ip) = \int \left[\icostd^{LLR}_{\boldsymbol{\param}}(\posterior) - \icostd^{LLR}_{\boldsymbol{\param}}(\prior)\right] \ \ip(\dd\posterior).
\]
\end{example}

\begin{example}[Posterior Separable Costs] \label{ex: Posterior Separable}
Posterior separable costs are a generalization of mutual information costs introduced by \cite{caplin2017rationally}. Let $\icostD$ be the set of convex, lower semicontinuous functions from $\DO$ to $\extreal$ that assign $\prior$ a finite nonnegative value.\footnote{Given the other assumptions on elements of $\icostD$, the assumption that $\icostd(\prior)$ is finite is equivalent to the induced information cost $\icost$ being proper. Also note, when $\Omega$ is finite and $\icostd$ is globally finite, $\icostd$ is lower semicontinuous if and only if it is continuous.} We say $\icost$ is \textbf{posterior separable} if some $\icostd \in \icostD$ exists such that every $\ip\in\Info$ has
\[
\icost(\ip) = \int \icostd (\posterior)\ \ip (\dd \posterior).
\]
In addition to mutual information (Example~\ref{ex: Mutual Info Costs}), the class of posterior separable costs includes the log-likelihood ratio cost function from Example~\ref{ex: LLR costs}, and the neighborhood-based cost function studied by \cite{hebert2021neighborhood}. 
\end{example}

\begin{example}[Transformed Costs]\label{ex: Transformed Costs}
Let $\transf:\real_+ \rightarrow \extrealp$ be a nondecreasing, proper, convex, lower semicontinuous function. Then,
\[
\icost(\ip) = \transf\left(\int \icostd (\posterior) \ \ip(\dd \posterior) \right)
\]
satisfies our assumptions for any $\icostd\in\icostD$.
\end{example}

\begin{example}[Quadratic Costs]\label{ex: Quadratic Costs}
Let $\tilde\icostd:\DO \times \DO \rightarrow \real$ be a symmetric, lower semicontinuous function that is convex in each argument and has $\tilde\icostd(\prior,\prior)\geq0$. Let
\[
\icost(\ip) = \int \int \tilde\icostd(\posteriorb,\posterior) \ \ip (\dd \posteriorb) \ \ip(\dd \posterior).
\]
Then, $\icost$ satisfies our assumptions as long as $\tilde\icostd$ is \emph{positive semidefinite}, namely, as long as
\[
\int\int \tilde\icostd(\posteriorb,\posterior) \ (\ip - \ipb)(\dd \posteriorb) \ (\ip - \ipb)(\dd \posterior) \geq 0
\]
holds for all $\ip, \ipb \in \Info$.
\end{example}

\begin{example}[Maximum over a Set]\label{ex: Maximum Costs}
Let $\icostdset\subseteq\icostD$ be a compact (with respect to the supremum norm) set of continuous functions. Then,
\[
\icost(\ip) = \max_{\icostd\in\icostdset}\int \icostd (\posterior) \ \ip(\dd \posterior)
\]
satisfies our assumptions.
\end{example}

\section{Cost Minimization}\label{sec: cost minimization}

We begin our analysis by studying the cheapest way to induce a given SCR. Specifically, we solve
\begin{equation*}
\cost(\scr) = \inf_{\ip \in \Info} \icost(\ip) \text{ s.t. }\ip \text{ can induce }\scr.
\end{equation*}
Note $\scr$ can be rationalizable only if the above program admits a solution: otherwise, no $(\ip,\alpha)$ that induces $\scr$ can ever be optimal, because one can always attain the same utility with lower costs. 

To solve the cost-minimization program, observe that every $\scr$ can be viewed as a signal structure whose realizations take the form of recommended actions. Therefore, one can apply Bayes rule to transform any $\scr$ into its associated information policy, $\ip^{\scr} \in \Info$. Formally, let
\[
\ip^{s}_{a} := \bbExp{\scr_a}
\]
be the ex-ante probability that $\scr$ generates the recommendation $a$. Whenever $\ip_{a}^{s}>0$, Bayes' rule dictates the agent's posterior belief $\posterior_{a}^{\scr}$ conditional on the realized action recommendation being $a$ is given by\footnote{When $\ip_{a}^{\scr}=0$, we can let $\posterior_{a}^{\scr}\in\DO$ be arbitrary wherever the term appears.}
\[
\posterior_{a}^{\scr}(\dd\omega) = \frac{\scr_{a}(\omega)}{\ip_{a}^{\scr}}\prior(\dd\omega).
\]
It follows one can write $\ip^{s}$ as
\[
\ip^{\scr} = \sum_{a\in A} \ip^{\scr}_{a} \delta_{\mu_{a}^{\scr}},
\]
where $\delta_{\mu_{a}^{\scr}}\in \Delta\DO$ denotes the distribution that generates $\mu_{a}^{\scr}$ with probability $1$. If every $a$ is associated with a different $\posterior_{a}^{\scr}$, then $\ip^{\scr}$ is the information policy that yields the posterior $\mu_{a}^{\scr}$ with probability $\ip^{\scr}_{a}$. We follow the literature \citep{caplin2015revealed} and refer to $\posterior_{a}^{\scr}$ as the \textbf{revealed posterior} of $a$ given $\scr$, and $\ip^{\scr}$ as $\scr$'s \textbf{revealed information policy}.

\begin{lemma}\label{lem: which policies can induce}
Policy $\ip \in \Info$ can induce $\scr\in \SCR$ if and only if $\ip \succeq \ip^{s}$. Therefore, $\cost(\scr) = \icost(\ip^{\scr}).$
\end{lemma}
Similar results are prevalent in the literature for less general settings \citep[e.g.,][]{caplin2015revealed}. Despite the difference in generality, the intuition is identical. If $(\ip,\alpha)$ induces $\scr$, one can generate $\ip^{\scr}$ by first drawing a posterior-action pair $(\posterior,a)$ according to $(\ip,\alpha)$ and then revealing only the realized action to the agent. Clearly, seeing $a$ alone is less informative than seeing both $a$ and $\posterior$. But because $a$ is independent of the state conditional on $\posterior$, seeing $\posterior$ and $a$ together is just as informative as observing $\posterior$ on its own. In other words, $\ip$ is more informative than $\ip^{\scr}$. Thus, the lowest-cost way of inducing $\scr$ is given by $\ip^{\scr}$.

Our next lemma uses the connection between SCRs and their revealed information policies to establish several useful properties of the agent's indirect cost function $\cost$.

\begin{lemma}\label{lemma: Properties of Indirect Cost}
The indirect cost $\cost$ is proper, convex, and weak* lower semicontinuous.
\end{lemma}

That $\cost$ is proper follows directly from existence of a finite-cost information policy. To prove lower semicontinuity, we establish that weak* convergence of SCRs implies convergence of their revealed information policies. To show $\cost$ is convex, we establish that the information policy revealed by a convex combination of two SCRs is less informative than the convex combination of the two policies. Using this fact, one can show $\cost$ is convex using convexity and monotonicity of $\icost$.

We now rephrase the agent's problem so that it takes the agent's SCR $\scr$ as its decision variable. Given $\ut$, the benefit from choosing $\scr$ is given by\footnote{Observe the expectation $\bbE[ \ut \cdot \scr ]$ is well-defined because $\ut_{a}$  and $\scr_{a}$ belong to the dual spaces $\lp{1}$ and $\lp{\infty}$, respectively.} 
\[
\bbE [ \ut \cdot \scr ] = \bbE \left[\sum_{a\in A} \ut_{a} \scr_{a}\right].
\]
Because $\cost(\scr)$ gives the minimum cost at which the agent obtains $\scr$, we get that $\scr$ is $\ut$-rationalizable if and only if it solves the program
\begin{equation}\label{SCRProg}
\max_{\scr \in \SCR} \ \big[ \bbE [\ut \cdot \scr] - \cost(\scr)\big].
\end{equation}
Observe the above program always admits a solution: 
because $\cost$ is weak* lower semicontinuous, the program maximizes a weak*-upper-semicontinuous objective over a weak*-compact set.\footnote{Compactness follows from the Banach-Aloaglu theorem \citep[see, e.g., Theorem 6.21 in][]{aliprantis2006infinite}.} In other words, every utility $\ut$ rationalizes some SCR. In the next section, we ask the converse question: which SCRs can be rationalized by some utility function? We show the answer is almost all of them.

\section{Knowing Costs Only: An Impossibility Result}\label{sec: knowing costs only}

In this section, we study the analyst's ability to predict the agent's behavior without knowing anything about her utility. Thus, we take the agent's cost function as given and ask which SCRs can be optimal for some utility function, that is, which SCRs are rationalizable. 

Clearly, the agent's cost function places some restrictions on the set of rationalizable SCRs. For instance, $\scr$ can only be rationalizable if it is feasible; that is, only if it is in the set of SCRs that are attainable at a finite cost,
\[
\Dom := \cost^{-1}(\mathbb{R}).
\]
The agent's cost function can also rule out some SCRs whose costs are finite. For example, whenever $\Omega$ is finite and $\icost$ equals mutual information (Example~\ref{ex: Mutual Info Costs}), only SCRs that use all actions in their support in all states are rationalizable; that is, $\scr$ is rationalizable only if $\scr_{a}$ is strictly positive almost surely whenever it is not identical to zero \citep{caplin2019rational}. Note, however, that this restriction has little practical relevance, because every SCR violating the restriction can be approximated arbitrarily well by SCRs that satisfy it. 

The above discussion raises the following question: Does $\icost$ impose any meaningful restrictions on the set of rationalizable SCRs beyond feasibility? The following result shows the answer is negative. To state the result, call a set of SCRs $\tilde{\SCR}$ \textbf{uniformly dense} in $\Dom$ if every $\scr\in\Dom$ and $\epsilon>0$ admit some $\scrb$ in $\tilde{\SCR}$ such that each $a\in A$ and $\prior$-almost every $\omega\in\Omega$ have $|\scr_a(\omega)-\scrb_a(\omega)|\leq\epsilon$.

\begin{proposition}\label{prop: dense set of rationalizable SCRs}\label{imposs:dense} The set of rationalizable SCRs is uniformly dense in $\Dom$.\footnote{If $\Omega$ is finite, every SCR in the relative interior of $\Dom$ is rationalizable.\label{imposs:finite}}
\end{proposition}

For intuition, recasting the agent's maximization program in geometric terms is useful. To that end, recall $\scr$ is $\ut$-rationalizable if and only if 
\[
\bbE[\ut \cdot \scr] - \cost(\scr) \geq \bbE[\ut \cdot \scrb] - \cost(\scrb) \text{ for all }\scrb\in \SCR,
\]
which is clearly equivalent to 
\[
\cost(\scrb) \geq \cost(\scr) + \bbE[\ut \cdot (\scrb-\scr)]\text{ for all }\scrb\in \SCR.
\]
The above condition has a geometric interpretation: $\scr$ being $\ut$-rationalizable is equivalent to $\ut$ being a subgradient of $\cost$ at $\scr$. It follows $\scr$ is rationalizable if and only if the set of all such subgradients
\[
\partial\cost(\scr) := \bigg\{\ut \in \UT:\ \cost(\scrb) \geq \cost(\scr) + \bbE[\ut \cdot (\scrb-\scr)]\text{ for all }\scrb\in \SCR \bigg\},
\]
also known as $\cost$'s \textbf{subdifferential} at $\scr$, is nonempty. The proposition then follows from observing that $\cost$ is the convex conjugate of a continuous convex function, 
and so one can apply a dual version of the Br{\o}ndsted-Rockafellar theorem \citep{brondsted1965subdifferentiability}. 

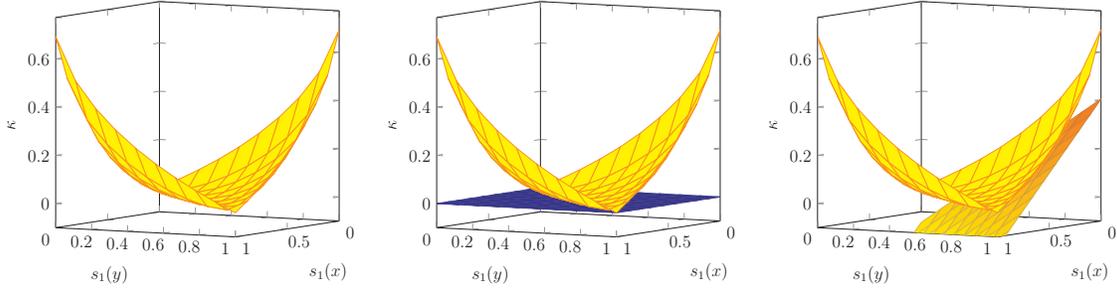
\begin{figure}[tbp]\centering
  \begin{subfigure}{0.32\textwidth}\centering
    \scalebox{0.55}{
    \begin{tikzpicture}
        \begin{axis}[
        samples=\samplepoints,
        domain=0:1,y domain=0:1, zmin=-0.1, view = {120}{5}, colormap/hot,
        xlabel={$s_{1}(x)$},
        ylabel={$s_{1}(y)$},
        zlabel={$\cost$},
        ]

        \addplot3[surf, fill = yellow, faceted color= orange] {0.5*(
        x*ln(x/(0.5*(x+y))) + (1-x)*ln((1-x)/(1-0.5*(x+y))) +
        y*ln(y/(0.5*(x+y))) + (1-y)*ln((1-y)/(1-0.5*(x+y)))) };
        \end{axis}
    \end{tikzpicture}
    }
    \label{fig: 3D Graph MI Costs}
  \end{subfigure}
  \begin{subfigure}{0.32\textwidth}\centering
    \scalebox{0.55}{
    \begin{tikzpicture}
        \begin{axis}[
        samples=\samplepoints,
        domain=0:1,y domain=0:1, zmin=-0.1, view = {120}{5}, colormap/hot,
        xlabel={$s_{1}(x)$},
        ylabel={$s_{1}(y)$},
        zlabel={$\cost$},
        ]
	    \addplot3[surf] {0};
 
        \addplot3[surf, fill = yellow, faceted color= orange] {0.5*(
        x*ln(x/(0.5*(x+y))) + (1-x)*ln((1-x)/(1-0.5*(x+y))) +
        y*ln(y/(0.5*(x+y))) + (1-y)*ln((1-y)/(1-0.5*(x+y)))) };
        \end{axis}
    \end{tikzpicture}
    }
    \label{fig: 3D Graph MI Costs with zero subgradient}
  \end{subfigure}
  \begin{subfigure}{0.32\textwidth}\centering
    \scalebox{0.55}{
    \begin{tikzpicture}
        \begin{axis}[
        samples=\samplepoints,
        domain=0:1,y domain=0:1, zmin=-0.1, view = {120}{5}, colormap/hot,
        xlabel={$s_{1}(x)$},
        ylabel={$s_{1}(y)$},
        zlabel={$\cost$},
        ]

	        \addplot3[surf] {0.5*(0.5*ln(0.25/0.5) + 1.5*ln(0.75/0.5)
            + (x-0.25)*ln(1/3) + (y-0.75)*ln(3)) };
            \addplot3[surf, fill = yellow, faceted color= orange] {0.5*(
            x*ln(x/(0.5*(x+y))) + (1-x)*ln((1-x)/(1-0.5*(x+y))) +
            y*ln(y/(0.5*(x+y))) + (1-y)*ln((1-y)/(1-0.5*(x+y)))) };
        \end{axis}
    \end{tikzpicture}
    }
    \label{fig: 3D Graph MI Costs with subgradient}
  \end{subfigure}
  \caption{Graphs of the indirect cost function $\cost$ and some of its subgradients for the case in which $\Omega = \{x,y\}$, $A = \{0,1\}$, and $\icost$ is given by mutual information (see Example~\ref{ex: Mutual Info Costs}). The cost function is drawn as a function of $\scr_{1}(x)$ and $\scr_{1}(y)$. }
  \label{fig: 3D Graphs}
\end{figure}

The relationship between subgradients and optimality is a standard fact from convex analysis \citep[see, e.g., Theorem 23.5 in][]{rockafellar1970convex}.\footnote{This relationship is analogously used to connect additive random utility models \citep[e.g.,][]{mcfadden1974measurement} with models of additive perturbed utility \citep[e.g.,][]{fudenberg2015stochastic}. In particular, \cite{hofbauer2002global} and \cite{norets2013surjectivity} use convex analysis to show every additive random utility model admits an additive perturbed utility representation.} Geometrically, $\ut$ is a subgradient of $\cost$ at $\scr$ if the function $\scrb\mapsto \cost(\scr) +\bbE[\ut \cdot (\scrb-\scr)]$ defines a hyperplane whose graph supports the epigraph of $\cost$ at the point $(\scr,\cost(\scr))$. Figure~\ref{fig: 3D Graphs} visualizes this condition when $\Omega=\{x,y\}$, $A=\{0,1\}$, and $\icost$ is given by mutual information. In this case, each $\scr$ can be summarized by the probability it takes action $1$ in each state, $(\scr_{1}(x),\scr_{1}(y))$. The left panel illustrates $\cost$ as a function of these probabilities. When $\scr_{1}(x)=\scr_{1}(y)$, the action generated by $\scr$ is uninformative about $\Omega$, and so the agent's learning costs are minimized. Such SCRs are rationalizable by any utility that is constant across both actions. Figure~\ref{fig: 3D Graphs}'s middle panel depicts the supporting hyperplane that results from such utility functions. As can be seen, the resulting hyperplane lies everywhere below the graph of $\cost$, touching the graph only at the SCRs that are rationalizable via such utilities. For non-constant utilities, the agent usually finds it optimal to collect some information about the state. The figure's right panel depicts a supporting hyperplane that corresponds to such a situation.

\cite{caplin2021rationally} prove a specialization of Proposition~\ref{prop: dense set of rationalizable SCRs} for the case in which $\icost$ is posterior separable and when all SCRs in the interior of $\SCR$ have a finite cost. Their argument constructs a utility function that rationalizes an interior $\scr$ by choosing $\ut_{a}$ to be an appropriate member of $\icostd$'s subdifferential at $\posterior_{a}^{\scr}$. This construction suggests a connection between the subdifferentials of $\icostd$ and $\cost$ whenever $\icost$ happens to be posterior separable. In section \ref{sec: discussion}, we discuss this connection in greater detail.

One potential concern with Proposition \ref{prop: dense set of rationalizable SCRs} is that it could be driven by indifference. For an extreme example, suppose learning is free; that is, $\icost(\ip) = 0$ for all $\ip$. In this case, only utility functions that do not depend on the agent's action (meaning $\ut_{a} = \ut_{a'}$ for all $a,a'$) can rationalize a conditionally full-support SCR. Such indifference seems problematic, because it leaves open the possibility of obtaining meaningful restriction on choices via the introduction of an appropriate way to refine optimal behavior.

Next, we state some assumptions under which one can rationalize essentially every SCR without relying on indifference. 

\begin{samepage}
\begin{assumptiona}{A1}\label{ass: regularity} {\color{white}{.}}
\begin{enumerate}[(i)]
\item The cost function $\icost$ is strictly monotone. \label{ass: regularity, part monotone}
\item The cardinality ranking $|\Omega|\geq|A|$ holds.\label{ass: regularity, part cardinality}
\item The domain $\Dom$ has a nonempty (norm) interior in $\SCR$.\label{ass: regularity, part domain}
\end{enumerate}
\end{assumptiona}
\end{samepage}

Part \eqref{ass: regularity, part monotone} means learning more always comes at a strictly positive cost. This part rules out indifference driven by some information being free. Part \eqref{ass: regularity, part cardinality} requires the set of states to be richer than the set of actions. Whenever part \eqref{ass: regularity, part cardinality} fails, the information policy $\ip^{\scr}$ associated with any full-support $\scr$ can be written as the convex combination of two other information policies, each of which is associated with a different SCR. Hence, if $\icost$ is posterior separable (i.e., affine) and Assumption \ref{ass: regularity}\eqref{ass: regularity, part cardinality} fails, $\scr$ cannot possibly be uniquely rationalizable when $\scr$ has full support.\footnote{In the appendix, we show that without Assumption \ref{ass: regularity}\eqref{ass: regularity, part cardinality}, a full support $\scr$ can be uniquely rationalizable only if $\icost$ is not affine around $\ip^{\scr}$.} Assumption \ref{ass: regularity}\eqref{ass: regularity, part domain} plays a similar role to \ref{ass: regularity}\eqref{ass: regularity, part cardinality}. Without Assumption \ref{ass: regularity}\eqref{ass: regularity, part domain}, the cost $\icost$ can restrict the agent to a low-dimensional set of information policies, a restriction with similar implications to violations of Assumption \ref{ass: regularity}\eqref{ass: regularity, part cardinality}.\footnote{Assumption \ref{ass: regularity}\eqref{ass: regularity, part domain} holds, for example, if the agent can obtain a signal that reveals the state with probability $\epsilon$, and provides no information otherwise. It also holds for mutual information costs (Example~\ref{ex: Mutual Info Costs}).}

Our next result shows that, under Assumption \ref{ass: regularity}, the inability to predict the agent's behavior using her information costs alone is not driven by indifference. 

\begin{theorem}\label{thm: dense set of uniquely rationalizable SCRs}
Under Assumption~\ref{ass: regularity}, the set of uniquely rationalizable SCRs is uniformly dense in $\Dom$.\footnote{We also show this set of SCRs admits a dense subset that is open when $\Omega$ is finite, and that the set of utilities that rationalize an $\scr$ in this subset is open (regardless of whether $\Omega$ is finite).}
\end{theorem}

The above result is based on a particular set of SCRs that we call linearly independent. Formally, an $\scr \in \SCR$ is \textbf{linearly independent} if 
$\{\scr_{a}\}_{a\in A} \subseteq \lp{\infty}$ consists of $|A|$ linearly independent elements. In the appendix, we show such SCRs satisfy two useful properties. First, if $\scr$ is linearly independent and $\icost$ is strictly monotone (i.e., Assumption~\ref{ass: regularity}\eqref{ass: regularity, part monotone} holds), then $\scr$ is optimal only if it is uniquely optimal.\footnote{This fact generalizes results by \cite{caplin2013behavioral} and \cite{matejka2015rational} about uniqueness of an optimum under mutual information costs. See the appendix for more details.} Second, the set of linearly independent SCRs is open and dense whenever parts \eqref{ass: regularity, part cardinality} and \eqref{ass: regularity, part domain} of Assumption \ref{ass: regularity} both hold. Since the set of rationalizable SCRs is dense, and the intersection of a dense set with an open and dense set is dense, we obtain that, under Assumption \ref{ass: regularity}, the set of linearly independent SCRs that is rationalizable is dense, and that every such SCR is, in fact, uniquely rationalizable. 

In addition to showing learning costs alone impose no substantial restrictions on behavior, the above analysis also highlights the futility of searching for cost functions that require the agent's behavior to satisfy certain desiderata. To illustrate, consider the following property, suggested by \cite{morris2021coordination}: a cost function $\icost$ satisfies \textbf{continuous choice} if only continuous SCRs are rationalizable.\footnote{An SCR $\scr$ is \textbf{continuous} if it admits a continuous version; that is, every $a$ admits a continuous function $\func_{a}:\Omega \rightarrow [0,1]$ such that $\func_{a} = \scr_{a}$ almost surely.} \cite{morris2021coordination} suggest continuous choice as a way to model agents who have a hard time distinguishing between similar states. Below, we show the only way to guarantee this property is to make discontinuous information structures infeasible. To show this result, we observe that continuity of an SCR is a closed property, and so can be satisfied by a dense subset of $\Dom$ only if it is satisfied by all of $\Dom$'s elements.

\begin{samepage}
\begin{proposition} \label{prop: continuous choice is impossible} Every rationalizable SCR is continuous if and only if every SCR in $\Dom$ is continuous.
\end{proposition}
\end{samepage}

Thus, our results suggest that except in cases where continuous choice is vacuous, it can only hold if discontinuous SCRs are infeasible, a property referred to by \cite{morris2021coordination} as \emph{infeasible perfect discrimination}.\footnote{We should highlight that we are abusing terminology slightly: \cite{morris2021coordination} require $\Omega$ to be a real interval, and say that a cost function satisfies infeasible perfect discrimination if it assigns finite cost only to SCRs that are \emph{absolutely} continuous.}
In terms of $\icost$, this property is equivalent to $\icost$ assigning infinite cost to every simple information policy that generates a $\posterior$ with a discontinuous Radon-Nikodym derivative with respect to $\prior$. In section \ref{sec: discussion}, we note one can circumvent the need for infinite costs by requiring continuous choice to hold only for a restricted set of objectives. For example, one may require the agent's choice to be continuous for any bounded objective. As we explain later, such a requirement is equivalent to a slight weakening of \citeapos{morris2021coordination} \emph{expensive perfect discrimination} property.

\section{Cross-Menu Predictions}\label{sec: cross-menu}
The previous section shows learning costs alone impose no meaningful restrictions on the agent's behavior given a fixed menu. In this section, we highlight these costs do constrain how the agent behaves \emph{across} menus. The reason is that the agent's choices from one menu allow the analyst to use the agent's cost function to extract information about her preferences. In fact, we show the information revealed by the agent's behavior is particularly sharp when her cost function is smooth in a manner we make precise below. This sharpness allows us to prove the following result: whenever our smoothness conditions holds, most choices in one menu pin down the agent's actions in all submenus.

\subsection{Smooth Costs}

We introduce our smoothness condition in stages. We begin by defining a differentiability notion for $\icost$. To that end, we recall some standard notation. Given a convex set $X$ in a real vector space, a convex function $\func:X\rightarrow \extreal$, and $x\in\func^{-1}(\real)$, define the \textbf{directional derivative} of $\func$ at $x$ as $d^+_{x}\func:X\to\real\cup\{\pm\infty\}$ via 
\[
d^+_{x}\func(x'):= \lim_{\epsilon \searrow 0}\tfrac{1}{\epsilon}\left[\func(x + \epsilon (x' - x)) - \func(x) \right].
\]
Say $\icostd \in \icostD$ is \textbf{a derivative} of $\icost$ at $\ip\in\Domicost$ if $\int \icostd (\posterior) \ \ip(\dd \posterior)$ is finite and if every $\ip' \in \Domicost$ has
\[ 
d^{+}_{\ip}\icost(\ip')= \int \icostd(\posterior)\ (\ip' - \ip)(\dd\posterior).
\]
The cost function $\icost$ is \textbf{differentiable at} $\ip\in\Domicost$ if it admits a derivative at $\ip$.\footnote{Note that we require the derivative to be convex (because $\icostd\in\icostD$). In the appendix, we show one can omit this requirement whenever $\icostd$ is finite and continuous.} We omit the dependency on the information policy and say $\icost$ is \textbf{differentiable} whenever it is differentiable at all $\ip \in \Domicost$.\footnote{
Our notion of differentiability is commonly used in the decision theory literature, where it is often called \emph{G\^{a}teaux differentiability} \citep[e.g.,][]{hong1987risk,cerreia2017stochastic}. The definition is slightly different than the way G\^{a}teaux differentiability is defined in convex analysis: whereas the latter requires the convergence to occur from all possible directions \citep[e.g.,][]{phelps2009convex,borwein2010convex}, we require convergence only from directions within the domain of $\icost$.
} 

Intuitively, a cost function is differentiable if it prices small shifts in its information in a posterior-separable manner. Indeed, observe that the function 
\[
\icost_{\icostd} (\ipb) = \int \icostd(\posterior)\ \ipb(\dd\posterior)
\]
defines a posterior-separable cost function.
In fact, this cost equals $\icost$ whenever the latter is already posterior separable: in this case, $\icost$ has a common derivative at all information policies. 

In general, the derivative of a differentiable cost function that is not posterior separable depends on $\ip$. For a demonstration, consider the cost function in Example~\ref{ex: Transformed Costs}, $\icost(\ip) = \transf\left(\int \icostd(\posterior)\ \ip(\dd \posterior) \right).$ Whenever $\transf$ is differentiable, this cost function admits 
\[
\transf'\left(\int \icostd(\posterior)\ \ip(\dd \posterior) \right)\icostd(\cdot)
\]
as its derivative at $\ip$, which depends on $\ip$ whenever $\transf$ is not affine. For another example, consider the cost function from Example~\ref{ex: Quadratic Costs}. This cost function is differentiable, with a derivative at $\ip \in \Domicost$ given by 
\(
2\int \tilde{\icostd}(\cdot,\posterior) \ \ip(\dd\posterior). 
\)
Clearly, this derivative typically changes as $\ip$ varies.\footnote{For a non-differentiable cost function, consider Example~\ref{ex: Maximum Costs} when $\icostdset$ is finite and not equivalent to a singleton. In this case, one obtains a point of non-differentiability at any $\ip$ on the boundary between two regions where the set of maximizers differs.} 

Our next result tightens the connection between posterior separability and differentiability. To state this result, define the indirect cost function associated with $\icost_{\icostd}$,
\[
\cost_{\icostd} (\scrb) := \icost_{\icostd}(\ip^{\scrb}).
\]
Lemma~\ref{lem: posterior separable approximation characterizes optimality} below shows that, whenever $\icost$ has full domain, $\scr$ is $\ut$-rationalizable if and only if it is $\ut$-rationalizable when costs are given by the cost function's posterior separable approximation at $\ip^{\scr}$. 

\begin{lemma}\label{lem: posterior separable approximation characterizes optimality}
Fix some $\scr \in \SCR$ and $\ut \in \UT$. Suppose $\icost$ is finite on $\Infof$
and that $\icostd$ is a derivative of $\icost$ at $\ip^{\scr}$.\footnote{
One can weaken the requirement that $\icost$ is finite on $\Infof$, assuming only that $\icost$ is finite for simple information policies around $\ip^{\scr}$. See the appendix for more details. Proposition~\ref{prop: multiplier result for iteratively differentiable case} below can be similarly strengthened. 
} Then,
\[
\scr \in \argmax_{\scrb \in \SCR} \left[\bbExp{\ut\cdot\scrb} -\cost(\scrb) \right]
\]
if and only if
\[
\scr \in \argmax_{\scrb \in \SCR} \left[\bbExp{\ut\cdot\scrb} -\cost_{\icostd}(\scrb) \right].
\]
\end{lemma}
Note the above result does not tell us $\cost$ and $\cost_{\icostd}$ have the same set of maximizers. The reason is that the cost function's derivative depends on the SCR around which the cost is approximated. In other words, $\ut$-rationalizability of $\scr$ is equivalent to $\scr$ being $\ut$-rationalizable under $\icost_{\icostd}$ only if $\icostd$ is a derivative of $\icost$ at $\ip^\scr$. 
Unless $\icost$ is posterior separable (in which case all SCRs admit a common derivative), different SCRs usually admit different derivatives.

Next, we define a differentiability notion for derivatives of $\icost$. 
Given $\icostd\in\icostD$ and a simply drawn $\posterior$, we say  $\icostdd_{\posterior}\in \lp{1}$ is a \textbf{derivative} of $\icostd$ at $\posterior$ if $\int \icostdd_{\posterior}(\omega) \ \posterior(\dd\omega) =\icostd(\posterior)$, and every simply drawn $\posterior' \in \DO$ has
\begin{equation*}
d_{\posterior}^{+}\icostd(\posterior') = \int \icostdd_{\posterior}(\omega)\ (\posterior' - \posterior)(\dd\omega).
\end{equation*}
The function $\icostd$ is \textbf{differentiable} at $\posterior$ if it admits a derivative there. Thus, $\icostd$ is differentiable if it can be locally approximated by an affine function.  

Our key notion of smoothness requires $\icost$ to admit a differentiable derivative. Specifically, we say $\icost$ is \textbf{iteratively differentiable} at $\ip\in\Infof$ if it admits a derivative $\icostd$ at $\ip$ that is differentiable at every $\posterior \in \supp \ \ip$. 
Thus, an iteratively differentiable cost is locally similar to a smooth, posterior-separable cost function. 

The benefit of having iteratively differentiable costs is summarized in the following proposition, which tightly characterizes when an interior SCR is optimal.

\begin{samepage}
\begin{proposition}\label{prop: multiplier result for iteratively differentiable case}
Suppose $\scr$ has full support, and $\icost$ is finite on $\Infof$ and iteratively differentiable at $\ip^{\scr}$ with derivative $\icostd$. The following are equivalent for $\ut\in\UT$:
\begin{enumerate}[(i)]
\item SCR $\scr$ is $\ut$-rationalizable; that is, $\scr \in \argmax_{\scrb \in \SCR} \left[\bbExp{\ut\cdot\scrb} - \cost(\scrb)\right]$.
\item\label{IterDiffFOC} Some $\nuis\in \lp{1}$ and $\mult\in\lp{1}_+^A$ exist such that every $a\in A$ have $$\ut_a = \nuis-\mult_a +\icostdd_{\posterior^{\scr}_{a}} \text{ and } \mult_a\scr_a=\mathbf{0}.$$ 
\end{enumerate}
\end{proposition}
\end{samepage}
One can view the proposition as establishing a Lagrange multiplier result for iteratively differentiable costs, with $\icostdd_{\posterior_{a}^{\scr}}$ serving the role of the derivative of $\cost$. To better see this interpretation, suppose $\Omega$ is finite and that $\cost$ is differentiable. In this case, $\scr$ being $\ut$-rationalizable is equivalent to $\scr$ solving the program
\begin{equation*}
\begin{aligned}
        \max_{\scrb \in \mathbb{R}^{A\times\Omega}} \quad & \big[\bbExp{\ut \cdot \scrb} - \cost(\scrb)\big] \\
        \textrm{s.t.} \quad & \scrb_{a}(\omega) \geq 0 & &  \forall a\in A, \ \forall\omega\in \Omega,\\
        & \sum_{a} \scrb_{a}(\omega) =1 & & \forall \omega \in \Omega. 
\end{aligned}    
\end{equation*}
Applying a standard Lagrangian result \citep[e.g.,][]{pourciau1983multiplier} gives a multiplier $\mult_{a}(\omega)$ for every instance of the first constraint and a multiplier $\nuis(\omega)$ for every instance of the second constraint such that $\scr$ is optimal if and only if 
\[
\ut_{a}(\omega) - \frac{1}{\prior(\omega)}\frac{\partial \cost}{\partial \scr_{a}(\omega)} + \mult_{a}(\omega) - \nuis(\omega)=0.
\]
Observe the above display equation specializes to condition \eqref{IterDiffFOC} of the proposition, but with $\frac{1}{\prior(\omega)}\frac{\partial \cost}{\partial \scr_{a}(\omega)}$ replacing $\icostdd_{\posterior_{a}^{\scr}}$.  

In addition to characterizing optimality of a full-support $\scr$, Proposition \ref{prop: multiplier result for iteratively differentiable case} also enables one to recover the agent's utility function from their behavior, up to a nuisance term. For an explanation, suppose $\scr$ has conditionally full support and that $\icost$ is iteratively differentiable at $\ip^{\scr}$ with derivative $\icostd$. Let $\ut^{\scr}$ be the utility function defined via $\ut_{a}^{\scr}:=\icostdd_{\posterior_{a}^{\scr}}$ for all $a$. Then, Proposition \ref{prop: multiplier result for iteratively differentiable case} implies a utility $\ut$ rationalizes $\scr$ if and only if $\ut$ generates the same optima for all menus. The reason is that $\ut$ equals $\ut^{\scr}$ plus an action-independent shift, $\nuis$, the addition of which has no impact on the set of maximizers.\footnote{Note $\mult_a\scr_a=\mathbf0$ means $\mult_a=\mathbf0$ when $\scr$ has conditionally full support.} The next section uses this observation to identify conditions under which the agent's choices in one menu are sufficient for pinning down her behavior for all submenus.

\subsection{Unique Subset Predictions}

We now use the results of the previous section to show the agent's learning costs restrict her behavior across menus. To model these restrictions, let $\Menu$ be all nonempty subsets of $A$. For any $\menu \in \Menu$, let $\SCR_\menu = \{\scr\in \SCR:\ \sum_{a\in\menu}\scr_a = \mathbf1 \}$ be the set of SCRs that only use actions in $\menu$. Given a utility function $\ut$, we say $\scr\in \SCR_{\menu}$ is $\ut$\textbf{-rationalizable over} $\menu$ if $\scr$ solves the agent's problem when the agent is restricted to only using actions in $\menu$; that is,
\[
\scr \in \argmax_{\scrb\in \SCR_{\menu}} \left[\bbExp{\ut \cdot \scrb} - \cost(\scrb) \right].
\]
Say $\scr \in \SCR$ \textbf{yields unique subset predictions} if every $\menu \in \Menu$ has a unique $\scrb \in \SCR_{\menu}$ that admits a utility that both rationalizes $\scrb$ over $\menu$ and rationalizes $\scr$ (over $A$).  Thus, if the agent chooses from $A$ according to $\scr$, and $\scr$ yields unique subset predictions, we can deduce exactly what must be the agent's chosen SCR when she's restricted to $\menu$.

Next, we introduce assumptions that enable us to conclude many SCRs yield unique subset predictions. To state the assumption, say an information policy is \textbf{fully mixed} if all posteriors in its support have a strictly positive Radon-Nikodym derivative with respect to $\prior$.

\begin{samepage}
\begin{assumptiona}{A2}\label{ass: smoothness} \textcolor{white}{.}
\begin{enumerate}[(i)]
\item $\icost$ is finite at every simple information policy.\footnote{This condition implies $\Dom=\SCR$, and, in particular, implies \ref{ass: regularity}\eqref{ass: regularity, part domain}.} \label{ass: smoothness, part finite}

\item $\icost$ is iteratively differentiable at every simple and fully mixed information policy. \label{ass: smoothness, part differentiability}

\item If $\ip \in \Infof$ is not fully mixed,  
$d_{\ip}^{+}\icost(\delta_{\prior}) = -\infty$. 
\label{ass: smoothness, part inada}
\end{enumerate}
\end{assumptiona}
\end{samepage}

The first two parts of the assumption require $\icost$ to be smooth, as explained in the previous section. The third part requires $\icost$ to have infinite slopes at simple information policies that are not fully mixed. This property, which is satisfied by mutual-information costs (Example \ref{ex: Mutual Info Costs}), implies an SCR can be rationalizable only if it uses all actions in its support in all states. In particular, a full-support SCR is rationalizable only if it has conditionally full support. Therefore, Assumption~\ref{ass: smoothness} implies one can use Proposition~\ref{prop: multiplier result for iteratively differentiable case} to recover the agent's utility function, up to a choice-irrelevant nuisance term, whenever she employs a full-support SCR. As Lemma~\ref{lem: unique prediction with known benefits} below shows, the agent's utility function enables the analyst to precisely predict the agent's conditional action distribution, except for a knige-edge set of utilities.

\begin{lemma}\label{lem: unique prediction with known benefits}
The set of utilities with multiple rationalizable SCRs is meager and shy.\footnote{A set is \textbf{meager} if it is a  countable union of nowhere dense sets. A set $\UTB \subseteq \UT$ is \textbf{shy} if some probability measure $\nu \in \Delta \UT$ with compact support assigns zero measure to every translation of $\UTB$, that is, if $\nu (\UTB+\ut)=0$ for all $\ut \in \UT$. Shy sets generalize Lebesgue null sets beyond the case of finite dimensions \citep[see][]{hunt1992prevalence}. Thus, the proposition implies that for finite $\Omega$, the set of utilities that admit multiple rationalizable SCRs is Lebesgue null.}
\end{lemma}

Lemma \ref{lem: unique prediction with known benefits} essentially follows from the fact that most hyperplanes that support the graph of a well-behaved convex function admit a unique support point. Figure \ref{fig: 1D convex function} illustrates a sense in which this fact holds for a convex function $\phi$ defined over a closed subinterval of $\mathbb{R}$. In one dimension, a subgradient is defined via the slope of the corresponding line. For a line to support $\phi$ in multiple points, the line's slope must equal the slope of $\phi$ in an interval over which $\phi$ is affine. Clearly, at most countably many such intervals can exist, and so the set of gradients that admit multiple support points must be countable as well. Being countable, the set is also meager and Lebesgue null, two properties that generalize to higher dimensions.

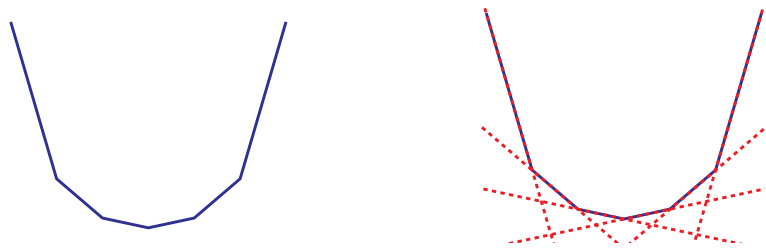
\begin{figure}[tbp]\centering
  \begin{subfigure}{0.32\textwidth}\centering
    \scalebox{0.55}{
    \begin{tikzpicture}
			\begin{axis} 
				[
					samples= 10, 
					ticks=none,
					ymax=0.9,
					ymin=-0.1,
					xmin=-0.02,
					xmax=1.01,
					axis on top=false,
					axis x line = none,
					axis y line = none,
					every inner x axis line/.append style={|-|},
				    xtick={0,1},
				    xticklabels={$0$,$1$},
				]

\addplot [blue, line width=2] coordinates {(0,0.875) (0.1667,0.2083) (0.3333,0.04167) (0.5,0) (0.6667,0.04167) (0.8333,0.2083) (1,0.875)};
				
			\end{axis}
\end{tikzpicture}
    }
    \label{fig: 1D convex function no subgradients}
  \end{subfigure}
  \begin{subfigure}{0.48\textwidth}\centering
    \scalebox{0.55}{
    \begin{tikzpicture}
			\begin{axis} 
				[
					samples=10, 
					ticks=none,
					ymax=0.9,
					ymin=-0.1,
					xmin=-0.02,
					xmax=1.01,
					axis on top=false,
                    axis x line = none,
					axis y line = none,
					every inner x axis line/.append style={|-|},
				    xtick={0,1},
				    xticklabels={$0$,$1$},
				]

\addplot [blue, line width=2] coordinates {(0,0.875) (0.1667,0.2083) (0.3333,0.04167) (0.5,0) (0.6667,0.04167) (0.8333,0.2083) (1,0.875)};
				
				\addplot[red, line width = 2, dashed]{0.875 + (-4)*(x - 0)};
				\addplot[red, line width = 2, dashed]{0.2073 + (-1)*(x - 0.1667)};
				\addplot[red, line width = 2, dashed]{0.04167 + (-0.25)*(x - 0.3333)};
				
				\addplot[red, line width = 2, dashed]{0.875 + (4)*(x - 1)};
				\addplot[red, line width = 2, dashed]{0.2073 + (1)*(x - 0.8333)};
				\addplot[red, line width = 2, dashed]{0.04167 + (0.25)*(x - 0.6667)};
				
			\end{axis}
\end{tikzpicture}
    }    
    \label{fig: 1D convex function w subgradients}
  \end{subfigure}
  \caption{A convex function over an interval with the set of its subgradients that attain multiple support points.}
  \label{fig: 1D convex function}
\end{figure}

Using the above facts, our next result demonstrates that the agent's learning costs often impose strong restrictions on the agent's behavior across menus. 

\newpage

\begin{theorem}\label{thm: cross-choice predictions}
Under \ref{ass: regularity} and \ref{ass: smoothness}, SCRs yielding unique subset predictions are weak* dense.
\end{theorem}

To prove the theorem, we first use Lemma~\ref{lem: unique prediction with known benefits} to find a dense set $\UT_{\Menu}$ of utilities that attain a unique optimum at every menu. The proof then proceeds by identifying a dense subset of SCRs that are rationalizable by utilities in $\UT_{\Menu}$. To identify this subset, we observe Theorem~\ref{thm: dense set of uniquely rationalizable SCRs} and Assumption~\ref{ass: smoothness}\eqref{ass: smoothness, part inada} imply the set of uniquely rationalizable SCRs with conditionally full support is dense in $\Dom$. Using a continuity property of the subdifferential, we then show one can approximate any $\scr$ in this set with SCRs that are rationalizable by some utility in $\UT_{\Menu}$. Moreover, because $\scr$ has conditional full support, the approximating SCRs can be taken to have the same property.\footnote{This property of the approximating SCRs relies on Assumption~\ref{ass: smoothness}\eqref{ass: smoothness, part inada}. We note this assumption is not needed when $\Omega$ is finite, because in this case, the set of conditionally full-support SCRs is open.}  
Focusing on one of the approximating SCRs, Proposition~\ref{prop: multiplier result for iteratively differentiable case} implies all utilities that rationalize it generate the same optima for all menus. Because one of these rationalizing utilities comes from $\UT_{\Menu}$, it follows that, for any given menu, all of these utilities uniquely rationalize the same SCR. In other words, each of the SCRs approximating $\scr$ yields unique subset predictions. Since $\scr$ is an arbitrary element of a dense subset of $\Dom$, the result follows.

\section{Cross-Menu Tests}\label{sec: cycle test}

As pointed out in the previous section, the agent's information acquisition costs impose significant restrictions on her behavior across menus. In this section, we use these restrictions to develop a test for whether the agent's behavior is consistent with a given cost function. Throughout this section, we focus on cost functions that satisfy Assumption~\ref{ass: smoothness}\eqref{ass: smoothness, part finite} and \ref{ass: smoothness}\eqref{ass: smoothness, part differentiability}, which are satisfied by virtually all models used in the literature. 

We consider an analyst who can estimate the agent's chosen SCRs in a subset of all possible menus. Specifically, we define a \textbf{data set} to be a pair, $\data=(\DMenu,\dscr)$, where the set $\DMenu \subseteq \Menu$ contains the menus on which the analyst can estimate the agent's behavior, and $\dscr:\DMenu \rightarrow \SCR$ maps each such menu to the SCR describing the agent's behavior in that menu, and so has $\dscr^\menu:=\dscr(\menu)\in\SCR_\menu$ for each $\menu\in\DMenu$. Such data sets have been produced in the literature, see \cite{dean2019experimental}, and \cite{caplin2020rational}, for example. Our main goal is to develop necessary and sufficient tests for a data set to be compatible with a given cost function and a common objective. Specifically, we are interested in testing whether a data set $\data$ is \textbf{consistent}, meaning a utility function $\ut \in \UT$ exists such that, for every $\menu \in \DMenu$, the SCR $\dscr^\menu$ is $\ut$-rationalizable over $\menu$. We find it convenient to concentrate on data sets in which the agent takes all feasible actions with positive probability conditional on the state. These are data sets that have \textbf{conditionally full support}, namely $\dscr^{\menu}_{a}(\omega)>0$ holds almost surely for all $a \in \menu\in \DMenu$. We discuss data sets without this property at the end of this section. 

Our tests are based on a particular class of cycles generated by the data. To describe these cycles, note that every data set induces a bipartite graph $\dgraph_{\data}$ with vertex classes $\DMenu$ and $A$, and for which an edge exists between $\menu \in \DMenu$ and $a\in A$ if and only if $a \in \supp \ \dscr^{\menu}$. A \textbf{testable cycle} is a cycle in this graph that begins (and therefore ends) with an action. In other words, a testable cycle is a sequence $a_0 \menu_1 a_1 \menu_2 a_2 \ldots \menu_N a_N$ such that $a_0 = a_N$, and $\{a_{n-1},a_n\}\subseteq \supp \ \dscr^{\menu_n}$ for all $n\in\{1,\ldots,N\}$. 

To better understand the concept of testable cycles, consider the following example. The agent has three actions $A=\{x,y,z\}$, the utility of which depends on a state that is uniformly distributed over $\{1,2,3\}$. The analyst can estimate the agent's chosen SCRs on three menus, $\DMenu=\left\{\{x,y,z\}, \{x,y\}, \{y,z\}\right\}$, with the agent's behavior in these menus is given by
\begin{equation}\label{eq: example dataset}
\begin{split}
\left(\dscr^{\{x,y,z\}}_{x},\dscr^{\{x,y,z\}}_{y},\dscr^{\{x,y,z\}}_{z}\right)(\omega) 
& = \begin{cases}
\left(\frac{1}{5},\frac{2}{5},\frac{2}{5}\right) & \text{if }\omega = 1, \\ 
\left(\frac{2}{5},\frac{1}{5},\frac{2}{5}\right) & \text{if }\omega = 2, \\ 
\left(\frac{2}{5},\frac{2}{5},\frac{1}{5}\right) & \text{if }\omega = 3, \\ 
\end{cases} 
\\
\left(\dscr^{\{x,y\}}_{x},\dscr^{\{x,y\}}_{y},\dscr^{\{x,y\}}_{z}\right)(\omega) 
& = \begin{cases}
\left(\frac{1}{3},\frac{2}{3},0\right) & \text{if }\omega = 1, \\ 
\left(\frac{2}{3},\frac{1}{3},0\right) & \text{if }\omega = 2, \\ 
\left(\frac{1}{2},\frac{1}{2},0\right) & \text{if }\omega = 3, \\ 
\end{cases}
\\ 
\left(\dscr^{\{y,z\}}_{x},\dscr^{\{y,z\}}_{y},\dscr^{\{y,z\}}_{z}\right)(\omega) 
& = \begin{cases}
\left(0,\frac{1}{2},\frac{1}{2}\right) & \text{if }\omega = 1, \\ 
\left(0,\frac{1}{3},\frac{2}{3}\right) & \text{if }\omega = 2, \\ 
\left(0,\frac{2}{3},\frac{1}{3}\right) & \text{if }\omega = 3. \\ 
\end{cases}
\end{split}
\end{equation}
In this example, the data set $\dscr$ has conditionally full support. 
Therefore, the bipartite graph induced by $\dscr$ has an edge between an action $a$ and a menu $\menu \in \DMenu$ if and only if $a \in \menu$. In other words, this data set induces the graph drawn in Figure~\ref{fig: bipartite graph example}. This graph contains three simple testable cycles: $x \ \{x,y\} \ y \ \{x,y,z\} \ x$,\ \ \ \ $y \ \{x,y,z\} \ z \ \{y,z\} \ y$, and $x \ \{x,y\} \ y \ \{y,z\} \ z \ \{x,y,z\} \ x$. Every other testable cycle involves concatenations of these three cycles (possibly in reversed order or starting at a different action). 

\begin{figure}[t]
\centering

\begin{subfigure}{0.22\textwidth}\centering
    \scalebox{0.7}{
\begin{tikzpicture}[scale=2.5]
\node [draw,circle,fill=black,scale=0.5, label=left:$x$] at (0,0) (x) {};
\node [draw,circle,fill=black,scale=0.5, label=left:$y$] at (0,-0.5) (y) {};
\node [draw,circle,fill=black,scale=0.5, label=left:$z$] at (0,-1) (z) {};

\node [draw,circle,fill=gray,scale=0.5, label=right:{$\{x,y\}$}] at (1,0) (xy) {};
\node [draw,circle,fill=gray,scale=0.5, label=right:{$\{x,y,z\}$}] at (1,-0.5) (xyz) {};
\node [draw,circle,fill=gray,scale=0.5, label=right:{$\{y,z\}$}] at (1,-1) (yz) {};

\draw (x) -- (xy); 
\draw (x) -- (xyz);
\draw (y) -- (xy);
\draw (y) -- (xyz);
\draw (y) -- (yz);
\draw (z) -- (xyz);
\draw (z) -- (yz);
\end{tikzpicture}
}
\end{subfigure}
%
\begin{subfigure}{0.22\textwidth}\centering
    \scalebox{0.7}{
\begin{tikzpicture}[scale=2.5]
\node [draw,circle,fill=black,scale=0.5, label=left:$x$] at (0,0) (x) {};
\node [draw,circle,fill=black,scale=0.5, label=left:$y$] at (0,-0.5) (y) {};
\node [draw,circle,fill=black,scale=0.5, label=left:$z$] at (0,-1) (z) {};

\node [draw,circle,fill=gray,scale=0.5, label=right:{$\{x,y\}$}] at (1,0) (xy) {};
\node [draw,circle,fill=gray,scale=0.5, label=right:{$\{x,y,z\}$}] at (1,-0.5) (xyz) {};
\node [draw,circle,fill=gray,scale=0.5, label=right:{$\{y,z\}$}] at (1,-1) (yz) {};

\draw [-stealth,line width=1] (x) -- (xy); 
\draw [stealth-,line width=1] (y) -- (xy);
\draw [-stealth,line width=1] (y) -- (xyz);
\draw [stealth-,line width=1] (x) -- (xyz);
\draw [style=dashed] (y) -- (yz);
\draw [style=dashed] (z) -- (xyz);
\draw [style=dashed] (z) -- (yz);
\end{tikzpicture}
}
\end{subfigure}
%
\begin{subfigure}{0.22\textwidth}\centering
    \scalebox{0.7}{
\begin{tikzpicture}[scale=2.5]
\node [draw,circle,fill=black,scale=0.5, label=left:$x$] at (0,0) (x) {};
\node [draw,circle,fill=black,scale=0.5, label=left:$y$] at (0,-0.5) (y) {};
\node [draw,circle,fill=black,scale=0.5, label=left:$z$] at (0,-1) (z) {};

\node [draw,circle,fill=gray,scale=0.5, label=right:{$\{x,y\}$}] at (1,0) (xy) {};
\node [draw,circle,fill=gray,scale=0.5, label=right:{$\{x,y,z\}$}] at (1,-0.5) (xyz) {};
\node [draw,circle,fill=gray,scale=0.5, label=right:{$\{y,z\}$}] at (1,-1) (yz) {};

\draw [style=dashed] (x) -- (xy); 
\draw [style=dashed] (x) -- (xyz);
\draw [style=dashed] (y) -- (xy);
\draw [-stealth, line width=1] (y) -- (xyz);
\draw [stealth-, line width=1] (z) -- (xyz);
\draw [-stealth, line width=1] (z) -- (yz);
\draw [stealth-, line width=1] (y) -- (yz);

\end{tikzpicture}
}
\end{subfigure}
%
\begin{subfigure}{0.22\textwidth}\centering
    \scalebox{0.7}{
\begin{tikzpicture}[scale=2.5]
\node [draw,circle,fill=black,scale=0.5, label=left:$x$] at (0,0) (x) {};
\node [draw,circle,fill=black,scale=0.5, label=left:$y$] at (0,-0.5) (y) {};
\node [draw,circle,fill=black,scale=0.5, label=left:$z$] at (0,-1) (z) {};

\node [draw,circle,fill=gray,scale=0.5, label=right:{$\{x,y\}$}] at (1,0) (xy) {};
\node [draw,circle,fill=gray,scale=0.5, label=right:{$\{x,y,z\}$}] at (1,-0.5) (xyz) {};
\node [draw,circle,fill=gray,scale=0.5, label=right:{$\{y,z\}$}] at (1,-1) (yz) {};

\draw [style=dashed](y) -- (xyz);
\draw [-stealth,line width=1](x) -- (xy); 
\draw [-stealth,line width=1](xy) -- (y);
\draw [-stealth,line width=1](y) -- (yz);
\draw [-stealth,line width=1](yz) -- (z);
\draw [-stealth,line width=1](z) -- (xyz);
\draw [-stealth,line width=1](xyz) -- (x);
\end{tikzpicture}
}
\end{subfigure}

\caption{The left panel depicts the bipartite graph induced by a conditionally full support data set $(\DMenu,\dscr)$ with $\DMenu=\{\{x,y\},\{x,y,z\},\{y,z\}\}$. The right three panels depict the three simple cycles contained in this graph.}
\label{fig: bipartite graph example}
\end{figure}
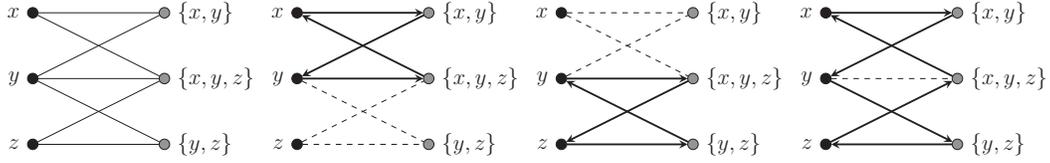

It turns out to be convenient to focus on a special subset of testable cycles called a \emph{cycle basis}. To define such a subset, let $\dgraphE_{\data}$ be the set of all edges in the graph induced by $\data$---that is, the set of all pairs, $\{a,\menu\}$ such that $\menu \in \DMenu$ and $a \in \supp \ \dscr^{\menu}$. Associate each testable cycle $\cycle=a_0 \menu_1 a_1 \menu_2 a_2 \ldots \menu_N a_N$ with the cycle vector $\cyclev^{\cycle}\in \mathbb{R}^{\dgraphE_{\data}}$ given by
\[
\cyclev^{\cycle}_{\{a,\menu\}} = \left|\left\{n: (a,\menu) = (a_{n-1},\menu_n)\right\}\right| - \left|\left\{n: (\menu,a) = (\menu_{n},a_{n})\right\}\right|.
\]
Let $\Cyclev$ be the vector subspace of $\real^{\dgraphE_{\data}}$ spanned by the set of all cycle vectors. A \textbf{cycle basis} is a set of cycles whose corresponding vectors form a basis for $\Cyclev$. Informally, a cycle basis is a minimal set of cycles that is sufficient for generating any other cycle in the graph. For more information on different types of cycle basis and algorithms that generate them, see a standard reference, e.g., \cite{bollobas2012graph}.\footnote{\cite{bollobas2012graph} Chapter~II, $\S$3 defines cycle bases, characterizes their cardinality, and provides an explicit construction of one.} For the current purpose, it suffices to note that several algorithms exist for finding a cycle basis, and that one can easily calculate a basis' cardinality: any cycle basis of the graph induced by a given data set $\data$ contains 
\begin{equation}\label{eq: cycle basis cardinality}
\sum_{\menu \in \DMenu}|\supp \ \dscr^{\menu}| - |A| - |\DMenu| + k
\end{equation}
cycles, where $k$ is the number of connected components of the graph. 

To demonstrate the definition of cycle basis, consider again the data set presented in equation~\eqref{eq: example dataset}. As noted earlier, every testable cycle in this example is essentially a concatenation of some combination of the following three simple cycles: $\cycle_{1} = x \ \{x,y\} \ y \ \{x,y,z\} \ x$,\ \ \ \ $\cycle_2 = y \ \{x,y,z\} \ z \ \{y,z\} \ y$, and $\cycle_3 = x \ \{x,y\} \ y \ \{y,z\} \ z \ \{x,y,z\} \ x$. However, the calculation from equation~\eqref{eq: cycle basis cardinality} says that only $2$ cycles are necessary for obtaining a cycle basis. We now explain that $\{\cycle_1,\cycle_2\}$ is a cycle basis. To do so, observe first that the corresponding set of vectors $\{\cyclev^{\cycle_1},\cyclev^{\cycle_2}\}$ is linearly independent, because each cycle contains an edge that is not included in the other. Second, recall every cycle is built from the three previously-mentioned cycles. Therefore, for every cycle $\cycle$, one can find three integers $k_1,k_2,k_3 \in \mathbb{Z}$ such that $\cyclev^{\cycle} = k_1 \cyclev^{\cycle_1} + k_2 \cyclev^{\cycle_2} + k_3 \cyclev^{\cycle_3}$. And third, observe $\cyclev^{\cycle_3}= \cyclev^{\cycle_1}-\cyclev^{\cycle_2}.$ It follows $\{\cyclev^{\cycle_1},\cyclev^{\cycle_2}\}$ is a linearly independent set of vectors that spans $\Cyclev$---i.e., $\{\cyclev^{\cycle_1},\cyclev^{\cycle_2}\}$  is a cycle basis. 

Testable cycles are important because they trace all the restrictions imposed by the consistency requirement. To introduce these restrictions, we require some additional notation. Fix some data set $\data=(\DMenu,\dscr)$. For each $\menu\in\DMenu$, let $\icostd^\menu$ be some derivative of $\icost$ at $\ip^{\dscr^{\menu}}$, and take $\posterior^{\menu}_{a}=\posterior^{\dscr^\menu}_{a}$ to be posterior revealed by action $a$ given SCR $\dscr^{\menu}$. Then for every $\menu$ and action $a \in \supp \ \dscr^{\menu}$, we can define the function 
\[
\func^{\data}_{a,\menu}:=\icostdd^{\menu}_{\posterior_a^{\menu}}.
\]
Armed with these definitions, we can associate each testable cycle $a_0 \menu_1 a_1 \menu_2 a_2 \ldots \menu_N a_N$ with the following equation\footnote{Although the definition of $(\func^{\data}_{a,\menu})_{a\in B}$ depends on the choice of derivative $\icostd^\menu$ for each $\menu$, this dependence shifts $\func^{\data}_{a,\menu}$ in the same way for every $a\in \menu$. Because each testable cycle has exactly one outgoing edge from $\menu$ for each incoming one, the left-hand side of equation~\eqref{eq: testable cycle equation} turns out not to depend on this choice of derivative.}
\begin{equation}\label{eq: testable cycle equation}
    \sum_{n=1}^{N} \left( \func^{\data}_{a_{n},\menu_{n}}-\func^{\data}_{a_{n-1},\menu_{n}} \right) = \mathbf{0}.
\end{equation}
Corollary~\ref{cor: testable cycles} below states that a conditionally full-support data set is consistent only if equation~\eqref{eq: testable cycle equation} holds for all testable cycles. This corollary also shows that any conditionally full-support data set satisfying equation~\eqref{eq: testable cycle equation} for all testable cycles is consistent. In fact, the corollary shows something stronger: to prove that a data set is consistent, it suffices to establish equation~\eqref{eq: testable cycle equation} only for the testable cycles contained in a given cycle basis.

\begin{samepage}
\begin{corollary}\label{cor: testable cycles}
Suppose $\icost$ satisfies Assumptions~\ref{ass: smoothness}\eqref{ass: smoothness, part finite} and \ref{ass: smoothness}\eqref{ass: smoothness, part differentiability}, and let $\data$ be a data set with conditionally full support. Given any cycle basis, the following are equivalent: \begin{enumerate}[(i)]
    \item The data set $\data$ is consistent. \label{cor: testable cycles - consistency}
    \item Equation~\eqref{eq: testable cycle equation} holds for every testable cycle. \label{cor: testable cycles - all cycles}
    \item Equation~\eqref{eq: testable cycle equation} holds for every testable cycle in the cycle basis. \label{cor: testable cycles - basis condition}
\end{enumerate}
\end{corollary}
\end{samepage}

We now briefly discuss the logic behind Corollary~\ref{cor: testable cycles}. The equivalence between \eqref{cor: testable cycles - all cycles} and \eqref{cor: testable cycles - basis condition} follows from observing that equation~\eqref{eq: testable cycle equation} requires all cycle vectors to be in the kernel of a fixed linear map. Thus, the crux of the corollary is in the equivalence between parts \eqref{cor: testable cycles - consistency} and \eqref{cor: testable cycles - all cycles}. To understand this equivalence, revisiting Proposition~\ref{prop: multiplier result for iteratively differentiable case} is useful. According to this proposition, a conditionally full support data set $\data=(\DMenu,\dscr)$ is consistent if and only if we can find two vectors of $\lp{1}$ functions, $(\ut_a)_{a\in A}$, and $(\nuis_{\menu})_{\menu \in \DMenu}$, such that for all $a \in A$ and $\menu \in \DMenu$,
\begin{equation}\label{eq: func equals ut minus nuis}
\func^{\data}_{a,\menu} = \ut_a - \nuis_\menu.
\end{equation}
Thus, Proposition~\ref{prop: multiplier result for iteratively differentiable case} reduces the question of whether a data set is consistent to a specific instance of a standard problem in electrical networks: when can one assign each $a\in A$ an electric potential $\ut_{a}$, and each $\menu \in \DMenu$ an electric potential $\nuis_{\menu}$, so that $\func^{\data}_{a,\menu}$ gives the difference in potentials between $a$ and $\menu$? The answer is that such an assignment is feasible if and only if $\func^{\data}$ satisfies \emph{Kirchhoff's voltage law}, which corresponds exactly with equation~\eqref{eq: testable cycle equation} holding for all testable cycles. 

As noted above, Part~\eqref{cor: testable cycles - basis condition} of Corollary~\ref{cor: testable cycles} tells us that to establish consistency of a data set, it suffices to test equation~\eqref{eq: testable cycle equation} only for the cycles contained in some cycle basis. One might wonder whether it is possible to do the same with a different, potentially smaller set of cycles. The next result establishes the answer is negative. 

\begin{samepage}
\begin{proposition}\label{prop: testable cycles converse}
Suppose $\icost$ satisfies Assumptions~\ref{ass: regularity} and~\ref{ass: smoothness}. Let $\DMenu\subseteq\Menu$ be any set, and let $\Cycle$ be a set of testable cycles (for full-support data sets on $\DMenu$) that contains no cycle basis. Then some inconsistent, conditionally full-support data set $(\DMenu,\dscr)$ exists such that equation~\eqref{eq: testable cycle equation} holds for every cycle in $\Cycle$.
\end{proposition}
\end{samepage}

The above results allow analysts to test for a particular cost function without making assumptions on the agent's payoffs. Consider the task of testing whether costs are given by mutual information, for example. The experimental literature focuses on the Independent Likelihood Ratio (ILR) condition \citep[see][]{caplin2013behavioral,dean2019experimental}, which states that for finite states, $\scr$ is optimal when costs are given by $\icost(\ip)=\param \int \KL (\posterior)\ \ip(\dd \posterior)$ for some $\param >0$ only if 
\[
\frac{\posterior^{\scr}_{a}(\omega)}{e^{\ut_a(\omega)/\param}} = \frac{\posterior^{\scr}_{b}(\omega)}{e^{\ut_b(\omega)/\param}}
\]
for all $a$ and $b$ that $\scr$ generates with positive probability. Testing this condition requires analysts to know the agent's payoffs from different actions. To pin down these payoffs, experimenters either assume participants' payoffs are linear in money (see Caplin and Dean 2013), or craft their experiment so that rewards are given via ``probability points'' (e.g., \cite{caplin2020rational}, or Dean and Neligh 2022). By contrast, Corollary~\ref{cor: testable cycles} allows one to test for mutual information costs using the agent's decisions across menus. Specifically, Corollary~\ref{cor: testable cycles} implies a data set with conditionally full-support $\dscr$ is consistent with the cost function $\icost(\ip)=\param \int \KL (\posterior)\ \ip(\dd \posterior)$ if and only if every testable cycle satisfies equation \eqref{eq: testable cycle equation}, with\footnote{To obtain this condition, observe that if $\scr$ has full support, the iterated derivative of mutual information costs is given by $\icostdd_{\posterior_{\scr}^a} = \log \frac{\dd\posterior^{\scr}_{a}}{\dd \prior} = \log \frac{\scr_{a}}{\ip^{\scr}_a}.$} 
\[
\func^{\data}_{a,\menu} = \theta \log \frac{\dscr^{\menu}_{a}}{\ip^{\dscr^{\menu}}_{a}}.
\]
For a demonstration, consider again the data set from equation~\eqref{eq: example dataset}. As explained earlier, the two cycles $\cycle_{1} = x \ \{x,y\} \ y \ \{x,y,z\} \ x$ and $\cycle_2 = y \ \{x,y,z\} \ z \ \{y,z\} \ y$ form a cycle basis. Hence, checking whether equation~\eqref{eq: testable cycle equation} holds for these two cycles is sufficient for testing consistency of that data set. Since there are three states, each of these cycles generates three equations, giving a total of six equations to check. For example, evaluating equation~\eqref{eq: testable cycle equation} for the first cycle $\cycle_{1}$ at state $\omega=1$ gives
\begin{equation*}
\begin{split}
& \func^{\data}_{x,\{x,y\}}(1) - \func^{\data}_{y,\{x,y\}}(1) + \func^{\data}_{y,\{x,y,z\}}(1) - \func^{\data}_{x,\{x,y,z\}}(1) \\ 
=& \param \left[\log \frac{1/3}{1/2} - \log \frac{2/3}{1/2} +  \log \frac{2/5}{1/3} - \log \frac{1/5}{1/3} \right] \\
=& 0.
\end{split}
\end{equation*}
More generally, one can show all instances of equation~\eqref{eq: testable cycle equation} holds for both cycles. Hence, one can use Corollary~\ref{cor: testable cycles} to show that the data set from \eqref{eq: testable cycle equation} is consistent with the agent having mutual information costs without having any payoff information.\footnote{Observe this test is independent of $\param$. This independence comes from the fact that an $\scr$ is $\ut$-rationalizable when costs are given by $\param\icost$ if and only if it is $(\ut/\param)$-rationalizable when costs are given by $\icost$.}

Corollary~\ref{cor: testable cycles} is also useful for testing for cost functions that do not have a revealed preference axiomatization. For an example, suppose one is interested in testing whether costs are given by mutual information to the power of $\param>1$,
\begin{equation*}
\icost(\ip) = \left[\int \KL(\posterior)\ \ip(\dd \posterior) \right]^{\param}. 
\end{equation*}
In this case, the function $\func^{\data}_{a,\menu}$ in equation \eqref{eq: testable cycle equation} becomes
\[
\func^{\data}_{a,\menu} = \param \left[\cost_{MI} (\dscr^{\menu})\right]^{\param-1}\log \frac{\dscr^{\menu}_{a}}{\ip^{\dscr^{\menu}}_{a}},
\]
where $\cost_{MI}(\scr) = \bbE \left[ \sum_{a \in A} \scr_{a}\log\frac{\scr_a}{\ip^{\scr}_{a}}\right]$ is the indirect cost function for the mutual information case. We now use our test to show the example from equation~\eqref{eq: example dataset} is inconsistent with the cost $\icost(\ip) = \left[\int \KL(\posterior)\ \ip(\dd \posterior) \right]^{\param}$ any $\param > 1$. To see this, evaluate the left hand side of equation~\eqref{eq: testable cycle equation} for the cycle $\cycle_{1} = x \ \{x,y\} \ y \ \{x,y,z\} \ x$ at state $\omega=1$ to get
\begin{align*}
\func^{\data}_{x,\{x,y\}}(1) - \func^{\data}_{y,\{x,y\}}(1) + \func^{\data}_{y,\{x,y,z\}}(1) - \func^{\data}_{x,\{x,y,z\}}(1) = \\ 
\left\{\left[\cost_{MI} (\dscr^{\{x,y,z\}})\right]^{\param-1} - \left[\cost_{MI} (\dscr^{\{x,y\}})\right]^{\param-1} \right\} \param \log 2,
\end{align*}
which differs from zero for all $\param>1$, because $\cost_{MI} (\dscr^{\{x,y,z\}}) \neq \cost_{MI} (\dscr^{\{x,y\}})$ and $\cost_{MI}\geq 0$. Thus, this example data set is inconsistent with a strictly convex power of mutual information.

So far, we focused on data sets with conditionally full support. We conclude the section by discussing what happens when this assumption is violated. For this purpose, we introduce two definitions. First, say a data set $\data=(\DMenu,\dscr)$ is \textbf{fully mixed} if $\dscr^{\menu}_{a}$ is strictly positive almost surely whenever $\dscr^{\menu}_{a}$ is not identical to zero. In other words, a fully-mixed data set has SCRs that use all actions \emph{in their support} in all states, but may not use all available actions in some menus. Second, refer to $\data=(\DMenu,\dscr)$ as having \textbf{full support} if it always uses all available actions; that is, if $\supp \ \dscr^{\menu} = \menu$ for all $\menu \in \DMenu$. Observe that a conditionally full-support data set is exactly one that is fully mixed and has full support. 

For a fully-mixed data set without full support, the restrictions imposed by \eqref{eq: testable cycle equation} are necessary for consistency, but not sufficient. The lack of sufficiency arises because \eqref{eq: testable cycle equation} does not incorporate the fact that certain actions are not taken in some menus. In Online Appendix~\ref{app: aux}, we present an example of a data set that satisfies equation \eqref{eq: testable cycle equation} for all cycles but is inconsistent. 

For full-support data sets that are not fully mixed, our test is sufficient for consistency, but is not necessary. For intuition let us return to Proposition~\ref{prop: multiplier result for iteratively differentiable case}. This proposition implies that a full support data set is consistent if and only if we can find a utility $\ut \in \UT$, and multipliers $\nuis \in \lp{1}^{\DMenu}$ and $\mult \in \lp{1}_{+}^{A \times \DMenu}$, such that
\begin{equation}\label{eq: func equals ut minus nuis plus mult}
\forall a \in A \text{ and } \menu \in \DMenu: \func^{\data}_{a,\menu} = \ut_{a} - \nuis_{\menu} + \mult_{a,\menu} \ \ \text{and} \ \ \mult_{a,\menu}(\cdot)\dscr^{\menu}_{a}(\cdot) = 0 \text{ almost surely}.
\end{equation}
As explained earlier, our test is equivalent to satisfying the above with the additional requirement that $\mult$ is zero. Thus, our test is clearly sufficient for consistency of a full-support data set. When the data set is also fully-mixed, our test is also necessary, because $\mult_{a,\menu}(\cdot)\dscr^{\menu}_{a}(\cdot)$ can be zero for all $a\in A$ and $\menu\in \DMenu$ if and only if $\mult=\mathbf{0}$. However, without the fully-mixed requirement, our test is too strong, because $\mult$ need not equal zero for \eqref{eq: func equals ut minus nuis plus mult} to hold. And indeed, Online Appendix~\ref{app: aux} presents an example of a consistent data set that fails the test outlined in Corollary~\ref{cor: testable cycles}.\footnote{Observe that any such example must violate Assumption~\ref{ass: smoothness}. Indeed, if 
Assumption~\ref{ass: smoothness}\eqref{ass: smoothness, part inada} were satisfied, every consistent full-support data set would have conditionally full support.}

\section{Discussion}\label{sec: discussion}
In this section, we discuss our model's assumptions, additional results, and the relationship to existing literature.

\paragraph*{\emph{Partial Knowledge of Benefits.}}
For a fixed menu, we studied the restrictions imposed on the agent's behavior by her information acquisition costs with complete knowledge and complete ignorance of her preferences. In the appendix, we develop tools for analyzing an intermediate case in which the analyst knows the agent's utility belongs to a well-behaved set, $\UTB\subseteq \UT$.  In particular, we characterize when an SCR $\scr$ can be rationalized by some utility in $\UTB$. The characterization involves comparing the directional derivative of $\cost$ at $\scr$ with the \emph{support function} of $\UTB$. Intuitively, the directional derivative gives the marginal cost of shifting behavior away from $\scr$, whereas the support function gives the maximum benefit of shifting one's SCR towards $\scr$. Our result shows that, whenever the latter is lower than the former, some utility in $\UTB$ rationalizes $\scr$. 

To apply the above-mentioned result, one needs to calculate the support function of the set of utilities $\UTB$, and the directional derivative of the agent's indirect cost $\cost$. Calculating the support function of $\UTB$ is a standard optimization problem. In the appendix we solve this problem explicitly for a few economically relevant sets of utilities. 

In general, calculating the directional derivative of $\cost$ can be a complicated problem. To simplify it, we show that for differentiable costs, one can replace $\cost$ with its posterior separable approximation, $\cost_{\icostd}$. This replacement is useful, because one can calculate the directional derivative of $\cost_{\icostd}$ using the directional derivative of $\icostd$, which is often easier to derive. We refer the reader to the appendix for the exact statement of these results.

\paragraph*{\emph{Partial Knowledge of Costs.}} Our model assumes the analyst knows the agent's learning costs exactly. Maintaining this stylized assumption enables us to study the degree to which the agent's learning costs pin down her behavior. In practice, many analysts may not know $\icost$ exactly, but are instead capable of restricting it to belong to some set $\mathfrak C$. Some of our results also speak to this case. For example, Proposition~\ref{prop: dense set of rationalizable SCRs} implies the set of SCRs that is rationalizable by some cost function in $\mathfrak C$ is uniformly dense in the set of SCRs that are feasible for some $\icost \in \mathfrak C$. Similarly, if each $\icost\in\mathfrak C$ satisfies \ref{ass: regularity},  Theorem~\ref{thm: dense set of uniquely rationalizable SCRs} immediately implies the set of uniquely rationalizable SCRs is uniformly dense in the set of SCRs that can be induced at finite $\icost$-cost for some $\icost\in\mathfrak C$. Lemma~\ref{lem: unique prediction with known benefits} also extends somewhat: if $\mathfrak C$ is countable, the set of utilities that does not generate a unique prediction for some $\icost \in \mathfrak C$ is meager and shy. Hence, for most utility functions, the analyst's uncertainty about the agent's behavior reduces to the uncertainty about $\icost$---provided $\mathfrak C$ is countable. The reason is that a countable union of meager and shy sets is itself meager and shy. By contrast, we do not know of an immediate way to extend Theorem~\ref{thm: cross-choice predictions} to accommodate multiple cost functions. 

Generalizing Corollary~\ref{cor: testable cycles} to a set $\mathfrak C$ of cost functions satisfying Assumption~\ref{ass: smoothness}\eqref{ass: smoothness, part finite} and \eqref{ass: smoothness, part differentiability} is straightforward: a data set is consistent with the set $\mathfrak C$ if and only if the Corollary~\ref{cor: testable cycles}'s test holds for some cost function in the set. Even if the data set is consistent with $\mathfrak C$, one can still use Corollary~\ref{cor: testable cycles}'s test to identify which subset of $\mathfrak C$ could have generated a given data set. For an example, consider an analyst equipped with the data set from \eqref{eq: example dataset} who believes costs are given by the LLR cost function (see Example~\ref{ex: LLR costs}), but is unsure about the value of the parameter $\boldsymbol{\param}.$ Using the test in Corollary~\ref{cor: testable cycles}, the analyst can deduce $\boldsymbol{\param}$ must be symmetric. Moreover, every symmetric $\boldsymbol{\param}$ is consistent with the data set in \eqref{eq: example dataset} (see Appendix~\ref{app: sets of costs} for more details).

\paragraph*{\emph{Unique Rationalizability and Strict Convexity.}} Our analysis showed that, under \ref{ass: regularity}, one can rule out indifference as the source of the analyst's inability to predict behavior using the agent's learning costs. We focused on \ref{ass: regularity} because it accommodates the important case of posterior separable costs. An alternative way to obtain a similar result is to look at costs that are strictly convex. We now provide such a result.

\begin{proposition}\label{prop: dense set of uniquely rationalizable SCRs under strict convexity}
Suppose $\icost$ is strictly convex on $\Domicost$, and $\Domicost\neq\{\delta_\prior\}$. Then, a set of uniquely rationalizable SCRs exists that is uniformly dense in $\Dom$ and is open if $\Omega$ is finite. Moreover, this set of SCRs is rationalized by an open set of utilities. 
\end{proposition}

As the above result highlights, some rationalizable SCRs need not be uniquely rationalizable even when $\icost$ is strictly convex. The reason is that some convex combinations of SCRs change the way the agent randomizes over actions conditional on her information without changing the information itself. For example, suppose the action set is binary, $A=\{0,1\}$, and take $\scr$ and $\scrb$ to be the SCRs that respectively take action $1$ and $0$ regardless of the state. Clearly, both SCRs reveal an uninformative information policy, $\ip^{\scr} = \ip^{\scrb} = \delta_{\prior}.$ Moreover, the same is true for any convex combination of $\scr$ and $\scrb$, because any such combination results in the agent's actions being independent of the state. It follows the cost of any such convex combination is identical to the cost of $\scr$ and $\scrb$. In other words, even when $\icost$ is strictly convex, $\cost$ is still affine over some line segments. 

To prove Proposition~\ref{prop: dense set of uniquely rationalizable SCRs under strict convexity}, we use strict convexity of $\icost$ to identify regions where $\cost$ is strictly convex. Specifically, we show $\cost$ is strictly convex over any line segment with an end point that satisfies the following property: every action is used with positive probability, and no two actions reveal the same posterior. Therefore, any rationalizable SCR with this property is uniquely rationalizable. The proposition's proof then proceeds as the proof of Theorem~\ref{thm: dense set of uniquely rationalizable SCRs}, but with SCRs with the previously mentioned property taking the role of the set of linearly independent SCRs.

Next, we note substituting strict convexity of $\icost$ for \ref{ass: regularity} does not alter the conclusions of Theorem~\ref{thm: cross-choice predictions}. 
\begin{proposition}\label{prop: cross-choice predictions under strict convexity}
If $\icost$ is strictly convex on its domain, $|\Omega|>1$, and \ref{ass: smoothness} holds, the set of SCRs that yield unique subset predictions is weak* dense.
\end{proposition}
The argument for Proposition~\ref{prop: cross-choice predictions under strict convexity} follows the same reasoning as Theorem~\ref{thm: cross-choice predictions}. We refer the reader to the appendix for the specific details.

\paragraph*{\emph{Subdifferentials, Rationalizability, and Posterior Separable Costs.}}

Among their many contributions, \cite{caplin2021rationally} also identify a connection between rationalizability and subdifferentiability for the case in which $\icost$ is posterior separable. More specifically, \cite{caplin2021rationally} show that if $\Omega$ is finite and $\icost$ is posterior separable, $\scr$ is rationalizable if and only if $\icostd$ has a nonempty subdifferential at all beliefs in the support of $\scr$'s revealed information policy; that is, $\partial \icostd(\posterior_{a}^{\scr})\neq \varnothing$ for all $a \in \supp\ \scr$.\footnote{See the definition of $\partial \icostd(\posterior_{a}^{\scr})$ in Appendix \ref{app: On Iterative Differentiability}.} 
Because rationalizability of $\scr$ is equivalent to $\partial \cost(\scr)$ being nonempty, \citeapos{caplin2021rationally} result delivers the following conclusion: whenever $\Omega$ is finite and $\icost$ is posterior separable, $\partial \cost(\scr)$ is nonempty if and only if $\partial\icostd(\posterior_{a}^{\scr})$ is nonempty for all $a \in \supp \ s$. 

One can decompose \citeapos{caplin2021rationally} argument into two. First, they show one can construct a utility function that rationalizes $\scr$ by setting $\ut_{a}$ to be an appropriately normalized member of $\partial \icostd(\posterior_{a}^{\scr})$. In the appendix, we show \citeapos{caplin2021rationally} construction extends to infinite states and the case in which costs are merely differentiable. In other words, we show $\partial \cost(\scr)$ is nonempty whenever $\icost$ admits a derivative $\icostd$ at $\ip^{\scr}$ for which $\partial \icostd(\posterior_{a}^{\scr})$ is nonempty for all $a \in \supp (\scr)$. 

\cite{caplin2021rationally} also establish a converse: when costs are posterior separable and the state is finite, $\scr$ is rationalizable \emph{only if} $\partial \icostd(\posterior_{a}^{\scr})$ is nonempty for all $a$ in $\scr$'s support. To prove this claim, \cite{caplin2021rationally} prove a duality result to obtain a Kuhn-Tucker-like necessary condition for $\scr$ to be $\ut$-rationalizable, and show adding the relevant multiplier to $\ut_{a}$ witnesses  $\partial \icostd(\posterior_{a}^{\scr})$ being nonempty. Our results imply this approach generalizes to differentiable costs as well. The reason is that, under Lemma~\ref{lem: posterior separable approximation characterizes optimality}'s conditions, $\scr$ is rationalizable if and only if it is rationalizable by $\icost$'s  posterior separable approximation at $\ip^{\scr}$. Thus, given $\scr$ and a finite-valued cost that admits a derivative $\icostd$ at $\ip^{\scr}$, the SCR $\scr$ is rationalizable only if $\partial \icostd(\posterior_{a}^{\scr})$ is nonempty at all $a \in \supp (\scr)$---\emph{provided} the state is finite. Finite states are necessary because \citeapos{caplin2021rationally} duality result may not apply with infinite states. To establish similar duality results for the infinite state case, one usually needs additional regularity conditions \citep[see, e.g.,][]{gretsky2002subdifferentiability,dworczak2019persuasion}. Lacking such a result or an alternative proof method, we do not know whether $\partial \cost(\scr)$ being nonempty implies the nonemptiness of $\partial \icostd(\posterior^{\scr}_{a})$ for all $a\in \supp (\scr)$ when the state is infinite.

The construction from the first part of \citeapos{caplin2021rationally} argument delivers an alternative way of proving Proposition~\ref{prop: dense set of rationalizable SCRs} for the case in which the state is finite and costs are posterior separable. Given our above-mentioned result, the same argument extends to costs that are differentiable. For an explanation, recall the subdifferential of a convex function is nonempty over the relative interior of its domain \citep[see, e.g.,][ Theorem 23.4]{rockafellar1970convex}, which is always nonempty in finite dimensions. Because an interior $\scr$ reveals only posteriors that have full support, and because all such posteriors are interior when the state is finite, one gets that, when $\Omega$ is finite, $\partial\icostd (\posterior^{\scr}_{a})\neq \varnothing$ for all $a$ whenever $\scr$ is interior. Therefore, when costs are differentiable, one can use \citeapos{caplin2021rationally} construction to prove Proposition~\ref{prop: dense set of rationalizable SCRs} for the finite-state case. However, with infinite states,  \citeapos{caplin2021rationally} approach does not deliver an immediate proof for Proposition~\ref{prop: dense set of rationalizable SCRs}. The reason is that the relative interior of $\icostd$'s domain is empty, and so $\partial \icostd(\posterior^{\scr}_{a})$ may be empty as well. By focusing on the subdifferential of $\cost$ (which is well defined even when $\icost$ is not differentiable), our proof not only avoids this issue, but also establishes Proposition~\ref{prop: dense set of rationalizable SCRs} for a more general class of cost functions.

\paragraph*{\emph{Convexity and Monotonicity.}}

In our analysis, we assumed $\icost$ is monotone and convex. We now explain these two assumptions are essentially without loss. In particular, we argue the indirect cost function generated by every lower semicontinuous and proper $\hat{\icost}$ is identical to the indirect cost generated by a convex and monotone cost function.

To get the result, we must first redefine the agent's indirect cost function so as to allow randomization over information policies. Specifically, we let the agent choose a distribution over information policies, $\mip \in \Delta \Info$. We say such a distribution \textbf{can induce} an SCR $\scr$  if some action strategy $\alpha : \DO \rightarrow \Delta A$ is such that, for every $a \in A$ and every event $\tilde{\Omega}$,
\[
\int \alpha (a|\posterior) \posterior(\hat{\Omega}) \ \ip(\dd \posterior) \ \mip(\dd \ip) = \bbExp{\mathbf{1}_{\tilde{\Omega}}\scr_{a}}.
\]
The indirect cost function is then given by 
\begin{equation}\label{eq:Indirect Cost with Mixtures}
\cost(\scr) = \inf_{\mip \in \Delta\Info} \int \hat{\icost}(\ip)\ \mip(\dd\ip) \text{ s.t. }\mip \text{ can induce }\scr.
\end{equation}
Note randomization is not necessary when $\icost$ is convex, in which case the above reduces to our previous definition of $\cost$. The next result shows a sense in which such convexity always holds. 
\begin{proposition}\label{prop: monotonicity and convexity are wlog}
Let $\cost$ be the indirect cost function induced by a lower semicontinuous and proper $\hat{\icost}:\Info \rightarrow \extreal$. Then, the infimum in \eqref{eq:Indirect Cost with Mixtures} is attained. Moreover, $\cost$ is also induced by some cost function $\icost$ that is lower semicontinuous, proper, convex, and monotone.
\end{proposition}
Our argument begins by observing $\mip$ can induce $\scr$ if and only if its mean, $\ipb=\int \ip \ \mip(\dd\ip)$, is more informative than $\scr$'s revealed information policy $\ip^{\scr}$. It follows the agent's indirect cost function remains unchanged if we replace $\hat{\icost}$ with a cost function $\icost$ that assign each $\ip$ with the expected cost (under $\hat{\icost}$) of the cheapest randomization whose mean is more informative than $\ip$. We then show Berge's theorem guarantees $\icost$ is lower semicontinuous and that monotonicity and convexity of $\icost$ follow from $\succeq$ being a transitive order that respects convex combinations. 

\cite{caplin2015revealed} use a similar construction to show every behavior generated from costly flexible learning in a finite collection of menus can be rationalized using a convex and monotone cost function. \cite{de2017rationally} prove a representation result for preferences over menus with similar implications. In particular, they show one can always take the cost of $\ip$ to equal the maximum difference between the agent's benefit from using $\ip$ in some menu and her certainty equivalent for that menu. Moreover, the resulting cost function is the unique cost function that is simultaneously convex, monotone, zero at no information, and consistent with the agent's preferences over menus.

\paragraph*{\emph{Continuous Choice with Bounded Utilities.}}
Proposition~\ref{prop: continuous choice is impossible} shows the only way to guarantee continuous choice across all objective functions is to require all discontinuous SCRs to have infinite cost. We now explain one can avoid the use of infinite costs if one is willing to require the agent's choice to be continuous for all \emph{bounded} utility functions. In particular, $\icost$ generates continuous choice for all bounded utility functions if and only if it satisfies an infinite-slope condition. 
\begin{samepage}
\begin{proposition}
\label{prop: continous choice with bounded utilities}
An SCR $\scr$ is not rationalizable by any bounded utility function if and only if 
\begin{equation}\label{eq: infinite L1 steepness}
    \inf_{\scrb\in \Dom\setminus \{\scr\}} \frac{\cost(\scr) - \cost(\scrb)}{\Vert \scr - \scrb\Vert_{1}} = - \infty. 
\end{equation}
In particular, only continuous SCRs are rationalizable by a bounded utility function if and only if \eqref{eq: infinite L1 steepness} holds for all discontinuous $\scr$.
\end{proposition}
\end{samepage}

The result is an immediate consequence of \cite{gale1967geometric}, who shows bounded steepness is a necessary and sufficient condition for the subdifferential of a convex function to contain some linear function that is continuous with respect to a given norm.\footnote{Using identical reasoning, one can replace $\Vert \cdot \Vert_{1}$ in Proposition~\ref{prop: continous choice with bounded utilities}'s statement to analogously characterize which SCRs are rationalizable by other subspaces of utilities. See Online Appendix \ref{app: Continuous Choice with Bounded Utilities} for details.} 

Proposition~\ref{prop: continous choice with bounded utilities}'s infinite-slope condition is reminiscent of a different condition by \cite{morris2021coordination}, who introduce the notion of continuous choice to study equilibrium selection in global games. Holding other players' strategies fixed, one can view the problem of each agent in their game as an instance of our model in which $\Omega\subseteq \mathbb{R}$ is an interval, $A=\{0,1\}$, payoffs are bounded, and $\Dom$ is contained in the space of all SCRs for which $\scr_{1}$ is nondecreasing. \cite{morris2021coordination} show multiplying $\cost$ by a vanishing constant leads to a sharp equilibrium-selection result, provided only SCRs in $\SCR_{\text{AC}}:=\{\scr\in\SCR:\ \scr_1 \text{ absolutely continuous}\}$ can be rationalized by some monotone and bounded utility function. They also prove this latter property holds whenever $\cost$ satisfies a condition called \emph{expensive perfect discrimination}, which states that for every $\scr \in \Dom\setminus \SCR_{\text{AC}}$, the cost function $\cost$ exhibits unbounded $\Vert\cdot\Vert_1$-steepness in the direction of SCRs in $\SCR_{\text{AC}}$. By contrast, Proposition~\ref{prop: continous choice with bounded utilities}'s condition allows $\cost$ to exhibit unbounded steepness from \emph{any} direction, and shows allowing for these additional directions results in a condition that is both necessary and sufficient for ruling out discontinuous SCRs as rationalizable by any bounded utility, including utilities that are not monotone. 

\paragraph*{\emph{Costly Stochastic Choice.}} 
In our model, we assumed the agent faces a cost to acquire information, which we then used to derive an indirect cost function over the set of SCRs. By contrast, some models formulate a cost function $\tilde\cost$ on SCRs directly, without micro-founding it via information acquisition \citep[e.g.,][]{mattsson2002probabilistic,fosgerau2020discrete,flynn2021strategic,morris2021coordination}. Some of our results apply to those models as well, provided  $\tilde\cost$ is convex, proper, and weak* lower semicontinuous. Because the arguments for Lemma~\ref{lem: unique prediction with known benefits} and Proposition~\ref{prop: dense set of rationalizable SCRs} rely only on properties of the agent's indirect cost function, both of these results also apply to models in which the cost of an SCR does not originate from information acquisition. The same holds for parts \eqref{prop: continuous choice is impossible: feasible is continuous} and \eqref{prop: continuous choice is impossible: interior means finite state} of Proposition~\ref{prop: continuous choice is impossible}, as well as Proposition~\ref{prop: continous choice with bounded utilities}. 

To get an analogue of Theorem~\ref{thm: dense set of uniquely rationalizable SCRs}, the cost $\tilde{\cost}$ must satisfy additional properties. The most obvious such property is strict convexity: if $\tilde{\cost}$ is strictly convex, every rationalizable SCR is uniquely rationalizable, and so the logic behind Proposition~\ref{prop: dense set of rationalizable SCRs} delivers a uniformly dense set of SCRs that are uniquely rationalizable. 

With strict convexity, one can also get an analogue of Theorem~\ref{thm: cross-choice predictions}, provided $\tilde\cost$ is sufficiently smooth. Say $\ut \in \UT$ is a derivative of $\tilde\cost$ at $\scr$ if for all $\scrb$,
\[
d^{+}_{\scr}\tilde{\cost}(\scrb) = \bbExp{\ut \cdot (\scrb - \scr)}.
\]
Similar to the case in which $\icost$ is iteratively differentiable, we show in the appendix that given an interior $\scr$, the cost $\tilde\cost$ admits $\ut$ as a derivative at $\scr$ only if all utilities that rationalize $\scr$ differ from $\ut$ by a nuisance term; that is, $\utb$ rationalizes $\scr$ if and only if some $\nuis\in \lp{1}$ is such that $\utb_{a} = \ut_{a} + \nuis$ for all $a$. Armed with this observation, one can repeat the arguments guaranteeing Theorem~\ref{thm: cross-choice predictions} to show a weak*-dense set of SCRs exists that induce unique subset predictions, provided $\tilde \cost$ is finite-valued, admits a derivative at any interior SCR, and has infinite slope at the edges. Thus, whereas $\tilde \cost$ imposes few restrictions on the agent's behavior in a given menu, across menus, one can still use $\tilde \cost$ to make meaningful predictions about the agent's choices.

\bibliographystyle{jpe}
\bibliography{InfoCosts}

\newcommand{\forth}[1]{}
\begin{thebibliography}{62}
\newcommand{\enquote}[1]{``#1''}
\providecommand{\natexlab}[1]{#1}
\providecommand{\url}[1]{\texttt{#1}}
\providecommand{\urlprefix}{URL }

\bibitem[{Aliprantis and Border(2006)}]{aliprantis2006infinite}
Aliprantis, Charalambos~D and Kim Border. 2006.
\newblock \emph{Infinite Dimensional Analysis: A Hitchhiker's Guide}.
\newblock Springer Science \& Business Media.

\bibitem[{Aumann and Maschler(1995)}]{Aumann1995}
Aumann, Robert~J and Michael Maschler. 1995.
\newblock \emph{Repeated games with incomplete information}.
\newblock MIT press.

\bibitem[{Beno{\^\i}t and Dubra(2011)}]{benoit2011apparent}
Beno{\^\i}t, Jean-Pierre and Juan Dubra. 2011.
\newblock \enquote{{Apparent Overconfidence}.}
\newblock \emph{Econometrica} 79~(5):1591--1625.

\bibitem[{Benyamini and Lindenstrauss(1998)}]{benyamini1998geometric}
Benyamini, Yoav and Joram Lindenstrauss. 1998.
\newblock \emph{Geometric Nonlinear Functional Analysis}, vol.~48.
\newblock American Mathematical Soc.

\bibitem[{Blackwell(1953)}]{Blackwell1953}
Blackwell, David. 1953.
\newblock \enquote{{Equivalent Comparisons of Experiments}.}
\newblock \emph{The Annals of Mathematical Statistics} 24~(2):265--272.

\bibitem[{Bloedel and Zhong(2020)}]{bloedel2020cost}
Bloedel, Alexander~W and Weijie Zhong. 2020.
\newblock \enquote{The Cost of Optimally-Acquired Information.}
\newblock \emph{Unpublished Manuscript, November} .

\bibitem[{Bogachev(2007)}]{bogachev2007measureV1}
Bogachev, Vladimir~Igorevich. 2007.
\newblock \emph{Measure Theory}, vol.~1.
\newblock Springer.

\bibitem[{Bollob{\'a}s(2012)}]{bollobas2012graph}
Bollob{\'a}s, B{\'e}la. 2012.
\newblock \emph{Graph theory: an introductory course}, vol.~63.
\newblock Springer Science \& Business Media.

\bibitem[{B{\"o}rgers(2015)}]{borgers2015introduction}
B{\"o}rgers, Tilman. 2015.
\newblock \emph{{An Introduction to the Theory of Mechanism Design}}.
\newblock Oxford University Press, USA.

\bibitem[{Borwein and Vanderwerff(2010)}]{borwein2010convex}
Borwein, Jonathan~M and Jon~D Vanderwerff. 2010.
\newblock \emph{Convex Functions: Constructions, Characterizations and
  Counterexamples}, vol. 172.
\newblock Cambridge University Press Cambridge.

\bibitem[{Br{\o}ndsted and
  Rockafellar(1965)}]{brondsted1965subdifferentiability}
Br{\o}ndsted, Arne and Ralph~Tyrrell Rockafellar. 1965.
\newblock \enquote{{On the Subdifferentiability of Convex Functions}.}
\newblock \emph{Proceedings of the American Mathematical Society}
  16~(4):605--611.

\bibitem[{Caplin et~al.(2020)Caplin, Csaba, Leahy, and
  Nov}]{caplin2020rational}
Caplin, Andrew, D{\'a}niel Csaba, John Leahy, and Oded Nov. 2020.
\newblock \enquote{Rational Inattention, Competitive Supply, and
  Psychometrics.}
\newblock \emph{The Quarterly Journal of Economics} 135~(3):1681--1724.

\bibitem[{Caplin and Dean(2013)}]{caplin2013behavioral}
Caplin, Andrew and Mark Dean. 2013.
\newblock \enquote{Behavioral Implications of Rational Inattention with Shannon
  Entropy.}
\newblock Tech. rep., National Bureau of Economic Research.

\bibitem[{Caplin and Dean(2015)}]{caplin2015revealed}
---{}---{}---. 2015.
\newblock \enquote{Revealed Preference, Rational Inattention, and Costly
  information acquisition.}
\newblock \emph{American Economic Review} 105~(7):2183--2203.

\bibitem[{Caplin, Dean, and Leahy(2017)}]{caplin2017rationally}
Caplin, Andrew, Mark Dean, and John Leahy. 2017.
\newblock \enquote{Rationally Inattentive Behavior: Characterizing and
  Generalizing Shannon Entropy.}
\newblock Tech. rep., National Bureau of Economic Research.

\bibitem[{Caplin, Dean, and Leahy(2019)}]{caplin2019rational}
---{}---{}---. 2019.
\newblock \enquote{Rational Inattention, Optimal Consideration Sets, and
  Stochastic Choice.}
\newblock \emph{The Review of Economic Studies} 86~(3):1061--1094.

\bibitem[{Caplin, Dean, and Leahy(2021)}]{caplin2021rationally}
---{}---{}---. 2021.
\newblock \enquote{{Rationally Inattentive Behavior: Characterizing and
  Generalizing Shannon Entropy}.}
\newblock Tech. rep., National Bureau of Economic Research.

\bibitem[{Cerreia-Vioglio, Maccheroni, and
  Marinacci(2017)}]{cerreia2017stochastic}
Cerreia-Vioglio, Simone, Fabio Maccheroni, and Massimo Marinacci. 2017.
\newblock \enquote{{Stochastic Dominance Analysis Without the Independence
  Axiom}.}
\newblock \emph{Management Science} 63~(4):1097--1109.

\bibitem[{Chambers, Liu, and Rehbeck(2020)}]{chambers2020costly}
Chambers, Christopher~P, Ce~Liu, and John Rehbeck. 2020.
\newblock \enquote{Costly Information Acquisition.}
\newblock \emph{Journal of Economic Theory} 186:104979.

\bibitem[{Csisz{\'{a}}r(1974)}]{Csiszar1974}
Csisz{\'{a}}r, Imre. 1974.
\newblock \enquote{{On an Extremum Problem of Information Theory}.}
\newblock \emph{Studia Scientiarum Mathematicarum Hungarica} 9:57--71.

\bibitem[{de~Oliveira et~al.(2017)de~Oliveira, Denti, Mihm, and
  Ozbek}]{de2017rationally}
de~Oliveira, Henrique, Tommaso Denti, Maximilian Mihm, and Kemal Ozbek. 2017.
\newblock \enquote{Rationally Inattentive Preferences and Hidden Information
  Costs.}
\newblock \emph{Theoretical Economics} 12~(2):621--654.

\bibitem[{Dean and Neligh(2019)}]{dean2019experimental}
Dean, Mark and Nathaniel Neligh. 2019.
\newblock \enquote{Experimental Tests of Rational Inattention.} .

\bibitem[{Denti(2022)}]{denti2022posterior}
Denti, Tommaso. 2022.
\newblock \enquote{Posterior-Separable Cost of Information.}
\newblock Tech. rep., working paper.

\bibitem[{Denti, Marinacci, and Montrucchio(2020)}]{denti2020note}
Denti, Tommaso, Massimo Marinacci, and Luigi Montrucchio. 2020.
\newblock \enquote{A Note on Rational Inattention and Rate Distortion Theory.}
\newblock \emph{Decisions in Economics and Finance} 43~(1):75--89.

\bibitem[{Denti et~al.(2021)Denti, Marinacci, Rustichini
  et~al.}]{denti2021experimental}
Denti, Tommaso, Massimo Marinacci, Aldo Rustichini et~al. 2021.
\newblock \emph{Experimental cost of information}.
\newblock IGIER, Universit{\`a} Bocconi.

\bibitem[{Dewan and Neligh(2020)}]{dewan2020estimating}
Dewan, Ambuj and Nathaniel Neligh. 2020.
\newblock \enquote{Estimating Information Cost Functions in Models of Rational
  Inattention.}
\newblock \emph{Journal of Economic Theory} 187:105011.

\bibitem[{Dworczak and Kolotilin(2019)}]{dworczak2019persuasion}
Dworczak, Piotr and Anton Kolotilin. 2019.
\newblock \enquote{The Persuasion Duality.}
\newblock \emph{Available at SSRN 3474376} .

\bibitem[{Flynn and Sastry(2021)}]{flynn2021strategic}
Flynn, Joel~P and Karthik Sastry. 2021.
\newblock \enquote{Strategic Mistakes.}
\newblock \emph{Available at SSRN 3663481} .

\bibitem[{Fosgerau et~al.(2020)Fosgerau, Melo, De~Palma, and
  Shum}]{fosgerau2020discrete}
Fosgerau, Mogens, Emerson Melo, Andre De~Palma, and Matthew Shum. 2020.
\newblock \enquote{Discrete Choice and Rational Inattention: A General
  Equivalence Result.}
\newblock \emph{International economic review} 61~(4):1569--1589.

\bibitem[{Frankel and Kamenica(2019)}]{frankel2019quantifying}
Frankel, Alexander and Emir Kamenica. 2019.
\newblock \enquote{Quantifying Information and Uncertainty.}
\newblock \emph{American Economic Review} 109~(10):3650--80.

\bibitem[{Fudenberg, Iijima, and Strzalecki(2015)}]{fudenberg2015stochastic}
Fudenberg, Drew, Ryota Iijima, and Tomasz Strzalecki. 2015.
\newblock \enquote{Stochastic Choice and Revealed Perturbed Utility.}
\newblock \emph{Econometrica} 83~(6):2371--2409.

\bibitem[{Gale(1967)}]{gale1967geometric}
Gale, David. 1967.
\newblock \enquote{A Geometric Duality Theorem with Economic Applications.}
\newblock \emph{The Review of Economic Studies} 34~(1):19--24.

\bibitem[{Gretsky, Ostroy, and Zame(2002)}]{gretsky2002subdifferentiability}
Gretsky, Neil~E, Joseph~M Ostroy, and William~R Zame. 2002.
\newblock \enquote{Subdifferentiability and The Duality Gap.}
\newblock \emph{Positivity} 6~(3):261--274.

\bibitem[{H{\'e}bert and Woodford(2021{\natexlab{a}})}]{hebert2021neighborhood}
H{\'e}bert, Benjamin and Michael Woodford. 2021{\natexlab{a}}.
\newblock \enquote{Neighborhood-Based Information Costs.}
\newblock \emph{American Economic Review} 111~(10):3225--55.

\bibitem[{H{\'e}bert and Woodford(2021{\natexlab{b}})}]{hebert2021time}
---{}---{}---. 2021{\natexlab{b}}.
\newblock \enquote{Rational Inattention when Decisions Take Time.} .

\bibitem[{Hofbauer and Sandholm(2002)}]{hofbauer2002global}
Hofbauer, Josef and William~H Sandholm. 2002.
\newblock \enquote{On the Global Convergence of Stochastic Fictitious Play.}
\newblock \emph{Econometrica} 70~(6):2265--2294.

\bibitem[{Hong, Karni, and Safra(1987)}]{hong1987risk}
Hong, Chew~Soo, Edi Karni, and Zvi Safra. 1987.
\newblock \enquote{{Risk Aversion in the Theory of Expected Utility with Rank
  Dependent Probabilities}.}
\newblock \emph{Journal of Economic theory} 42~(2):370--381.

\bibitem[{Hunt, Sauer, and Yorke(1992)}]{hunt1992prevalence}
Hunt, Brian~R, Tim Sauer, and James~A Yorke. 1992.
\newblock \enquote{Prevalence: A Translation-Invariant “Almost Every” on
  Infinite-Dimensional Spaces.}
\newblock \emph{Bulletin of the American mathematical society} 27~(2):217--238.

\bibitem[{Kallenberg(2017)}]{kallenberg2017random}
Kallenberg, Olav. 2017.
\newblock \emph{Random Measures, Theory and Applications}, vol.~1.
\newblock Springer.

\bibitem[{Kamenica and Gentzkow(2011)}]{Kamenica2011}
Kamenica, Emir and Matthew Gentzkow. 2011.
\newblock \enquote{{Bayesian Persuasion}.}
\newblock \emph{American Economic Review} 101~(October):2590--2615.

\bibitem[{Khan et~al.(2019)Khan, Yu, Zhang et~al.}]{khan2019information}
Khan, Ali, Haomiao Yu, Zhixiang Zhang et~al. 2019.
\newblock \enquote{Information Structures on a General State Space: An
  Equivalence Theorem and an Application.}
\newblock Tech. rep.

\bibitem[{Lin(2022)}]{lin2022stochastic}
Lin, Yi-Hsuan. 2022.
\newblock \enquote{Stochastic Choice and Rational Inattention.}
\newblock \emph{Journal of Economic Theory} :105450.

\bibitem[{Lipnowski and Mathevet(2018)}]{lipnowski2018disclosure}
Lipnowski, Elliot and Laurent Mathevet. 2018.
\newblock \enquote{Disclosure to a Psychological Audience.}
\newblock \emph{American Economic Journal: Microeconomics} 10~(4):67--93.

\bibitem[{Mat{\v{e}}jka and McKay(2015)}]{matejka2015rational}
Mat{\v{e}}jka, Filip and Alisdair McKay. 2015.
\newblock \enquote{Rational Inattention to Discrete Choices: A New Foundation
  for the Multinomial Logit Model.}
\newblock \emph{American Economic Review} 105~(1):272--98.

\bibitem[{Mattsson and Weibull(2002)}]{mattsson2002probabilistic}
Mattsson, Lars-G{\"o}ran and J{\"o}rgen~W Weibull. 2002.
\newblock \enquote{Probabilistic Choice and Procedurally Bounded Rationality.}
\newblock \emph{Games and Economic Behavior} 41~(1):61--78.

\bibitem[{McFadden(1974)}]{mcfadden1974measurement}
McFadden, Daniel. 1974.
\newblock \enquote{The Measurement of Urban Travel Demand.}
\newblock \emph{Journal of public economics} 3~(4):303--328.

\bibitem[{Mensch(2018)}]{mensch2018cardinal}
Mensch, Jeffrey. 2018.
\newblock \enquote{Cardinal Representations of Information.}
\newblock \emph{Available at SSRN 3148954} .

\bibitem[{Morris and Strack(2019)}]{morris2019wald}
Morris, Stephen and Philipp Strack. 2019.
\newblock \enquote{The Wald Problem and the Relation of Sequential Sampling and
  Ex-Ante Information Costs.}
\newblock \emph{Available at SSRN 2991567} .

\bibitem[{Morris and Yang(2021)}]{morris2021coordination}
Morris, Stephen and Ming Yang. 2021.
\newblock \enquote{Coordination and Continuous Stochastic Choice.} .

\bibitem[{Norets and Takahashi(2013)}]{norets2013surjectivity}
Norets, Andriy and Satoru Takahashi. 2013.
\newblock \enquote{On the Surjectivity of the Mapping Between Utilities and
  Choice Probabilities.}
\newblock \emph{Quantitative Economics} 4~(1):149--155.

\bibitem[{Phelps(2001)}]{phelps2001lectures}
Phelps, Robert~R. 2001.
\newblock \emph{Lectures on Choquet's theorem}.
\newblock Springer Science \& Business Media.

\bibitem[{Phelps(2009)}]{phelps2009convex}
---{}---{}---. 2009.
\newblock \emph{Convex Functions, Monotone Operators and Differentiability},
  vol. 1364.
\newblock Springer.

\bibitem[{Pomatto, Strack, and Tamuz(2020)}]{pomatto2020cost}
Pomatto, Luciano, Philipp Strack, and Omer Tamuz. 2020.
\newblock \enquote{The Cost of Information.} .

\bibitem[{Posner(1975)}]{posner1975random}
Posner, Edward. 1975.
\newblock \enquote{Random Coding Strategies for Minimum Entropy.}
\newblock \emph{IEEE Transactions on Information Theory} 21~(4):388--391.

\bibitem[{Pourciau(1983)}]{pourciau1983multiplier}
Pourciau, BH. 1983.
\newblock \enquote{Multiplier Rules and the Separation of Convex Sets.}
\newblock \emph{Journal of Optimization Theory and Applications}
  40~(3):321--331.

\bibitem[{Rockafellar(1970)}]{rockafellar1970convex}
Rockafellar, Ralph~Tyrell. 1970.
\newblock \emph{Convex Analysis}.
\newblock Princeton University Press.

\bibitem[{Sims(1998)}]{sims1998stickiness}
Sims, Christopher~A. 1998.
\newblock \enquote{Stickiness.}
\newblock In \emph{Carnegie-Rochester Conference Series on Public Policy},
  vol.~49. Elsevier, 317--356.

\bibitem[{Sims(2003)}]{sims2003implications}
---{}---{}---. 2003.
\newblock \enquote{{Implications of Rational Inattention}.}
\newblock \emph{Journal of Monetary Economics} 50~(3):665--690.

\bibitem[{Sims(2006)}]{sims2006rational}
---{}---{}---. 2006.
\newblock \enquote{{Rational Inattention: Beyond the Linear-Quadratic Case}.}
\newblock \emph{American Economic Review} 96~(2):158--163.

\bibitem[{Sion(1958)}]{sion1958general}
Sion, Maurice. 1958.
\newblock \enquote{On General Minimax Theorems.}
\newblock \emph{Pacific Journal of Mathematics} 8~(1):171--176.

\bibitem[{Srivastava(2008)}]{srivastava2008course}
Srivastava, Sashi~Mohan. 2008.
\newblock \emph{A Course on Borel Sets}, vol. 180.
\newblock Springer Science \& Business Media.

\bibitem[{Winkler(1988)}]{winkler1988extreme}
Winkler, Gerhard. 1988.
\newblock \enquote{Extreme Points of Moment Sets.}
\newblock \emph{Mathematics of Operations Research} 13~(4):581--587.

\end{thebibliography}

\newpage

\appendix

\begin{center}
{\LARGE \sc Online Appendix}
\end{center}

\section{Proof Appendix}

\subsection{Section~\ref{sec: cost minimization} Proofs}

\begin{proof}[Proof of Lemma~\ref{lem: which policies can induce}]
First, suppose $\ip\succeq\ip^\scr$, as witnessed by $\mps:\DO\to\Delta\DO$. Any $a\in \supp(\scr)$ then has $\mps(\cdot|\posterior^\scr_a)\ll\ip$ because $\ip$ is a proper weighted average of the finitely many measures $\left\{\mps(\cdot|\posterior^\scr_a) \right\}_{a\in\supp(\scr)}$. So let $\alpha_a:\DO\to[0,1]$ be some version of the scaled Radon-Nikodym derivative $\ip^\scr_a\ \tfrac{\dd \mps(\cdot|\posterior^\scr_a)}{\dd\ip}$ for each $a\in A$ with $\ip^\scr_a>0$; and let $\alpha_a=\mathbf0$ for every other $a\in A$. By construction, $\sum_{a\in A} \alpha_a=_{\ip\text{-a.e.}}\mathbf1$, so we can change $\{\alpha_a\}_{a\in A}$ on a $\ip$-null set to ensure the equation holds globally. Let us now show the strategy $(\ip,\alpha)$ induces $\scr$, where $\alpha:=\sum_{a\in A} \alpha_a \delta_a$. Indeed, for any action $a\in A$ and event $\hat\Omega\subseteq\Omega$, the strategy's induced probability of action $a$ being played and event $\hat\Omega$ occurring is zero (like under $\scr$) if $\ip^\scr_a$ is zero, and is otherwise equal to \begin{eqnarray*}
\int \posterior(\hat\Omega)\alpha_a(\posterior) \ \ip(\dd\posterior) 
&=& \int \posterior(\hat\Omega) \ip^\scr_a\ \tfrac{\dd r(\cdot|\posterior^\scr_a)}{\dd\ip}(\posterior) \ \ip(\dd\posterior)  \\
&=& \ip^\scr_a\int \posterior(\hat\Omega) \ \mps(\dd\posterior|\posterior^\scr_a) \\
&=& \ip^\scr_a \posterior^\scr_a(\hat\Omega) \\
&=& \bbExp{\mathbf1_{\hat\Omega} \ \scr_a}.
\end{eqnarray*}
Therefore, $\ip$ can induce $\scr$.

Conversely, suppose some strategy $(\ip,\alpha)$ induces $\scr$. 
For any $a\in A$ with $\ip^\scr_a>0$, we can define $\ipb^a\in\Delta\DO$ by letting $\ipb^a(D):=\tfrac1{\ip^\scr_a} \int_D \alpha(a|\posterior)\ \ip(\dd\posterior)$ for every Borel $D\subseteq\DO$. Every $a\in\supp(\scr)$ and every event $\hat\Omega\subseteq\Omega$ then have
$$
\int \posterior(\hat\Omega)\ \ipb^a(\dd\posterior) 
= \tfrac1{\ip^\scr_a} \int\posterior(\hat\Omega) \alpha(a|\posterior)\ \ip(\dd\posterior) 
= \tfrac1{\ip^\scr_a} \int_{\hat\Omega} \scr_a(\omega) \ \prior(\dd\omega)
= \posterior^\scr_a(\hat\Omega).
$$
Said differently, every $a\in\supp(\scr)$ has $\int \posterior\ \ipb^a(\dd\posterior) = \posterior^\scr_a$. Hence, 
$\ip = \sum_{a\in A} \ip^\scr_a \ipb^a \succeq \sum_{a\in A} \ip^\scr_a \delta_{\posterior^\scr_a}=\ip^\scr.$
\end{proof}

For the proof that follows, recall the Hardy-Littlewood-Polya-Blackwell-Stein-Sherman-Cartier theorem \citep[][p. 94]{phelps2001lectures}---hereafter, the \textbf{HLPBSSC theorem}---says $\ip,\ipb\in\Info$ satisfy $\ip\succeq\ipb$ if and only if $\int f\ \ip(\dd\posterior)\geq \int f\ \ipb(\dd\posterior)$ for every convex continuous $f:\DO\to\real$.

Our next lemma establishes several connections between stochastic choice rules and their revealed information policies. This lemma is the crucial step required for proving Lemma~\ref{lemma: Properties of Indirect Cost}.

\begin{samepage}
\begin{lemma}\label{lem: SCR basic properties from IP basic properties}
The following hold:
\begin{enumerate}[(i)]
\item\label{lem: SCR basic properties from IP basic properties-continuity} If $\sequence{\scr^{n}}$ weak* converges to $\scr$, then $\sequence{\ip^{\scr^{n}}}$ converges to $\ip^{\scr}$.\footnote{A sequence $\sequence{\scr^{n}}$ weak* converges to $\scr$ if $\bbExp{\ut_a \scr^n_a}$ 
converges to $\bbExp{\ut_a\scr_a}$ 
for all $\ut \in \UT$ and $a \in A$. Note the weak* topology on $\SCR\subset\SCRU$ is determined by its convergent sequences because the predual $\UT$ is separable.} 

\item\label{lem: SCR basic properties from IP basic properties-convexity} For any $\scr, \scrb$ in $\SCR$ and $\wt \in (0,1)$,
\begin{equation}\label{convex_info}
(1-\wt)\ip^{\scr} + \wt\ip^{\scrb} \succeq \ip^{(1-\wt)\scr + \wt \scrb}.
\end{equation}

\item\label{lem: SCR basic properties from IP basic properties-strictness} Moreover, the information ranking in \eqref{convex_info} is strict whenever some $a \in \supp (\scr) \cap \supp (\scrb)$ exists such that $\posterior_{a}^{\scr} \neq \posterior_{a}^{\scrb}$.

\end{enumerate}

\end{lemma}
\end{samepage}

\begin{proof}
Let us first prove part \eqref{lem: SCR basic properties from IP basic properties-continuity}. 
Recall $\scr^n\to \scr$ weak* in $\SCR$ tells us $$\ip^{\scr^n}_a\int f(\omega) \ \posterior^{\scr^n}_a(\dd\omega)\to \ip^{\scr}_a\int \func(\omega) \ \posterior^{\scr}_a(\dd\omega) \text{ as } n\to\infty, \ \forall a\in A \text{ and } \func\in\lp{1}.$$
Consequently, any $a\in A$ (specializing to $\func=\mathbf1$) has $\ip^{\scr^n}_a\to \ip^{\scr}_a$; and any $a\in \supp(\scr)$---scaling a given $\func\in\lp{1}$ by $\tfrac{\ip^{\scr^n}_a}{\ip^{\scr}_a}$, which converges to $1$---has $\int \func(\omega) \ \posterior^{\scr^n}_a(\dd\omega)\to \int \func(\omega) \ \posterior^{\scr}_a(\dd\omega)$ for every $\func\in\lp{1}$. Because every continuous $\func:\Omega\to\real$ represents some element of $\lp{1}$, the latter property tells us $\posterior^{\scr^n}_a\to \posterior^{\scr}_a$ in $\DO$ if $a\in\supp(\scr)$. Hence, $\ip^{\scr^n}\to \ip^{\scr}$ in $\Info$, as desired.

Now, we turn to parts \eqref{lem: SCR basic properties from IP basic properties-convexity} and \eqref{lem: SCR basic properties from IP basic properties-strictness}]
Let $\scrc:=(1-\wt)\scr+\wt\scrb$ and $\ip:=(1-\wt)\ip^\scr+\wt\ip^\scrb$ Direct computation shows $\supp(\scrc)= \supp(\scr)\cup \supp(\scrb)$ and, for every $a\in \supp(\scrc)$, \begin{eqnarray*}
\ip^\scrc_a &=& (1-\wt)\ip^{\scr}_a+\wt\ip^{\scrb}_a \\
\posterior^\scrc_a &=& \tfrac{(1-\wt)\ip^{\scr}_a}{\ip^{\scrc}_a} \posterior^{\scr}_a + \tfrac{\wt\ip^{\scrb}_a}{\ip^{\scrc}_a } \posterior^{\scrb}_a.
\end{eqnarray*}
Observe now that any convex continuous $f:\DO\to\real$ has \begin{eqnarray*}
\int f(\posterior) \ \ip(\dd\posterior) - \int f(\posterior) \ \ip^\scrc(\dd\posterior) &=& \sum_{a\in \supp(\scrc)} \left[ 
(1-\wt)\ip^{\scr}_a f\left(\posterior^{\scr}_a\right) + \wt\ip^{\scrb}_a f\left(\posterior^{\scrb}_a\right) - \ip^{\scrc}_a f\left(\posterior^{\scrc}_a\right)
 \right] \\
&=& \sum_{a\in \supp(\scrc)} \ip^{\scrc}_a \left[ 
\tfrac{(1-\wt)\ip^{\scr}_a}{\ip^{\scrc}_a} f\left(\posterior^{\scr}_a\right) + \tfrac{\wt\ip^{\scrb}_a}{\ip^{\scrc}_a } f\left(\posterior^{\scrb}_a\right) - f\left(\posterior^{\scrc}_a\right)
 \right] \\
&\geq& \sum_{a\in \supp(\scrc)} \ip^{\scrc}_a \left[ 
f\left(\tfrac{(1-\wt)\ip^{\scr}_a}{\ip^{\scrc}_a} \posterior^{\scr}_a + \tfrac{\wt\ip^{\scrb}_a}{\ip^{\scrc}_a } \posterior^{\scrb}_a\right) - f\left(\posterior^{\scrc}_a\right)
 \right] \\
 &=& \sum_{a\in \supp(\scrc)} \ip^{\scrc}_a \left[ f\left(\posterior^{\scrc}_a\right)- f\left(\posterior^{\scrc}_a\right)
 \right] \\
 &=&0.
\end{eqnarray*}
The HLPBSSC theorem then implies $\ip\succeq\ip^\scrc$, delivering part (ii). 

Toward~\eqref{lem: SCR basic properties from IP basic properties-strictness}, suppose some $a\in \supp(\scr)\cap\supp(\scrb)$ is such that $\posterior^{\scr}_{a}\neq \posterior^{\scrb}_{a}$. Now, specialize the above algebra to the case in which $\func|_{\co\left\{\posterior^{\scr}_{a}, \posterior^{\scrb}_{a}\right\}}$ is strictly convex.\footnote{For instance, $\func$ could be given by $\func(\posterior):=\left[\int \funcb(\omega)\ \posterior(\dd\omega)\right]^2$ for some continuous $\funcb:\Omega\to\real$ with $\int \funcb(\omega)\ \posterior^{\scr}_{a}(\dd\omega)\neq \int \funcb(\omega)\ \posterior^{\scrb}_{a}(\dd\omega)$.} Then, the inequality in the above chain is strict, witnessing $\int \func(\omega) \ \ip(\dd\omega) - \int \func(\omega) \ \ip^\scrc(\dd\omega)>0$ so that $\ip\succ\ip^\scrc$. The result follows.
\end{proof}

Now, we prove the indirect cost inherits the information cost's regularity properties.

\begin{proof}[Proof of Lemma~\ref{lemma: Properties of Indirect Cost}]
Because $\icost$ is proper, any constant SCR has finite cost, and so $\cost$ is proper. 
To prove that $\cost$ is convex and weak* lower semicontinuous, recall $\cost(\scr)=\icost(\ip^\scr)$. 

To see lower semicontinuity, consisder any sequence $\sequence{\scr^{n}}$ of SCRs converging to $\scr$. By Lemma~\ref{lem: SCR basic properties from IP basic properties}\eqref{lem: SCR basic properties from IP basic properties-continuity},
\[
\liminf_{n\to\infty} \cost(\scr^{n}) = \liminf_{n\to\infty} \icost(\ip^{\scr^{n}}) \geq \icost(\ip^{\scr}) = \cost(\scr),
\]
where the inequality follows from lower semicontinuity of $\icost$.

Toward convexity, take any $\scr,\scrb \in \SCR$ and $\wt \in (0,1)$. Lemma~\ref{lem: SCR basic properties from IP basic properties}\eqref{lem: SCR basic properties from IP basic properties-convexity} implies 
\begin{equation}\label{costconvex}
\begin{split}
\cost((1-\wt)\scr + \wt \scrb) 
& = \icost(\ip^{(1-\wt)\scr + \wt \scrb})
\\ & \leq \icost((1-\wt)\ip^{\scr} + \wt \ip^{\scrb})
\\ & \leq (1-\wt)\icost(\ip^{\scr}) + \wt \icost(\ip^{\scrb})
=(1-\wt)\cost(\ip^{\scr}) + \wt \cost(\ip^{\scrb}),
\end{split}
\end{equation}
where monotonicity of $\icost$ implies the first inequality, and convexity of $\icost$ delivers the second. 
\end{proof}

\subsection{On the Value Function}
Extend $\cost$ to $\SCRU$ by setting $\cost(\scr) = \infty$ for all $\scr \in \SCRU \setminus \SCR.$ The goal of this section is to prove some results regarding the optimal \textbf{value function} $\maxval:\UT \rightarrow \mathbb{R}$ defined as  
$$\maxval(\ut):=\max_{\scr \in \SCRU} \left[ \bbE[\ut\cdot\scr] - \cost(\scr) \right]=\max_{\scr \in \SCR} \left[ \bbE[\ut\cdot\scr] - \cost(\scr) \right].$$
Our results also pertain to the subdifferential of $\maxval$ at a utility $\ut$,  
\[
\begin{split}
\partial \maxval: \UT & \rightrightarrows \SCR, \\
\ut & \mapsto \left\{\scr \in \SCR:\ \maxval(\utb) \geq \maxval(\ut) + \bbExp{(\utb - \ut)\scr} \text{ for all }\utb \in \UT\right\}.
\end{split}
\]
For any subset $\SCRset \subseteq \SCR$, define the upper and lower inverses of $\partial\maxval$: 
\[
\begin{split}
\partial\maxval^{U}(\SCRset) & := \{\ut\in \UT:\ \partial\maxval(\ut) \subseteq \SCRset\},\\
\partial\maxval^{L}(\SCRset) & := \{\ut\in \UT:\ \partial\maxval(\ut) \cap \SCRset \neq \varnothing\}.
\end{split}
\]
We say $\partial\maxval$ is \textbf{norm-to-norm (resp. norm-to-weak*) upper hemicontinuous} if $\partial\maxval^{U}(\SCRset)$ is norm open whenever $\SCRset$ is norm (resp. weak*) open.

The following lemma establishes useful continuity properties of the optimal value function and its subdifferential.

\begin{lemma}\label{lem: value function continuity}
The value function $\maxval:\UT\to\real$ is convex and norm continuous, and its subdifferential is norm-to-weak* upper hemicontinuous.
\end{lemma}
\begin{proof}
The value function is real-valued because (as noted in the main text as a consequence of Banach-Alaoglu) each $\ut\in\UT$ admits some $\ut$-rationalizable SCR. It is convex as a supremum of affine functions.

To show $\maxval$ is norm-continuous, we need only show \citep[][Theorem 5.43]{aliprantis2006infinite} it is bounded above on some ball. 
And indeed, any $\ut \in \UT$ with $\Vert\ut\Vert_1\leq 1$ has
$$\maxval(\ut) \leq \max_{\scr \in \SCRU:\ \cost(\scr)<\infty} \bbE[\ut\cdot\scr] 
\leq\Vert \ut \Vert_1 \max_{\scr \in \SCR}\Vert \scr \Vert_\infty
= \Vert \ut \Vert_1
\leq 1.$$ 
Finally, upper hemicontinuity of $\partial \maxval$ then follows from Proposition 6.1.1 of \cite{borwein2010convex}.
\end{proof}

Next, we collect some standard facts from convex analysis, applied directly to our setting. To state them, given $\scr \in \Dom$ and $\scr'\in\SCR$, let $\dcost{\scr}(\scr')$ denote the \textbf{directional derivative} of $\cost$ at $\scr$ in direction $\scr'-\scr$.

\begin{lemma}\label{lem: subdifferential characterization of optimality}
Viewing $\UT$ with its norm topology and $\SCRU$ with its weak* topology, $\maxval$ is the convex conjugate of $\cost$, and $\cost$ is the convex conjugate of $\maxval$. Moreover, the following are equivalent:
\begin{enumerate}[(i)]
\item $\scr \in \argmax_{\scrb \in \SCR} \left[\bbE[\ut \cdot \scrb] - \cost(\scrb)\right].$
\item $\ut \in \partial \cost(\scr).$
\item Every $\scr'\in\Dom$ has $d_{\scr}^{+}\cost (\scr') \geq \bbExp{\ut\cdot(\scr'-\scr)}$.
\item $\scr \in \partial \maxval(\ut).$
\end{enumerate}
\end{lemma}
\begin{proof}
Recall $\UT$ with its norm topology and $\SCRU$ with its weak* topology form a dual pair with the bilinear map $(\ut,\scru)\mapsto\bbExp{\ut\cdot\scru}$. With these respective topologies, $\cost$ is proper, convex, and lower semicontinuous (by Lemma~\ref{lemma: Properties of Indirect Cost}). By definition of $\maxval$, it equals the convex conjugate of the indirect cost function; that is, $\maxval=\cost^*$. Hence, by the Fenchel-Moreau theorem \citep[e.g.,][Proposition 4.4.2]{borwein2010convex}, $\cost$ is the convex conjugate of $\maxval$; that is, $\cost = \maxval^*$. 

Now, we pursue the four-way equivalence. First, that (i) is equivalent to (ii) is immediate (see discussion after the statement of Lemma~\ref{lem: unique prediction with known benefits}). Next, \citeapos{aliprantis2006infinite} Theorem~7.16 directly implies (ii) is equivalent to (iii). Finally, to show (i) is equivalent to (iv), \citeapos{borwein2010convex} Proposition 4.4.1 part (a) delivers that $\scr \in \partial \maxval(\ut)$ holds if and only if
\[
\maxval(\ut) + \cost(\scr) = \bbE[\ut\cdot \scr],
\]
which is equivalent to
\[
\bbE[\ut \cdot \scr] - \cost(\scr) = \maxval(\ut) = \max_{\scrb \in \SCR} \left[\bbE[\ut\cdot \scrb] - \cost(\scrb)\right].
\]
It follows (iv) is equivalent to (i), as desired.
\end{proof}

The following lemma provides sufficient conditions for a set of SCRs to have its rationalizing utilities be an open set.

\begin{lemma}\label{lem: open set of utilities for weak* open set of uniquely rationalized SCRs}
Suppose $\SCRset\subseteq\SCR$ is weak* open, and every $\ut \in \UT$ and $\ut$-rationalizable $\scr\in \SCRset$ are such that $\scr$ is uniquely $\ut$-rationalizable. Then, the set of utilities that rationalize SCRs in $\SCRset$ is open in $\UT$. 
\end{lemma}
\begin{proof}
Recall $\ut\in\UT$ rationalizes $\scr\in\SCR$ if and only if $\scr \in \partial \maxval(\ut)$ (Lemma~\ref{lem: subdifferential characterization of optimality}). 
It follows the set of utilities that rationalize the SCRs in $\SCRset$ is given by $\partial\maxval^{L}(\SCRset)$.  
Moreover, because $\partial \maxval$ is norm-to-weak* upper hemicontinuous (Lemma~\ref{lem: value function continuity}) and $\SCRset$ is weak* open, we know $\partial\maxval^{U}(\SCRset)$ is norm open. Hence, the lemma will follow if we can establish that $\partial\maxval^{L}(\SCRset) = \partial\maxval^{U}(\SCRset)$.

Toward showing $\partial\maxval^{L}(\SCRset) = \partial\maxval^{U}(\SCRset)$, we use the fact that
$\ut$ rationalizes $\scr \in \SCRset$ if and only if $\scr$ is uniquely $\ut$-rationalizable. 
Therefore, $\partial\maxval(\ut) \cap\{\scr\} \neq \varnothing $ if and only if $\partial\maxval(\ut)\subseteq\{\scr\}$, meaning $\partial\maxval^{L}(\{\scr\})=\partial\maxval^{U}(\{\scr\})$ must hold for all $\scr\in \SCRset$. As such, 
\[
\partial\maxval^{L}(\SCRset) = \cup_{\scr \in \SCRset} \partial\maxval^{L}(\{\scr\}) = \cup_{\scr \in \SCRset} \partial\maxval^{U}(\{\scr\}) \subseteq \partial\maxval^{U}(\SCRset) \subseteq \partial\maxval^{L}(\SCRset),
\]
where the last containment follows from $\partial\maxval$ being nonempty-valued. It follows $\partial\maxval^{L}(\SCRset)=\partial\maxval^{U}(\SCRset)$.
\end{proof}

\subsection{Section~\ref{sec: knowing costs only} Proofs}

\begin{proof}[Proof of Proposition~\ref{prop: dense set of rationalizable SCRs}]
Recall from Lemma~\ref{lem: subdifferential characterization of optimality} that $\scr\in\SCR$ is rationalizable if and only if $\cost$ is subdifferentiable at $\scr$, that is, if and only if $\partial\cost(\scr)\neq\varnothing$. Also from that lemma, $\cost$ is the convex conjugate of the function $\maxval$ defined on the Banach space $\UT$. Moreover, $\maxval$ is proper, convex, and continuous (Lemma~\ref{lem: value function continuity}). Therefore, the dual version of the Br{\o}ndsted-Rockafellar theorem \citep[][second part of Theorem~2]{brondsted1965subdifferentiability} says $\cost$ is subdifferentiable on a norm-dense subset of its domain $\Dom$, as required.

The second point follows from the fact \citep[Theorem 23.4 in][]{rockafellar1970convex} that a proper convex function on a Euclidean space is subdifferentiable everywhere in the relative interior of its domain.
\end{proof}

The following lemma shows $\cost$ is ``almost'' strictly convex if no marginal information can be acquired for free.

\begin{lemma}\label{lem: cost strictly convex between non-proportional SCRs}
If $\icost$ is strictly monotone, and $\scr,\scrb\in\SCR$ and $a\in A$ are such that $\scr_a$ and $\scrb_a$ are not proportional, then $\cost$ is strictly convex on $\co\{\scr,\scrb\}$. 
\end{lemma}
\begin{proof}
Recall Lemma~\ref{lem: which policies can induce} tells us $\cost(\scrc)=\icost(\ip^\scrc)$ for every $\scrc\in\SCR$. 

Take any $\wt\in(0,1)$, and let $\scrc:=(1-\wt)\scr+\wt\scrb$ and $\ip:=(1-\wt)\ip^{\scr}+\wt\ip^{\scrb}$. 
Part~(iii) from Lemma~\ref{lem: SCR basic properties from IP basic properties} shows $\ip\succ\ip^\scrc$. 
Hence, 
$$\cost(\scrc)=\icost(\ip^\scrc)< \icost(\ip)\leq (1-\wt)\icost(\ip^{\scr})+\wt\icost(\ip^{\scrb})=(1-\wt)\cost(\scr)+\wt\cost(\scrb),$$
where the inequalities come from strict monotonicity and convexity of $\icost$, respectively.
\end{proof}

Next, we show that $\cost$ is strictly convex through any linearly independent SCR whenever $\icost$ is strictly monotone.

\begin{proposition}\label{prop: linear independence implies unique rationalizability}
Suppose $\icost$ is strictly monotone. If $\scr$ is linearly independent, then $\cost$ is strictly convex through $\scr$.\footnote{That is, $\cost$ is strictly convex on any line segment in $\Dom$ that includes $\scr$.} Consequently, if $\scr$ is $\ut$-rationalizable, it is uniquely $\ut$-rationalizable.
\end{proposition}

Proposition \ref{prop: linear independence implies unique rationalizability} identifies the set of linearly independent SCRs as ones that cannot be rationalized using indifference when $\icost$ is strictly monotone. This result generalizes a well-known condition for unique rationalizability under mutual information costs. Specifically, \cite{caplin2013behavioral} and \cite{matejka2015rational} show that, with mutual information costs, $\scr$ is uniquely $\ut$-rationalizable whenever the set
$
\big\{e^{\ut_{a}}: a \in A\big\}
$
consists of $|A|$ affinely independent elements. To see why this result is a specialization of Proposition~\ref{prop: linear independence implies unique rationalizability}, recall that \cite{matejka2015rational} show that,\footnote{See also \cite{Csiszar1974}.} with mutual information costs, $\scr$ is optimal when the utility is $\ut$ only if \[
\scr_{a}(\omega) = \frac{\ip^{\scr}_{a} e^{\ut_{a}(\omega)}}{\sum_{b\in A} \ip^{\scr}_{b} e^{\ut_{b}(\omega)}.}
\]
It is easy to verify that any $\scr$ satisfying the above display equation must be linearly independent whenever $\big\{e^{\ut_{a}}: a \in A\big\}$ consists of $|A|$ affinely independent elements. Unique rationalizability then follows from Proposition~\ref{prop: linear independence implies unique rationalizability}. 

We now turn to proving Proposition~\ref{prop: linear independence implies unique rationalizability}. 

\begin{proof}[Proof of Proposition~\ref{prop: linear independence implies unique rationalizability}]
The second assertion follows immediately from the first, because the expected benefit from a stochastic choice rule is an affine function of the stochastic choice rule, and a strictly concave function can have at most one maximizer. We turn now to the first assertion.

Given that $\cost$ is already known to be weakly convex by Lemma~\ref{lemma: Properties of Indirect Cost}, we need only show, given arbitrary $\scrb\in\SCR\setminus\{\scr\}$, that $\cost$ is strictly convex on $\co\{\scr,\scrb\}$. 

So suppose $\scrb\in\SCR$ is such that $\cost$ is not strictly convex on $\co\{\scr,\scrb\}$.  
Lemma~\ref{lem: cost strictly convex between non-proportional SCRs} then tells us $\scrb_a$ is a scalar multiple of $\scr_a$ (which is assumed to be nonzero) for every $a\in A$. Equivalently, $\posterior_a^{\scrb}=\posterior_a^\scr$ for every $a\in \supp({\scrb})$. Hence, 
$$\sum_{a\in A} \ip_a^{\scrb}\posterior_a^\scr= \sum_{a\in A} \ip_a^{\scrb}\posterior_a^{\scrb}= \prior = \sum_{a\in A} \ip_a^{\scr}\posterior_a^\scr.$$
Affine independence of the $|A|$ beliefs $\{\posterior_a^\scr\}_{a\in A}$ then implies $\ip_a^{\scrb}=\ip_a^{\scr}$ for every $a\in A$, so that $\scrb=\scr$, delivering the result.
\end{proof}

Although not relevant to our subsequent results, we briefly note a stronger uniqueness property (proven in Appendix~\ref{app: aux}) that follows readily for the special case of binary actions.

\begin{corollary}\label{cor: unique information choice for binary actions}
Suppose $|A|=2$ and $\icost$ is strictly monotone, and fix $\ut\in\UT$. Either a unique SCR is $\ut$-rationalizable, or every $\ut$-rationalizable SCR generates state-independent behavior. In particular, all optimal strategies entail the same information policy.
\end{corollary}

The following lemma formalizes a sense in which the affine independence case of Proposition~\ref{prop: linear independence implies unique rationalizability} is the typical case.

\begin{lemma}\label{lem: linearly independent SCRs are open and dense}
If $|\Omega|\geq|A|$, then $\supp(\ip^\scr)$ consists of $|A|$ affinely independent beliefs---that is, $\scr$ is linearly independent---for a weak*-open and uniformly dense set of $\scr\in\SCR$.
\end{lemma}
\begin{proof}
Define $\hat\SCRU$ to be the set of all $\scr\in\SCRU$ such that $\sum_{a\in A}\scr_a=\mathbf1$.

Let $\Omega=\bigsqcup_{a\in A}\Omega_a$ be a Borel partition into sets with positive $\prior$-measure; such a partition exists because $\Omega$ is metrizable with at least $|A|$ distinct points and $\prior$ has full support. Then, define the set $\hat M \subseteq \real^{A\times A}$ as the set of matrices for which the $a$ row's entries sum to $\prior(\Omega_a)$ for each $a\in A$, and the function $\pi:\hat\SCRU\to\hat M$ given by $\pi(\scr):=\left[ \int_{\Omega_a} \scr_{\tilde a}(\omega) \ \prior(\dd\omega) \right]_{a,\tilde a\in A}$. The map $\pi$ is affine, weak* continuous, and hence norm continuous, and, since $\{\Omega_a\}_{a\in A}$ are pairwise disjoint with nonzero measure, surjective. Because both $\hat\SCRU$ and $\hat M$ are closed affine subspaces of Banach spaces, it follows from the open mapping theorem that $\pi$ is a norm-open map: it maps norm-open sets to open sets. Hence, if $G$ is any open subset of $\hat M$ with closure equal to $\pi(\SCR)$, then $\pi^{-1}(G)$ is weak* open in $\hat\SCRU$ (because $\pi$ is weak* continuous) with norm closure containing $\SCR$ (because $\pi$ is norm open).\footnote{To see the latter implication, let $\SCRset$ be any norm-open set in $\hat\SCRU$ that intersects $\SCR$. 
Observe that because $\SCRset$ intersects $\SCR$, the set $\pi(\SCRset)$ intersects $\pi(\SCR)$ too. Note also that $\pi(\SCRset)$ is open in $\hat M$, because $\pi$ is norm open. Therefore, $\pi(\SCR)$ is contained in the norm closure of $G$, and so $G$ intersects $\pi(\SCRset)$. Said differently, $\SCRset$ intersects $\pi^{-1}(G)$, as required.}
It therefore suffices to find some open $G\subseteq\hat M$ with closure equal to $\pi(\SCR)$ such that any given $\scr\in\SCR\cap\pi^{-1}(G)$ has $\supp(\ip^\scr)$ consisting of $|A|$ affinely independent beliefs.

To that end, let $G$ be the convex set of invertible matrices in $\hat M$ with strictly positive entries. 
The set $G$ is open in $\hat M$ because the determinant function is continuous. Moreover, it is contained in $\pi(\SCR)$, which is the set of matrices in $\hat M$ with nonnegative entries. 
Next, to see $G$ is nonempty, define $g_\epsilon:= \left[\epsilon+\left(\prior(\Omega_a)-|A|\epsilon\right)\mathbf1_{a=\tilde a}\right]_{a,\tilde a\in A}\in\hat M$ for $\epsilon\in\left[0,\tfrac{\min_{a\in A}\prior(\Omega_a)}{|A|} \right]$. Because $g_0$ is diagonal with nonzero diagonal entries, its determinant is nonzero. Thus, since a nonzero univariate polynomial has only finitely many roots, $\det(g_\epsilon)$ is nonzero for $\epsilon>0$ sufficiently small.  The matrix $g:=g_\epsilon$ belongs to $G$ for such an $\epsilon$. Now, we observe $G$ is dense in $\pi(\SCR)$. Fixing an arbitrary $h\in\pi(\SCR)$, we want to show $h$ is a limit of matrices from $G$. Define the function $\real\to\real$ via $\gamma\mapsto\det\left[ \gamma g + (1-\gamma) h \right]$ and note it is a polynomial that is not globally zero (by evaluating at $\gamma=1$). Hence, the polynomial has only finitely many roots. In particular, some $\bar\gamma\in(0,1)$ exists such that $\gamma g + (1-\gamma) h\in\hat M$ is invertible for every $\gamma\in (0,\bar\gamma)$---and so $h$ is in the closure of $\hat M$. 

All that remains is to show, for any given $\scr\in\SCR\cap\pi^{-1}(G)$, that $\supp(\ip^\scr)$ consists of $|A|$ affinely independent beliefs. Toward showing this property, observe $\ip^\scr_{a}>0$ for every $a\in A$ because $\pi(\scr)$ has nonnegative entries and no zero columns.
We can now define the matrix $m:= \left[\tfrac1{\ip^\scr_{\tilde a}} \int_{\Omega_a} \scr_{\tilde a} \ddd\prior \right]_{a,\tilde a\in A} = \left[\posterior_{\tilde a}^\scr(\Omega_a) \right]_{a,\tilde a\in A}$. Because $m$ is a product of $\pi(\scr)$ and another invertible matrix, its columns are  linearly independent, and therefore affinely independent. It follows directly that the $|A|$ vectors $\{\posterior_a^\scr\}_{a\in A}$ are affinely independent, as required.
\end{proof}

Next, we prove the following claim about necessity of Assumption \ref{ass: regularity}\eqref{ass: regularity, part cardinality} for unique rationalizability when costs are locally affine.

\begin{claim}\label{claim: necessity of cardinality ranking for affine case}
Suppose $\scr \in \Dom$ is such that $|\supp\ \scr|>|\Omega|$. If $\icost$ is affine in a neighborhood of $\ip^{\scr}$, $\scr$ is not uniquely rationalizable.
\end{claim}

\begin{proof}[Proof of Claim~\ref{claim: necessity of cardinality ranking for affine case}]
Take $\ut\in\UT$ and $\scr\in\SCR$. Suppose $\scr$ is the unique stochastic choice rule induced by some optimal strategy given $\ut$. We must show $|\supp(\scr)|\leq|\Omega|$. The result is vacuous (given $|A|<\infty$) when $\Omega$ is infinite, so we focus on the case in which $\Omega$ is finite.

First, let us establish $\scr_a$ and $\scr_{\tilde a}$ cannot be proportional for any two distinct $a,\tilde a\in \supp(\scr)$. For a contradiction, assume they are in fact proportional. For any $\wt\in[0,1]$, then, we can define $\scr^\wt\in\SCR$ via $\scr^\wt_a:=(1-\wt)(\scr_a+\scr_{\tilde a})$, $\scr^\wt_{\tilde a}:=\wt(\scr_a+\scr_{\tilde a})$, and $\scr^\wt_{a'}:=\scr_{a'}$ for every $a'\in A\setminus\{a,\tilde a\}$. By construction, $\ip^{\scr^\wt}$ is the same for every $\wt\in[0,1]$, so that Lemma~\ref{lem: which policies can induce} implies $\cost(\scr^\wt)$ is the same for every $\wt\in[0,1]$. Therefore, the objective $\wt\mapsto \bbExp{\ut\cdot\scr^\wt}-\cost(\scr^\wt)$ is affine because $\wt\mapsto \scr^\wt$ is. It follows this objective cannot be uniquely maximized at $\wt^\scr=\tfrac{ \ip^\scr_{\tilde a}}{\ip^\scr_a + \ip^\scr_{\tilde a}}\in(0,1)$---contradicting unique optimality of $\scr$ because different values of $\wt$ generate different stochastic choice rules.

Having shown $\scr_a$ and $\scr_{\tilde a}$ cannot be proportional for any two distinct $a,\tilde a\in \supp(\scr)$, we know $\{\posterior^\scr_{a}\}_{a\in\supp(\scr)}$ are $|\supp(\scr)|$ distinct elements of $\DO$. No more than $|\Omega|$ affinely independent vectors can exist in a set (like $\DO$) with dimension $|\Omega|-1$, so the claim will follow if we show $\{\posterior^\scr_{a}\}_{a\in\supp(\scr)}$ are affinely independent.

Assume, for a contradiction,  $\supp(\ip^\scr)=\{\posterior^\scr_{a}\}_{a\in\supp(\scr)}$ are affinely dependent. By \citeapos{winkler1988extreme} Theorem 2.1(b), we then have $\ip^\scr\notin\ext(\Info)$. That is, $\ip^\scr$ is the midpoint of two distinct information policies $\ip^1,\ip^2\in\Info$. 
Moreover, replacing each of $\ip^1,\ip^2$ with a weighted average with $\ip^\scr$ if necessary, we may assume $\icost$ is affine on $\co\{\ip^1,\ip^2\}$. 
Optimality of $\scr$ implies 
the agent's optimal value is
$$\bbExp{\sum_{a\in \supp(\scr)}\ut_a\scr_a}  -\icost(\ip^\scr) \\
= \sum_{i=1,2} \tfrac12\left\{ \left[\sum_{a\in \supp(\scr)}\ip^i(\posterior_a^\scr)\int\ut_a\ddd\posterior^\scr_a \right] -\icost(\ip^i) \right\}.$$
Hence, some $i\in\{1,2\}$ has $v:=\left[\sum_{a\in \supp(\scr)}\ip^i(\posterior_a^\scr)\int\ut_a\ddd\posterior^\scr_a \right] -\icost(\ip^i)$ weakly higher than the agent's optimal value. Payoff $v$ is clearly attainable via some strategy that induces stochastic choice rule $\scr^i\in\SCR$ given by $\scr^i_a:=\tfrac{\ip^i(\posterior_a^\scr)}{\ip_a^\scr}\scr_a$ for $a\in \supp(\scr)$ and $\scr^i_a:=0$ for $a\in A\setminus\supp(\scr)$. Hence, $\scr^i$ is rationalizable, in contradiction to $\scr$ being uniquely so. The claim follows.
\end{proof}

\begin{proof}[Proof of Theorem~\ref{thm: dense set of uniquely rationalizable SCRs}]
Let $\SCRrat \subseteq \SCR$ be the set of all rationalizable SCRs if $\Omega$ is infinite, and let it be the relative interior of $\Dom$ (which is the interior of $\Dom$ in $\SCR$ by Assumption~\ref{ass: regularity}\eqref{ass: regularity, part domain}) if $\Omega$ is finite. By Proposition~\ref{prop: dense set of rationalizable SCRs}, $\SCRrat$ is norm dense (i.e., uniformly dense) in $\Dom$, and every SCR in $\SCRrat$ is rationalizable. 
So, $\SCRrat$ is a norm-dense subset of $\Dom$ consisting only of rationalizable SCRs, and $\SCRrat$ is open in $\SCR$ if $\Omega$ is finite. 

Now, Lemma~\ref{lem: linearly independent SCRs are open and dense} delivers a weak*-open (and hence norm-open) and norm-dense subset $G$ of $\SCR$ such that every $\scr\in G$ has $\supp(\ip^\scr)$ consisting of $|A|$ different affinely independent beliefs. By Proposition~\ref{prop: linear independence implies unique rationalizability}, every $\ut\in\UT$ and $\ut$-rationalizable $\scr \in \SCRrat \cap G$ are such that $\scr$ is uniquely $\ut$-rationalizable. Moreover, $\SCRrat\cap G$ is a norm-dense subset (being the intersection of two norm-dense subsets, one of which is open) of $\Dom$ that is also open in $\Dom$ whenever $\SCRrat$ is (which is true whenever $\Omega$ is finite).

It remains to show the set of utilities that rationalize the SCRs in $\SCRrat\cap G$ is open. But these utilities are the ones that rationalize SCRs in $G$ (resp. $\SCRrat\cap G$) if $\Omega$ is infinite (resp. finite)---a weak*-open set in either case---and so Lemma~\ref{lem: open set of utilities for weak* open set of uniquely rationalized SCRs} applies directly.
\end{proof}

\begin{proof}[Proof of Proposition~\ref{prop: continuous choice is impossible}]
Here, we prove the following three claims:
\begin{enumerate}[(i)]
\item Every rationalizable SCR is continuous if and only if every SCR in $\Dom$ is continuous. \label{prop: continuous choice is impossible: feasible is continuous}
\item If the domain $\Dom$ has a nonempty (norm) interior in $\SCR$, then every rationalizable SCR is continuous if and only if $\Omega$ is finite.\label{prop: continuous choice is impossible: interior means finite state}
\item If Assumption~\ref{ass: regularity} holds, then every uniquely rationalizable SCR is continuous if and only if $\Omega$ is finite.
\end{enumerate}
Observe (i) corresponds to the statement of Proposition~\ref{prop: continuous choice is impossible}. We now proceed with the proof.

First, observe the set $\Dom_{\text{cont}}$ of stochastic choice rules in $\Dom$ with a continuous version is $\Vert\cdot\Vert_\infty$-closed in $\Dom$. 
To see it is closed, note that because $\prior$ has full support, the natural quotient taking a continuous function to its $\prior$-a.e. equivalence class is an isometry from the space of of continuous functions $\Omega\to\real$ (with the supremum norm) into $\lp{\infty}$. Because the space of continuous functions is complete, its image under this isometry is necessarily closed in $\lp{\infty}$.\footnote{To explain, take any sequence $(\func_n)_n$ of continuous functions whose equivalence classes converge to some $\funcb$ (and so form a Cauchy sequence). The sequence $(\func_n)_n$ of functions themselves form a Cauchy sequence by isometry. But then, by completeness, $(\func_n)_n$ converges to some continuous function $\func$. Again by isometry, the equivalence class of $\func_n$ converges to that of $\func$ as $n\to\infty$. Because the metric space $\lp{\infty}$ is Hausdorff, $\funcb$ is in fact $\func$'s equivalence class, and so lives in the image of this isometry.} Hence, being the intersection of $\Dom$ with a power of this set, the set $\Dom_{\text{cont}}$ is closed in $\Dom$.

Now, because $\Dom_{\text{cont}}$ is closed, any given dense subset of $\Dom$ is contained in $\Dom_{\text{cont}}$ if and only if $\Dom_{\text{cont}}=\Dom$. In particular, Proposition~\ref{prop: dense set of rationalizable SCRs} (resp. Theorem~\ref{thm: dense set of uniquely rationalizable SCRs} under Assumption~\ref{ass: regularity}) implies every (uniquely) rationalizable stochastic choice rule lives in $\SCR^\cost_{\text{cont}}$ if and only if $\SCR^\cost_{\text{cont}}=\Dom$. This observation delivers the first part of the proposition. It would also deliver the second and third parts if we could show that, assuming $\Dom$ has a nonempty interior, every element of $\Dom$ has a continuous version if and only if $\Omega$ is finite. 

Suppose $\Dom$ has a nonempty interior. 
Our goal is now to show every element of $\Dom$ has a continuous version if and only if $\Omega$ is finite. One direction of this equivalence is trivial, because every map on a finite metric space is continuous. Toward establishing the converse, suppose $\Omega$ is infinite. Let us now observe some $\scr\in\Dom\setminus\Dom_{\text{cont}}$ would exist if we had some Borel $f:\Omega\to[0,1]$ that is not $\prior$-a.e. equal to a continuous function. Indeed, because $\Dom$ has a nonempty interior in $\SCR$, we know some $\scr^0\in\Dom$ and $\epsilon>0$ exist such that $\Dom \supseteq \left\{\scr\in\SCR:\ |\scr_a-\scr^0_a|\leq_{\prior\text{-a.e.}}\epsilon,\forall a\in A\right\}$. 
Shifting $\scr^0$ and reducing $\epsilon$ if necessary, we may further assume $\epsilon \leq_{\prior\text{-a.e.}} \scr_a \leq_{\prior\text{-a.e.}}1-\epsilon$ for every $a\in A$. Then, fix two distinct $a_+,a_-\in A$, and define $\scr^1\in\SCRU$ by letting $\scr^1_{a_+}:=\scr^0_{a_+}+\epsilon f$, $\scr^1_{a_-}:=\scr^0_{a_-}-\epsilon f$, and $\scr^1_{a}:=\scr^0_{a}$ for every $a\in A\setminus\{a_+, a_-\}$. By construction, $\scr^1\in\Dom$ too, and $\scr^1-\scr^0$ has no continuous version. Therefore, at least one of $\{\scr^0,\scr^1\}$ has no continuous version, as desired.

All that remains now is to construct (for infinite $\Omega$) some Borel $f:\Omega\to[0,1]$ that is not $\prior$-a.e. equal to a continuous function. Because $\Omega$ is an infinite compact metrizable space, it has some sequence $\{\omega_n\}_{n=1}^\infty$ without repetition that converges to some $\omega_\infty$. In particular, $\omega_\infty$ is the sole accumulation point of $\{\omega_m\}_{m\in\mathbb N \cup\{\infty\}}$.  Dropping to a subsequence if necessary, we may assume $\omega_n\neq\omega_\infty$ for every $n\in\mathbb N$. 
Now, take some $\rho$ that metrizes the topology on $\Omega$.
 For each $n\in\mathbb N$, let $r_n>0$ be such that the $\rho(\omega_n,\omega_m)>2 r_n$ for every $m\in\mathbb (\mathbb N\cup\{\infty\})\setminus\{n\}$, and let $N_n$ be the $\rho$-ball of radius $r_n$ centered on $\omega_n$. By the triangle inequality, the neighborhoods $\{N_n\}_{n=1}^\infty$ are pairwise disjoint. Moreover, that $\omega_n\to\omega_\infty$ implies $r_n\to0$ as $n\to\infty$. Define now the bounded Borel function $f:=\mathbf1_{\bigcup_{n=1}^\infty N_{2n}}:\Omega\to[0,1]$, which we show has no continuous version.  
To do so, take an arbitrary Borel function $\tilde f:\Omega\to\real$ that is $\prior$-a.e. equal to $f$. To conclude the proof, we show $\tilde f$ cannot be continuous. For any $n\in\mathbb N$, that $\prior$ has full support implies $\prior(N_n)>0$, and so some $\tilde\omega_n\in N_n$ has $\tilde f(\tilde\omega_n)=f(\tilde\omega_n)$. Then, $\rho(\tilde\omega_n,\omega_\infty)\leq r_n+\rho(\omega_n,\omega_\infty)\to0$ as $n\to\infty$. However, $\tilde f(\tilde\omega_n)=f(\tilde\omega_n)$ is equal to $0$ for odd $n$ and $1$ for even $n$, so that $\tilde f$ is discontinuous at $\omega_\infty$.
\end{proof}

\subsection{On Differentiable Costs}

We begin this subsection with some convenient notations. First, given $\ut\in \UT$, let
\begin{equation*}
\begin{split}
\posteriorut_{\ut}: \Delta\Omega & \rightarrow \mathbb{R} \\
\posterior & \mapsto \max_{a \in A} \int_{\Omega}\ut_{a}(\omega) \ \posterior(\dd\omega).
\end{split}
\end{equation*}

\begin{notation}
Let $\feas \ \icost$ be the (convex) set of $\ip \in \Infof\cap \Domicost$ such that every $\ipb \in \Infof$ admits some $\epsilon\in (0,1)$ with $\ip + \epsilon(\ipb - \ip) \in \Domicost$.

Analogously, for any $\icostd\in\icostD$, let $\feas \ \icostd$ be the (convex) set of simply drawn posteriors $\posterior$ such that every simply drawn posterior $\posteriorb$ admits some $\epsilon\in (0,1)$ with $\icostd\left(\posterior + \epsilon(\posteriorb - \posterior) \right)<\infty$.
\end{notation}

Note $\feas \ \icost=\Infof$ if $\icost$ assigns a finite cost to every simple information policy; and analogously, $\feas \ \icostd$ is the set of all simply drawn posteriors whenever $\icostd$ assigns a finite cost to every simply drawn posterior.

The following lemma gives an equivalent optimality condition for an SCR: the agent responds optimally to any revealed posterior, and, assuming the agent responds optimally to any hypothetical posterior,
the information is optimal.

\begin{lemma}\label{lem: characterizing optimal SCRs by optimal information and optimal interim action choice}
Given $\ut\in\UT$ and $\scr\in\SCR$, the following are equivalent:
\begin{enumerate}[(i)]
\item The stochastic choice rule $\scr$ is optimal; that is, $\scr \in \argmax_{\scrb \in \Dom} \left\{\bbExp{\ut\cdot\scrb} -\cost(\scrb) \right\}$.
\item Every $a\in \supp(\scr)$ has $\int \ut_{a}(\omega) \ \posterior^{\scr}_{a}(\dd\omega)=\posteriorut_{\ut}(\posterior^{\scr}_{a})$, and $$\ip^{\scr} \in \argmax_{\ip\in\Domicost} \left[\int \posteriorut_{\ut}(\posterior) \ \ip(\dd\posterior) - \icost(\ip) \right].$$
\end{enumerate}
\end{lemma}
\begin{proof}
Let $\alpha^{\ut}: \Delta\Omega\rightarrow A$ be a measurable selection of $\posterior\mapsto \argmax_{a\in A} \int \ut_{a}(\omega) \ \posterior(\dd\omega)$, which exists by the measurable maximum theorem \cite[][Theorem 18.19]{aliprantis2006infinite}.

To see (i) implies (ii), we suppose (ii) does not hold. Because $\posteriorut_{\ut}(\posterior) \geq \int \ut_{a}(\omega) \ \posterior(\dd\omega)$ for all $a\in A$ and $\posterior\in \Delta\Omega$, if $\posteriorut_{\ut}(\posterior^{\scr}_{a}) \neq \int \ut_{a}(\omega) \ \posterior^{\scr}_{a}(\dd\omega)$ for some $a\in \supp(\scr)$, it must be that $\posteriorut_{\ut}(\posterior^{\scr}_{a}) > \int \ut_{a}(\omega) \ \posterior^{\scr}_{a}(\dd\omega)$, and so the SCR $\scrb$ induced by $({\ip^{\scr}},\alpha^{\ut})$ is a strict improvement over $\scr$. Similarly, if $\ip^{\scr}\notin \argmax_{\ip\in\Domicost} \left[ \int \posteriorut_{\ut}(\posterior) \ \ip(\dd\posterior)- \icost(\ip)\right]$, one can pick $\ip\in \Domicost$ with $\int \posteriorut_{\ut}(\posterior) \ \ip(\dd\posterior) - \icost(\ip)>\int \posteriorut_{\ut}(\posterior) \ \ip^\scr(\dd\posterior) - \icost(\ip^\scr)$, in which case the SCR $\scrb$ induced by $(\ip,\alpha^{\ut})$ is strictly better than $\scr$. Either way, (i) fails.

To see (ii) implies (i), suppose (ii) holds. Then, any SCR $\scrb$ has 
\begin{eqnarray*}
\bbExp{\ut\cdot\scrb} -\cost(\scrb) &=& \bbExp{\ut\cdot\scrb} -\icost(\ip^{\scrb}) \text{ (by Lemma~\ref{lem: which policies can induce})} \\
&=& \left[\sum_{a\in\supp(\ip^{\scrb})} \ip^{\scrb}_a \int \ut_a(\omega) \ \posterior^{\scrb}_{a}(\dd\omega)\right] -\icost(\ip^{\scrb}) \\
&\leq& \left[\sum_{a\in\supp(\ip^{\scrb})} \ip^{\scrb}_a \posteriorut_\ut(\posterior^{\scrb}_{a})\right] -\icost(\ip^{\scrb}) \\
&=& \int \posteriorut_\ut(\posterior) \ \ip^{\scrb}(\dd\posterior) -\icost(\ip^{\scrb}),
\end{eqnarray*}
where the inequality holds with equality for $\scrb=\scr$ by hypothesis. Then, $\scr$ maximizes $\scrb\mapsto\bbExp{\ut\cdot\scrb} -\cost(\scrb)$ because $\ip^\scr$ maximizes $\ip\mapsto\int \posteriorut_\ut(\posterior) \ \ip(\dd\posterior) -\icost(\ip)$.
\end{proof}

Our definition of a derivative of $\icost$ assumed the derivative $\icostd$ was convex, which made posterior-separable approximation $\icost_\icostd$ a valid cost function by Jensen's inequality. Here, we note that under sufficient regularity, this convexity property is redundant, being implied by monotonicity of $\icost$.  See Appendix~\ref{app: aux} for the proof.

\begin{fact}\label{fact: convexity of derivative is wlog under continuity}
Let $\ip \in \feas \ \icost$, and let $\icostd:\DO\to\real$ be a continuous function. If every $\ip' \in \Domicost$ has
\[ 
d^{+}_{\ip}\icost(\ip')= \int \icostd(\posterior)\ (\ip' - \ip)(\dd\posterior),
\]
then $\icostd$ is convex. In particular, in this case, $\icostd\in\icostD$, so that $\icostd$ is a derivative of $\icost$ at $\ip$.
\end{fact}

The following lemma establishes an equivalence between optimal information choice for the agent's cost function and optimal information choice for a posterior-separable approximation of her cost function.
\begin{lemma}\label{lem: posterior separable approximation characterizes optimal information}
If $\icostd$ is a derivative of $\icost$ at $\ip \in \feas\ \icost$, then any $\ut\in\UT$ has
$$\ip \in \argmax_{\ipb \in \Info} \left[\int \posteriorut_\ut(\posterior) \ \ipb(\dd\posterior) - \icost(\ipb) \right] \ \iff\ \ip \in \argmax_{\ipb \in \Info} \int \left[\posteriorut_\ut(\posterior) - \icostd(\posterior) \right] \ \ipb(\dd\posterior).$$
\end{lemma}
\begin{proof}
Suppose first $\ip \in \argmax_{\ipb \in \Info} \int \left[\posteriorut_\ut(\posterior) - \icostd(\posterior)\right] \ \ipb(\dd\posterior).$ Then, for every $\ipb \in \Domicost$,
\[
\int \posteriorut_\ut(\posterior) \ (\ipb - \ip)(\dd\posterior) \leq \int \icostd(\posterior) \ (\ipb - \ip)(\dd\posterior) \leq \icost(\ipb) - \icost(\ip),\]
where the last inequality follows from $\icost\left(\ip+\epsilon(\ipb - \ip)\right)$ being convex in $\epsilon\in[0,1]$. Because $\icost(\ipb) = \infty$ for all $\ipb \in \Info \setminus \Domicost$, the left-hand-side condition follows. 

Conversely, suppose some $\tilde\ip\in \Info$ has $\int \left[\posteriorut_\ut(\posterior) - \icostd(\posterior)\right] \ \tilde\ip(\dd\posterior) > \int \left[\posteriorut_\ut(\posterior) - \icostd(\posterior)\right] \ \ip(\dd\posterior)$. Because $\icostd$ is convex and lower semicontinuous, and $\posteriorut_\ut$ is the maximum of finitely many affine functions, we may assume without loss that $\tilde\ip \in \Infof$. 
Observe every $\epsilon \in (0,1)$ satisfies
$$\int \left[\posteriorut_\ut(\posterior) - \icostd(\posterior)\right] \ \left[\ip+\epsilon(\tilde\ip- \ip)\right](\dd\posterior) > \int \left[\posteriorut_\ut(\posterior) - \icostd(\posterior)\right] \ \ip(\dd\posterior).$$
Thus, because $\ip \in \feas\ \icost$, we may assume without loss some convex combination $\ip'$ of $\tilde\ip$ and $\ip$ has $\ip^{'}\in \Domicost$. But then 
\begin{equation*}
\begin{split}
0 & < \int \posteriorut_{\ut}(\posterior) \ (\ip'- \ip)(\dd\posterior) - \int \icostd(\posterior) \ (\ip'-\ip)(\dd\posterior)
\\ & = \lim_{\epsilon \searrow 0}\frac{1}{\epsilon}\left\{\int \posteriorut_\ut(\posterior) \ \left[\ip+\epsilon(\ip'- \ip)\right](\dd\posterior) - \icost\left(\ip+\epsilon(\ip'-\ip)\right) - \int \posteriorut_{\ut}(\posterior) \ \ip(\dd\posterior) + \icost(\ip)\right\},
\end{split}
\end{equation*}
so that small enough $\epsilon\in(0,1)$ will satisfy
$$\int \posteriorut_\ut(\posterior) \ \left[\ip+\epsilon(\ip'- \ip)\right](\dd\posterior) - \icost\left(\ip+\epsilon(\ip'-\ip)\right) > \int \posteriorut_{\ut}(\posterior) \ \ip(\dd\posterior) - \icost(\ip).$$
The lemma follows.
\end{proof}

The next lemma establishes an equivalence between optimality of an SCR for our agent's cost function and its optimality for a posterior-separable approximation of the same.
\begin{lemma}\label{lem: posterior separable approximation characterizes optimal SCR}
Suppose $\scr\in \SCR$ has $\ip^{\scr} \in \feas\ \icost$ and $\icostd$ is a derivative of $\icost$ at $\ip^{\scr}$. Given $\ut\in\UT$, 
$$\scr \in \argmax_{\scr' \in \SCR} \left[\bbExp{\ut\cdot\scr'} -\cost(\scr') \right] \ \iff \ \scr \in \argmax_{\scr' \in \SCR} \left[\bbExp{\ut\cdot\scr'} -\cost_{\icostd}(\scr') \right].$$
\end{lemma}
\begin{proof}
Lemma~\ref{lem: characterizing optimal SCRs by optimal information and optimal interim action choice} says $\scr$ is optimal with cost $\cost$ if and only if $\int \ut_{a}(\omega) \ \posterior^{\scr}_{a}(\dd\omega)=\posteriorut_{\ut}(\posterior^{\scr}_{a})$ for all $a\in \supp(\scr)$ and $\ip^{\scr} \in \argmax_{\ip\in\Domicost} \left[\int \posteriorut_{\ut}(\omega) \ \ip(\dd\omega) - \icost(\ip)\right].$ 
Further, the information cost $\icost_\icostd$ satisfies our standing hypotheses on $\icost$. Lemma~\ref{lem: characterizing optimal SCRs by optimal information and optimal interim action choice} therefore tells us $\scr$ is optimal with cost $\cost_\icostd$ if and only if $\int \ut_{a}(\omega) \ \posterior^{\scr}_{a}(\dd\omega)=\posteriorut_{\ut}(\posterior^{\scr}_{a})$ for all $a\in \supp(\scr)$ and $\ip^{\scr} \in \argmax_{\ipb\in\Domicost} \left[\int \posteriorut_{\ut}(\posterior) \ \ipb(\dd\posterior) - \icost_\icostd(\ipb)\right].$ The equivalence would therefore follow if we knew $\ip^{\scr} \in \argmax_{\ip\in\Domicost} \left[\int \posteriorut_{\ut}(\posterior) \ \ip(\dd\posterior) - \icost(\ip)\right]$ if and only if 
\[
\ip^{\scr} \in \argmax_{\ip\in\Domicost} \left[\int \posteriorut_{\ut}(\posterior) \ \ip(\dd\posterior) - \icost_\icostd(\ip)\right],
\]
which Lemma~\ref{lem: posterior separable approximation characterizes optimal information} guarantees.
\end{proof}

The next lemma gives an explicit formula for the directional derivatives of the indirect cost function when information costs are posterior separable.

\begin{lemma}\label{lem: directional derivative of cost for differentiable case}
If $\icostd\in\icostD$ and 
$\scr,\scr'\in\SCR$ are such that every $\supp(\ip^\scr) \subseteq\feas\ \icostd$, then\footnote{
When $\ip^{\scr'}_{a}=0$, we adopt the convention that $\ip^{\scr'}_{a}\ d^{+}_{\posterior^{\scr}_{a}}\icostd(\posterior_{a}^{\scr'})$ is zero too.} 
\begin{eqnarray*}
d^{+}_{\scr} \cost_{\icostd} (\scr') &=& 
\sum_{a\in \supp(\scr)} \left[ (\ip^{\scr'}_{a} - \ip^{\scr}_{a})\ \icostd(\posterior_{a}^{\scr}) 
+ \ip^{\scr'}_{a} \ d^{+}_{\posterior^{\scr}_{a}}\icostd(\posterior_{a}^{\scr'})\right] \\
&&+ \sum_{a\in \supp(\scr')\setminus\supp(\scr)}  \ip^{\scr'}_{a} \ \icostd(\posterior_{a}^{\scr'}).
\end{eqnarray*}
\end{lemma}
\begin{proof}
Let $\hat A:=\supp(\scr) \cup \supp(\scr')$, and define $\scr^\epsilon:=\scr+ \epsilon(\scr'-\scr)$ for each $\epsilon \in (0,1)$. Any $\epsilon\in(0,1)$ has $\supp(\scr^\epsilon)=\hat A$ and, for each $a\in \hat A$, \begin{eqnarray*}
\ip_a^{\scr^\epsilon} &=& \ip_a^{\scr} + \epsilon(\ip_a^{\scr'}-\ip_a^{\scr}) \\
\posterior_a^{\scr^\epsilon} &=& \posterior_a^{\scr} + \epsilon\tfrac{\ip_a^{\scr'}}{\ip_a^{\scr^\epsilon}}(\posterior_a^{\scr'}-\posterior_a^{\scr}).
\end{eqnarray*}
Therefore, every $\epsilon\in(0,1)$ and $a\in\hat A$ have
$$\eta_a(\epsilon):=\tfrac1\epsilon\left[ \ip^{\scr^\epsilon}_a \ \icostd(\posterior^{\scr^\epsilon}_a) - \ip^{\scr}_a \ \icostd(\posterior^{\scr}_a) \right] 
= (\ip^{\scr'}_a-\ip^{\scr}_a) \ \icostd(\posterior^{\scr^\epsilon}_a) + \ip^{\scr}_a\ \tfrac1{\epsilon} \left[ \icostd(\posterior^{\scr^\epsilon}_a)- \icostd(\posterior^{\scr}_a) \right].$$
Hence, $\eta_a(\epsilon)$ is equal to $\ip^{\scr'}_a \icostd(\posterior^{\scr'}_a)$ if $\ip^\scr_a=0$, is equal to $-\ip^{\scr}_a \icostd(\posterior^{\scr}_a)$ if $\ip^{\scr'}_a=0$, 
and is otherwise equal to
$$
(\ip^{\scr'}_a-\ip^{\scr}_a) \ \icostd(\posterior^{\scr^\epsilon}_a) + \ip^{\scr'}_a\ \dfrac{\ip^{\scr}_a}{\ip^{\scr^\epsilon}_a}\ \dfrac1{\epsilon\tfrac{\ip^{\scr'}_a}{\ip^{\scr^\epsilon}_a}} 
\left[ \icostd(\posterior^{\scr^\epsilon}_a)- \icostd(\posterior^{\scr}_a) \right],
$$
which converges as $\epsilon\to0$ (because $\icostd$ is convex, and hence continuous on any open line segment on its domain, and $\supp(\ip^\scr)\subseteq\feas\ \icostd$) to \begin{eqnarray*}
&&(\ip^{\scr'}_a-\ip^{\scr}_a) \ \icostd\left(\lim_{\epsilon\searrow0}\posterior^{\scr^\epsilon}_a\right) + \ip^{\scr'}_a\ \dfrac{\ip^{\scr}_a}{\lim_{\epsilon\searrow0}\ip^{\scr^\epsilon}_a}\ \lim_{\tilde\epsilon\searrow0}
\dfrac1{\tilde\epsilon} 
\left[ 
\icostd(\posterior^{{\scr} + \tilde\epsilon({\scr'}-{\scr})})- \icostd(\posterior^{\scr}_a) 
\right]\\
&=&(\ip^{\scr'}_{a} - \ip^{\scr}_{a})\ \icostd(\posterior_{a}^{\scr}) 
+ \ip^{\scr'}_{a}\ 1 \ d^{+}_{\posterior^{\scr}_{a}}\icostd(\posterior_{a}^{\scr'}).
\end{eqnarray*}
Hence, \begin{eqnarray*}
\lim_{\epsilon\searrow 0} \ \tfrac{1}{\epsilon}\left[\cost_{\icostd}(\scr + \epsilon(\scr'-\scr)) - \cost_{\icostd}(\scr)\right] 
&=& \lim_{\epsilon\searrow 0}\tfrac{1}{\epsilon}\int \icostd \ddd (\ip^{\scr^\epsilon} - \ip^{\scr}) =\lim_{\epsilon\searrow 0}\sum_{a\in\hat A} \eta_a(\epsilon) \\
&=& \sum_{a\in \supp(\scr)} \left[ (\ip^{\scr'}_{a} - \ip^{\scr}_{a})\ \icostd(\posterior_{a}^{\scr}) 
+ \ip^{\scr'}_{a} \ d^{+}_{\posterior^{\scr}_{a}}\icostd(\posterior_{a}^{\scr'})\right] \\
&&+ \sum_{a\in \supp(\scr') \setminus \supp(\scr)}  \ip^{\scr'}_{a} \ \icostd(\posterior_{a}^{\scr'}),
\end{eqnarray*}
as desired.
\end{proof}

The following lemma shows derivatives preserve the finite-cost property of any information policy.

\begin{lemma}\label{lem: derivatives preserve finite cost}
If $\icostd$ is a derivative of $\icost$ at $\ip\in\Domicost$, then $\int\icostd(\posterior)\ \ip'(\dd\posterior)<\infty$ for every $\ip'\in\Domicost$.
\end{lemma}
\begin{proof}
Consider any $\ip'\in\Domicost$. Convexity of $\icost$ implies $d^{+}_{\ip}\icost(\ip')\in\real\cup\{-\infty\}$. Hence, applying the definition of a derivative (including that $\int \icostd(\posterior)\ \ip(\dd\posterior)<\infty$), 
$$\int \icostd(\posterior)\ \ip'(\dd\posterior) = \int \icostd(\posterior)\ \ip(\dd\posterior)+ d^{+}_{\ip}\icost(\ip')<\infty,$$
as required.
\end{proof}

\subsection{On Iteratively Differentiable Costs}\label{app: On Iterative Differentiability}

Although not relevant to our subsequent results, let us briefly note a uniqueness property (proven in Appendix~\ref{app: aux}) for iterated differentiability that justifies the notation $\icostdd_{\posterior}$.

\begin{fact}\label{fact: derivative of derivative is unique}
Any function in $\icostD$ admits at most one derivative at a simply drawn posterior.
\end{fact}

The following lemma yields a more explicit form (relative to Lemma~\ref{lem: directional derivative of cost for differentiable case}) for the directional derivative of a posterior-separable approximation of the cost function, in the iteratively differentiable case. 

To state a slightly more general version of the result (which will be helpful for one result described in section~\ref{sec: discussion}), we invest in another definition. 
For any $\icostd\in\icostD$ and any simply drawn posterior $\posterior$, define the \textbf{subdifferential} 
$$\partial\icostd(\posterior):=\left\{\func\in\lp1:\ \icostd(\posterior')\geq\icostd(\posterior)+\int \func(\omega) \ (\posterior'-\posterior)(\dd\omega) \ \forall \text{ simply drawn } \posterior'\in\DO\right\}$$
of $\icostd$ at $\posterior$. Clearly, if $\icostd$ is differentiable at $\posterior$, then $\icostdd_\mu\in\partial\icostd(\posterior)$ because $\icostd$ is convex. 

\begin{remark}
Although we adopt the notation $\partial\icostd$ for parsimony, the above definition is best understood as the subdifferential of a function $\tilde\icostd:\lp\infty\to\extreal$, where we identify each element of $\lp\infty$ with the measure on $\Omega$ whose Radon-Nikodym derivative with respect to $\prior$ is given by that element.
\end{remark}

\begin{lemma}\label{lem: directional derivative of cost for iteratively differentiable case, and subgradient result for differentiable case}
Let $\icostd\in\icostD$ and $\scr,\scr'\in\SCR$ have $\supp(\ip^\scr) \subseteq\feas\ \icostd$, and let $\func_a\in\partial\icostd(\posterior^\scr_a)$ have $\int\func_a(\omega)\ \posterior^\scr_a(\dd\omega)$
for each $a\in\supp(\scr)$. 
Let $\ut^{\scr,\scr'}\in\UT$ have $\ut^{\scr,\scr'}_a=\func_a$ for $a\in\supp(\scr)$ and 
$\ut^{\scr,\scr'}_a=\icostd(\posterior_{a}^{\scr'})\mathbf1$ for $a\in \supp(\scr')\setminus\supp(\scr)$. Then, 
$$d^{+}_{\scr} \cost_{\icostd} (\scr')\geq\bbExp{(\scr'-\scr)\cdot\ut},$$
with equality holding if $\func_a=\icostdd_{\posterior^\scr_a}$ for every $a\in\supp(\scr)$.
\end{lemma}
\begin{proof}
Let $\ut:=\ut^{\scr,\scr'}$. Then, every $a\in \supp(\scr)$ has \[
\begin{split}
\ip^{\scr'}_{a}\ d^{+}_{\posterior^{\scr}_{a}}\icostd(\posterior_{a}^{\scr'})
& \geq\ip^{\scr'}_{a}\int \func_a(\omega) \ (\posterior_{a}^{\scr'}-\posterior^{\scr}_{a})(\dd\omega)
\\ & =\ip^{\scr'}_{a}\int \left[\frac{\scr'_{a}(\omega)}{\ip^{\scr'}_{a}} - \frac{\scr_{a}(\omega)}{\ip^{\scr}_{a}} \right] \func_a(\omega)  \ \prior(\dd\omega) 
\\ & =\int \left[\scr'_{a}(\omega) - \scr_{a}(\omega) \right] \func_a(\omega) \ \prior(\dd\omega)  
+ \int \left({\ip^{\scr}_{a} - {\ip^{\scr'}_{a}}}\right)\frac{\scr_{a}(\omega)}{\ip^{\scr}_{a}}
\func_a(\omega) \ \prior(\dd\omega)  
\\ & = \int \left[\scr'_{a}(\omega) - \scr_{a}(\omega) \right] \ut(\omega) \ \prior(\dd\omega)  
- \left(\ip^{\scr'}_{a} - \ip^{\scr}_{a}\right)\int \func_a(\omega)\ \posterior_{a}^{\scr}(\dd\omega) 
\\ & = \bbExp{ (\scr'_{a} - \scr_{a} )\ut_{a}} - (\ip^{\scr'}_{a} - \ip^{\scr}_{a} ) \ \icostd(\posterior_{a}^{\scr}).
\end{split}
\]
Moreover, the above inequality holds with equality if $\func_a=\icostdd_{\posterior^\scr_a}$ for every $a\in\supp(\scr)$. 
Therefore, Lemma~\ref{lem: directional derivative of cost for differentiable case} implies \begin{eqnarray*}
d^{+}_{\scr} \cost_{\icostd} (\scr') 
&\geq& \sum_{a\in \supp(\scr)} \bbExp{ (\scr'_{a} - \scr_{a} )\ut_{a}}+\sum_{a\in \supp(\scr')\setminus\supp(\scr)}  \ip^{\scr'}_{a} \ \icostd(\posterior_{a}^{\scr'}) \\
&=& \sum_{a\in \supp(\scr)} \bbExp{ (\scr'_{a} - \scr_{a} )\ut_{a}}+\sum_{a\in \supp(\scr')\setminus\supp(\scr)}  \bbExp{(\scr'_{a} - \scr_{a}) \ \icostd(\posterior_{a}^{\scr'})} \\
&=& \bbExp{(\scr'-\scr)\cdot\ut},
\end{eqnarray*}
again with equality if $\func_a=\icostdd_{\posterior^\scr_a}$ for every $a\in\supp(\scr)$. The result follows.
\end{proof}

The next lemma relates the set $\feas\ \icost$ to the corresponding set of beliefs for its derivatives.
\begin{lemma}\label{lem: feas for C to feas for c}
If $\icostd$ is a derivative of $\icost$ at $\ip\in\feas\ \icost$, then $\supp(\ip)\subseteq\feas\ \icostd$.
\end{lemma}
\begin{proof}
Consider $\posterior\in\supp(\ip)$, and any simply drawn posterior $\posterior'$; we must show some proper convex combination of $\posterior$ and $\posterior'$ belongs to $\icostd^{-1}(\real)$. 

By hypothesis, $\ip$ is simple, $\posterior\in\supp(\ip)$, and $\posterior'$ is simply drawn. Hence, some finite-support $\ipb,\ipb'\in\Delta\DO$ and $\wt,\wt'\in(0,1]$ exist with $\ip=(1-\wt)\ipb+\wt\delta_{\posterior}$ and $\ip':=(1-\wt')\ipb'+\wt'\delta_{\posterior'}\in\Domicost$. Now, that $\ip\in\feas\ \icost$ implies some $\epsilon\in(0,1)$ has $(1-\epsilon)\ip+\epsilon\ip'\in\Domicost$. Define $$\posterior^\epsilon:=\tfrac{(1-\epsilon)\wt}{(1-\epsilon)\wt+\epsilon\wt'}\posterior+\tfrac{\epsilon\wt'}{(1-\epsilon)\wt+\epsilon\wt'}\posterior',$$ and 
$$\ip^\epsilon:=(1-\epsilon)(1-\wt)\ipb+\epsilon(1-\wt')\ipb'+ \left[ (1-\epsilon)\wt+\epsilon\wt' \right]\delta_{\posterior^\epsilon}.
$$
Because $\icost$ is monotone and $\ip^\epsilon\preceq (1-\epsilon)\ip+\epsilon\ip'$ by construction, $\ip^\epsilon\in\Domicost$. Hence, Lemma~\ref{lem: derivatives preserve finite cost} tells us $\int\icostd(\posterior)\ \ip^\epsilon(\dd\posterior)<\infty$, and so too $\icostd(\posterior^\epsilon)<\infty$. The lemma follows.
\end{proof}

The following lemma gives an exact optimality condition for a given SCR when information costs are iteratively differentiable.
\begin{lemma}\label{lem: FOC for iteratively differentiable case}
Let $\scr\in\SCR$ have $\ip^{\scr} \in \feas\ \icost$, suppose $\icost$ is iteratively differentiable at $\ip^{\scr}$ with derivative $\icostd$. For $\ut\in\UT$, the following are equivalent: \begin{enumerate}[(i)]
    \item SCR $\scr$ is $\ut$-rationalizable.
    \item Every $\scr'\in\SCR$ has $\bbExp{(\ut^{\scr,\scr'}-\ut)\cdot(\scr'-\scr)}\geq0$, where $\ut^{\scr,\scr'}\in\UT$ has $\ut^{\scr,\scr'}_a=\icostdd_{\posterior^{\scr}_{a}}$ for $a\in\supp(\scr)$ and $\ut^{\scr,\scr'}_a=\icostd(\posterior_{a}^{\scr'})\mathbf1$ for $a\in \supp(\scr')\setminus\supp(\scr)$
\end{enumerate}
\end{lemma}
\begin{proof}
By Lemma~\ref{lem: posterior separable approximation characterizes optimal SCR}, $\ut$ rationalizes $\scr$ if and only if $\ut$ would rationalize $\scr$ given alternative information cost $\icost_{\icostd}$.
Hence, Lemma~\ref{lem: subdifferential characterization of optimality} (applied to the model with cost $\icost_\icostd$) tells us $\ut$ rationalizes $\scr$ if and only if every $\scr' \in \Dom$ has $d_{\scr}^{+}\cost_\icostd (\scr')\geq\bbExp{\ut\cdot(\scr' - \scr)}$. 
Because Lemma~\ref{lem: feas for C to feas for c} tells us $\supp(\ip^\scr) \subseteq\feas\ \icostd$, Lemma~\ref{lem: directional derivative of cost for iteratively differentiable case, and subgradient result for differentiable case} shows $d^{+}_{\scr} \cost_{\icostd} (\scr')=\bbExp{\ut^{\scr,\scr'}\cdot(\scr'-\scr)}$ for every $\scr'\in\SCR$. 
The equivalence follows.
\end{proof}

The following lemma characterizes the utilities that can rationalize a given full-support SCR when information costs are iteratively differentiable. Letting $\icostd$ denote a derivative of $\icost$ at the SCR, all such utilities are given by the derivative of $\icostd$ at the corresponding revealed posterior, augmented by a nuisance term, potentially with an additional payoff penalty for actions not chosen in a given state.

\begin{lemma}\label{lem: multiplier result for full-support iteratively differentiable case}
Let $\scr\in\SCR$ have $\ip^{\scr} \in \feas\ \icost$ and $\supp(\scr)=A$. 
Suppose $\icost$ is iteratively differentiable at $\ip^{\scr}$ with derivative $\icostd$.  
The following are equivalent for $\ut\in\UT$:
\begin{enumerate}[(i)]
\item SCR $\scr$ is $\ut$-rationalizable.
\item Some $\nuis\in \lp{1}$ and $\mult\in\lp{1}_+^A$ exist such that every $a\in A$ has \begin{eqnarray*}
\ut_a &=& \nuis-\mult_a+\icostdd_{\posterior^{\scr}_{a}}, \\
\scr_a\mult_a &=& 0.
\end{eqnarray*}
\end{enumerate}
\end{lemma}
\begin{proof}
Let $\ut^\scr:= \ut- (\icostdd_{\posterior^{\scr}_{a}})_{a\in A} \in \UT$. By Lemma~\ref{lem: FOC for iteratively differentiable case}, $\ut$ rationalizes $\scr$ if and only if every $\scr'\in\SCR$ has $\bbExp{-\ut^\scr\cdot(\scr'-\scr)}\geq0$, or equivalently, $\scr\in\argmax_{\scr'\in\SCR} \bbExp{\ut^{\scr}\cdot\scr'}$. Hence, $\ut$ rationalizes $\scr$ if and only if $\prior$-almost every $\omega$ has $$\{a\in A:\ \scr_a(\omega)>0\}\subseteq\argmax_{a\in A} \ut^\scr_a(\omega).$$
But the latter condition holds if and only if some $\nuis\in\lp{1}$ and $\mult\in\lp{1}_+^A$ have $\ut^\scr=(\nuis-
\mult_a)_{a\in A}$ and $(\mult_a\scr_a)_{a\in A}=0$: The ``if'' direction is immediate, and the ``only if'' direction comes from considering $\nuis:=\max_{a\in A}\ut^s_a$. The lemma follows.
\end{proof}

\subsection{Section~\ref{sec: cross-menu} Proofs}

\begin{proof}[Proof of Lemma~\ref{lem: posterior separable approximation characterizes optimality}]
Because $\icost$ is finite on $\Infof$, every SCR $\scr$ has $\ip^\scr\in\feas \ \icost$. The theorem therefore follows directly from Lemma~\ref{lem: posterior separable approximation characterizes optimal SCR}.
\end{proof}

\begin{proof}[Proof of Proposition~\ref{prop: multiplier result for iteratively differentiable case}]
Because $\icost$ is finite on $\Infof$, every SCR $\scr$ has $\ip^\scr\in\feas \ \icost$. The proposition therefore follows directly from Lemma~\ref{lem: multiplier result for full-support iteratively differentiable case}.
\end{proof}

Let us briefly note the infinite-slope condition, Assumption~\ref{ass: smoothness}\eqref{ass: smoothness, part inada}---which says costs decrease infinitely steeply as one moves from a not-fully-mixed information policy toward providing no information---could be equivalently replaced with a more permissive unbounded steepness condition. See Appendix~\ref{app: aux} for the proof.

\begin{fact}\label{fact: equivalent inada conditions}
For any $\ip\in\Domicost$, the following are equivalent:
\begin{enumerate}[(i)]
\item $d^{+}_{\ip}\icost(\delta_\prior) = -\infty.$
\item $\inf_{\ipb\in\Domicost,\ \epsilon\in(0,1)}\frac{1}{\epsilon} \left[ \icost(\ip + \epsilon(\ipb - \ip)) - \icost(\ip) \right] = -\infty.$
\end{enumerate}
\end{fact}

The following corollary follows immediately from Lemma~\ref{lem: multiplier result for full-support iteratively differentiable case} (and the weaker version requiring $\icost$ to be finite on $\Infof$ follows immediately from Proposition~\ref{prop: multiplier result for iteratively differentiable case}).

\begin{corollary}\label{cor: multiplier result for conditionally full-support, iteratively differentiable case}
Let $\scr\in\SCR$ have $\ip^{\scr} \in \feas \ \icost$ and $\scr$ have conditionally full support. If $\icost$ is iteratively differentiable at $\ip^\scr$, and $\ut$ and $\ut'$ both rationalize $\scr$, some $\nuis \in \lp{1}$ exists such that 
\[
\ut_{a} = \ut'_{a} + \nuis \ \forall a\in A.
\]
\end{corollary}

Next, we show Assumption~\ref{ass: smoothness}\eqref{ass: smoothness, part inada} means an optimal action recommendation from $\scr$ never rules out any state.

\begin{lemma}\label{lem: inada and full support implies conditionally full support}
Suppose $\icost$ satisfies Assumption~\ref{ass: smoothness}\eqref{ass: smoothness, part inada}. If $\scr$ is rationalizable, then any $a \in \supp \ \scr$ has $\scr_{a}$ strictly positive $\prior$-almost surely.
\end{lemma}
\begin{proof}
Take $\ut\in\UT$, and suppose $\scr\in\SCR$ admits some $a \in \supp \ \scr$ such that $\scr_{a}$ is not $\prior$-almost surely strictly positive. We wish to show $\scr$ is not $\ut$-rationalizable. 

We have nothing to show if $\scr\notin\Dom$, so we focus on the case in which $\scr \in \Dom$. Pick some $a_{0} \in A$ and let $\scrb$ be the unique SCR with $\scrb_{a_{0}} = \mathbf{1}$. Clearly, $\ip^{\scrb} = \delta_{\prior}$, and so $\scrb \in \Dom$ and $\ip^{\scrb} \in \Domicost$. 
For every $\epsilon\in (0,1)$, let $\scrb^{\epsilon}= \scr + \epsilon(\scrb - \scr)$. Then, because $\ip^\scr$ is not fully mixed,
\[
\begin{split}
\frac{1}{\epsilon}[\cost(\scrb^{\epsilon}) - \cost(\scr)]
& = \frac{1}{\epsilon}[\icost(\ip^{\scrb^{\epsilon}}) - \icost(\ip^\scr)]
\\ & \leq \frac{1}{\epsilon}[\icost((1-\epsilon)\ip^{\scr} + \epsilon\ip^{\scrb}) - \icost(\ip^\scr)]
\\ & = \frac{1}{\epsilon}[\icost((1-\epsilon)\ip^{\scr} + \epsilon\delta_{\prior}) - \icost(\ip^\scr)]
\xrightarrow{\epsilon\searrow 0} -\infty,
\end{split}
\]
where the inequality comes from monotonicity of $\icost$ and Lemma~\ref{lem: SCR basic properties from IP basic properties}\eqref{lem: SCR basic properties from IP basic properties-convexity}, and the limit calculation comes from Assumption~\ref{ass: smoothness}\eqref{ass: smoothness, part inada}. Hence, every $\ut \in \UT$ has
\[
\frac{1}{\epsilon}\left\{ \bbExp{\ut \cdot \scrb^{\epsilon}} - \cost(\scrb^{\epsilon}) -\left[\bbExp{\ut \cdot \scr} - \cost(\scr) \right]  \right\}
= \bbExp{\ut \cdot (\scrb - \scr)} - \frac{1}{\epsilon}\left[\cost(\scrb^{\epsilon}) - \cost(\scr) \right] \xrightarrow{\epsilon\searrow 0}-\infty.
\]
Thus, for all sufficiently small $\epsilon\in(0,1)$, the agent's objective must be strictly higher under $\scrb^{\epsilon}$ than under $\scr$. That is, $\scr$ is not rationalizable.
\end{proof}

Now, we prove most utilities give rise a unique SCR.

\begin{proof}[Proof of Lemma~\ref{lem: unique prediction with known benefits}]
Let $\UTB$ denote the set of $\ut\in\UT$ at which $\maxval$ is G\^ateaux differentiable, and recall $\maxval$ is continuous by Lemma~\ref{lem: value function continuity}. For any $\ut\in\UTB$, the function $\maxval$ has a unique subgradient at $\ut$ \citep[][Corollary 4.2.5]{borwein2010convex}, and so a unique SCR is $\ut$-rationalizable.

It therefore remains to show $\UT\setminus\UTB$ is meager and shy. That this set is meager follows from Mazur's theorem \citep[see][Theorem 4.6.3]{borwein2010convex}, which tells us $\UTB$ is dense and $G_\delta$, and hence co-meager. To see $\UT\setminus\UTB$ is shy, it suffices to show it is Haar null.\footnote{Given Fact 2 from \cite{hunt1992prevalence} \citep[resp. Proposition 4.6.1(c) from][]{borwein2010convex}, $\UT\setminus\UTB$ is shy (resp. Haar null) if and only if some compactly supported finite Borel (resp. Radon) measure assigns zero measure to every translation of it. Because every finite Borel measure on the Polish space $\UT$ is a Radon measure, $\UT\setminus\UTB$ is shy if and only if it is Haar null.}

To show $\UT\setminus\UTB$ is Haar null, note Theorem 5.44 from \cite{aliprantis2006infinite} tells us $\maxval$ is locally Lipschitz. Hence, being second countable, $\UT$ can be covered by countably many open balls $\{\UT_n\}_{n=1}^\infty$ with $\maxval|_{\UT_n}$ Lipschitz for each $n\in\mathbb N$. Moreover, for each $n\in\mathbb N$, Theorem 4.6.5 from \cite{borwein2010convex} tells us $\UT_n\setminus\UTB$ is Aronszajn null, and hence \citep[by][Proposition 6.25]{benyamini1998geometric} Haar null. Therefore, $\UT\setminus\UTB=\bigcup_{n=1}^\infty(\UT_n\setminus\UTB)$ is Haar null \citep[by][Proposition 4.6.1(e)]{borwein2010convex}.
\end{proof}

The following lemma shows a norm-dense set of uniquely rationalizable SCRs can be leveraged to find a weak*-dense set of SCRs generating unique subset predictions.

\begin{lemma}\label{lem: dense unique subset predictions from dense uniquely rationalizable}
Suppose Assumption~\ref{ass: smoothness} holds and some norm-dense subset $\SCR_{1}$ of $\Dom$ exists that comprises only uniquely rationalizable SCRs. Then, the set of SCRs yielding unique subset predictions is weak* dense in $\SCR$.
\end{lemma}
\begin{proof}
By Lemma~\ref{lem: unique prediction with known benefits}, every $\menu\in \Menu$ admits a co-meager set of utilities $\UT_{\menu}\subseteq \UT$ that generate a unique prediction over $\menu$; that is, each $\ut \in \UT_{\menu}$ admits a unique $\scr$ that is $\ut$-rationalizable over $\menu$. It follows the set $\UT_{\Menu} = \cap_{\menu\in\Menu}\UT_{\menu}$ is a co-meager (and therefore dense) subset of $\UT$ such that, for every $\menu$, every $\ut \in \UT_{\Menu}$ uniquely rationalizes some $\scr\in \SCR^\menu$ over $\menu$.

Observe now, the set $\SCR_{1}$ is norm-dense in $\SCR$ because Assumption~\ref{ass: smoothness}\eqref{ass: smoothness, part finite} implies $\Dom=\SCR$. Let $\SCR_0$ denote the set of full-support SCRs, which is weak* open (hence norm open) and norm dense in $\SCR$. Therefore, $\tilde{\SCR}:=\SCR_0\cap\SCR_1$ is norm dense in $\SCR$ too.  
Because every norm-dense set is weak* dense, it suffices to show $\tilde{\SCR}$ is contained in the weak* closure of the SCRs generating unique subset predictions.
To that end, take any $\scr \in \tilde{\SCR}$, and let $\neigh \subseteq \SCRU$ be a weak* neighborhood of $\scr$; we want to show some SCR in $\neigh$ generates unique subset predictions. 

Let $\ut$ be a utility that uniquely rationalizes $\scr$ (over $A$). Some sequence $(\ut^{n})_{n \in \mathbb{N}}$ from $\UT_{\Menu}$ converges to $\ut$ because $\UT_{\Menu}$ is dense in $\UT$. Let $\scr^{n}$ denote the SCR rationalized by $\ut^{n}$ for each $n\in\mathbb N$. Because $\scr^n\in \partial \maxval(\ut^{n})$ by Lemma~\ref{lem: subdifferential characterization of optimality}, and $\partial \maxval$ norm-to-weak* upper hemicontinuous (Lemma~\ref{lem: value function continuity}) and single valued at $\ut$, it follows that $\scr^{n} \xrightarrow{w^*} \scr$ as $n\to\infty$. In particular, $\scr^{n} \in \neigh$ for all $n$ sufficiently large.

Next, observe $\scr^{n}$ has full support (i.e., belongs to $\SCR_0$) for large enough $n$, because $\SCR_0$ is weak* open in $\SCR$, and the weak*-limit $\scr$ has full support. But Lemma~\ref{lem: inada and full support implies conditionally full support} tells us any rationalizable full-support SCR has conditionally full support in light of Assumption~\ref{ass: smoothness}\eqref{ass: smoothness, part inada}. Therefore, $\scr^{n}$ has conditionally full support for sufficiently large $n$.

Now, fix $n\in\mathbb N$ large enough that $\scr^n\in\neigh$ and $\scr^n$ has conditionally full support; the theorem will follow if we can show $\scr^n$ generates unique subset predictions. To that end, consider any $\hat{\ut}$ that rationalizes $\scr^{n}$. Because $\scr^{n}$ has conditionally full support, $\ip^{\scr^{n}}$ is fully mixed, and so $\icost$ is iteratively differentiable at $\ip^{\scr^{n}}$ by Assumption~\ref{ass: smoothness}\eqref{ass: smoothness, part differentiability}. Moreover, Assumption~\ref{ass: smoothness}\eqref{ass: smoothness, part finite} implies  
$\feas\ \icost \supseteq \Infof \ni \ip^{\scr^{n}}$, and so Corollary~\ref{cor: multiplier result for conditionally full-support, iteratively differentiable case} delivers some $\nuis\in \lp{1}$ for which every $a\in A$ has $\hat{\ut}_{a} = \ut^{n}_{a} + \nuis$. It follows that every $\scrb\in\Dom$ has
\[
\bbExp{\hat{\ut}\cdot\scrb} - \cost(\scrb) = \bbExp{\ut^{n}\cdot\scrb} - \cost(\scrb) + \bbExp{\nuis},
\]
and so every $\menu\in \Menu$ has
\[
\argmax_{\scrb \in \SCR_{\menu}} \left[\bbExp{\hat{\ut}\cdot\scrb} - \cost(\scrb) \right]= \argmax_{\scrb \in \SCR_{\menu}} \left[\bbExp{\ut^{n}\cdot\scrb} - \cost(\scrb)\right].
\] 
Hence, over any $\menu$, the utility $\hat{\ut}$ rationalizes the same set of SCRs as $\ut^{n}$ does. But $\ut^{n} \in \UT_{\Menu}$, meaning it rationalizes a unique SCR over $\menu$. It follows $\scr^{n}$ yields unique subset predictions, as required.
\end{proof}

\begin{proof}[Proof of Theorem~\ref{thm: cross-choice predictions}]
Theorem~\ref{thm: dense set of uniquely rationalizable SCRs} delivers a norm-dense subset of $\Dom$ comprising only uniquely rationalizable SCRs. 
The result then follows directly from Lemma~\ref{lem: dense unique subset predictions from dense uniquely rationalizable}.
\end{proof}

\subsection{Section~\ref{sec: cycle test} Proofs}\label{sec: cycle test proofs}

\begin{proof}[Proof of Corollary~\ref{cor: testable cycles}] 

Let $\dgraphE:=\dgraphE_{\data}$, and let $\func_{a,\menu}:=\func^{\data}_{a,\menu}$ for every $\{a,\menu\}\in\dgraphE$. For any $\cyclev\in\real^\dgraphE$, let $\cyclev\cdot\func:=\sum_{\{a,\menu\}\in\dgraphE} \cyclev_{\{a,\menu\}}\func^{\data}_{a,\menu}$. Observe now that condition~\eqref{cor: testable cycles - all cycles} says $\cyclev^\cycle\cdot\func=\mathbf{0}$ for every testable cycle $\cycle$, whereas condition~\eqref{cor: testable cycles - basis condition} says $\cyclev^\cycle\cdot\func=\mathbf{0}$ for every testable cycle $\cycle$ in the cycle basis. Because $\cyclev\mapsto\cyclev\cdot\func$ is linear, the equivalence of~\eqref{cor: testable cycles - all cycles} and~\eqref{cor: testable cycles - basis condition} follows directly.

Now, consider any $\menu\in\DMenu$ and any $\ut\in\UT$. Because $\dscr$ has conditionally full support, Proposition~\ref{prop: multiplier result for iteratively differentiable case} (applied to the model with action set $\menu$) tells us $\dscr^\menu$ is $\ut$-rationalizable over $\menu$ if and only if some $\nuis_\menu\in \lp{1}$ exists such that every $a\in\menu$ has $\ut_a = \nuis_\menu+\func_{a,\menu}$. Condition~\eqref{cor: testable cycles - consistency} is therefore equivalent to the existence of $(\ut,\nuis)\in\lp1^{A\cup\DMenu}$ such that
\begin{equation}\label{eqn: potential}
\func_{a,\menu} = \ut_a - \nuis_\menu \text{ for every } \{a,\menu\}\in\dgraphE.
\end{equation}
It therefore remains to see that condition~\eqref{cor: testable cycles - all cycles} is equivalent to the existence of $(\ut,\nuis)\in\lp1^{A\cup\DMenu}$  satisfying~\eqref{eqn: potential}.\footnote{The equivalence follows the same reasoning as the characterization of which current functions arise from some voltage function \citep[e.g., see p. 27 of ][]{bollobas2012graph}.  
For ease of exposition, and because that result is not stated for vector-valued currents, we provide here a self-contained proof in present notation.}

First, if $(\ut,\nuis)\in\lp1^{A\cup\DMenu}$  satisfies~\eqref{eqn: potential} and $\cycle=a_0 \menu_1 a_1 \menu_2 a_2 \ldots \menu_N a_N$ is any testable cycle, then 
$$\sum_{n=1}^{N} \left( \func_{a_{n},\menu_{n}}-\func_{a_{n-1},\menu_{n}} \right) 
= \sum_{n=1}^{N} \left[ (\ut_{a_n}-\nuis_{\menu_n})-(\ut_{a_{n-1}}-\nuis_{\menu_n}) \right] 
=\ut_{a_N}-\ut_0= \mathbf{0},$$
verifying condition~\eqref{cor: testable cycles - all cycles}.
Conversely, suppose~\eqref{cor: testable cycles - all cycles} holds.

Toward constructing $\ut\in\lp1^A$ and $\nuis\in\lp1^\DMenu$, first note that every menu $\menu\in\DMenu$ has some action $a^\menu$ such that $\{a,\menu\}\in\dgraphE$, and therefore every connected component $\hat\dgraph$ of the graph has some action $\refact(\hat\dgraph)$ in it. 
Now, for any action $a\in A$, let $\hat\dgraph^a$ denote its connected component, and fix some walk 
$$a_0^a \menu_1^a a_1^a \menu_2^a a_2^a \ldots \menu_{N^a}^a a_{N^a}^a$$
from $\refact(\hat\dgraph^a)$ to $a$---so in particular, $a_0^a=\refact(\hat\dgraph^a)$ and $a_{N^a}^a=a$. Then, let 
$$\ut_a:=\sum_{n=1}^{N_a} \left(\func_{a^a_{n}, \menu_n^a} - \func_{a^a_{n-1}, \menu_n^a}\right).$$
Next, for any menu $\menu\in\DMenu$, let $\nuis_\menu:=\ut_{a^\menu}-\func_{a^\menu,\menu}$. Having constructed $(\ut,\nuis)\in\lp1^{A\cup\DMenu}$, we need only show it satisfies~\eqref{eqn: potential}. To that end, consider any edge $\{a,\menu\}\in\dgraphE$, and let $\tilde a:=a^\menu$. 
Because they share a neighbor by construction, $a$ and $\tilde a$ belong to the same connected component, and so
$$\cycle:= a_0^a \menu_1^a a_1^a \menu_2^a a_2^a \ldots \menu_{N^a}^a a_{N^a}^a \ \menu \ 
a_{N^{\tilde a}}^{\tilde a}
\menu_{N^{\tilde a}}^{\tilde a} 
\ldots
a_2^{\tilde a}
\menu_2^{\tilde a}
a_1^{\tilde a} 
\menu_1^{\tilde a}
a_0^{\tilde a}$$
is a testable cycle. We can therefore apply equation~\eqref{eq: testable cycle equation} to $\cycle$ to learn
\begin{eqnarray*}
\mathbf0&=&
\sum_{n=1}^{N_a} \left(\func_{a^a_{n}, \menu_n^a} - \func_{a^a_{n-1}, \menu_n^a}\right)
+ \left(\func_{\tilde a, \menu} - \func_{a, \menu}\right)
+\sum_{n=1}^{N_{\tilde a}} \left(\func_{a^{\tilde a}_{n-1}, \menu_n^{\tilde a}} - \func_{a^{\tilde a}_{n}, \menu_n^{\tilde a}}\right)
\\
&=& \ut_a+ \left(\func_{\tilde a, \menu} - \func_{a, \menu}\right)
+ (-\ut_{\tilde a}) \\
&=& (\ut_a-\func_{a, \menu}) -\nuis_\menu,
\end{eqnarray*}
confirming~\eqref{eqn: potential}.
\end{proof}

We briefly note that some of the content of Corollary~\ref{cor: testable cycles} is readily recovered beyond the case of conditionally full-support data sets, but not all of it. 

\begin{fact}\label{fact: testable cycles counterexamples}
Suppose $\icost$ satisfies Assumptions~\ref{ass: smoothness}\eqref{ass: smoothness, part finite} and \ref{ass: smoothness}\eqref{ass: smoothness, part differentiability}, and let $\data=(\DMenu,\dscr)$ be a data set such that $\icost$ is iteratively differentiable at $\ip^{\dscr^\menu}$ for each $\menu\in\DMenu$.
\begin{enumerate}
    \item Given any cycle basis, conditions~\eqref{cor: testable cycles - all cycles} and~\eqref{cor: testable cycles - basis condition} of Corollary~\ref{cor: testable cycles}
    are equivalent. 
    \item Suppose $\data$ is fully mixed. If $\data$ is consistent, it satisfies condition~\eqref{cor: testable cycles - all cycles} of Corollary~\ref{cor: testable cycles}.
    \item Suppose $\data$ has full support. If $\data$ satisfies condition~\eqref{cor: testable cycles - all cycles} of Corollary~\ref{cor: testable cycles}, it is consistent.
\end{enumerate}
Moreover, neither of the latter two implications has a converse in general.
\end{fact}
See Appendix~\ref{app: aux} for the proof---which includes an example of a fully-mixed, inconsistent data set satisfying the cycle equation; and an example of a full-support, consistent data set not satisfying the cycle equation.

Our next goal is to prove Proposition~\ref{prop: testable cycles converse}. For this purpose, we first prove that under Assumptions~\ref{ass: regularity} and~\ref{ass: smoothness}, some consistent data set always exists which is linearly independent for every menu. 

\begin{lemma}\label{lem: a consistent conditional full support data exists}
Suppose \ref{ass: regularity}\eqref{ass: regularity, part cardinality} and~\ref{ass: smoothness}\eqref{ass: smoothness, part finite} hold. Then for any $\DMenu \subseteq \Menu$, some consistent data set $\data=(\DMenu,\dscr)$ is such that $\dscr^{\menu}$ is linearly independent (over $\menu$) for all $\menu \in \DMenu$.
\end{lemma}
\begin{proof}
Because $\prior$ has full support and $\Omega$ is metrizable, Assumption~\ref{ass: regularity}\eqref{ass: regularity, part cardinality} (i.e., $|A|\geq |\Omega|$) implies existence of a partition $\{\Omega_{a}\}_{a\in A}$ of $\Omega$ into $|A|$ measurable non-null sets. For $\paramb\in \mathbb{R}_{++}$, define the utility $\ut$ via
\[
\ut^{\paramb}_{a} = \paramb\mathbf{1}_{\Omega_{a}}.
\]
Define the set $\Omega_{\menu}=\cup_{a \in \menu}\Omega_{a}$ for any $\menu\subseteq\Menu$. For any $\menu\in\DMenu$, take $\scrb^{\menu} \in \SCR$ to be\footnote{If we weakened Assumption~\ref{ass: smoothness}\eqref{ass: smoothness, part finite} to assume only that every simple and fully mixed information policy has a finite cost, then the same result would hold with a slightly modified construction.}
\[
\scrb^{\menu}_{a} = \mathbf{1}_{\Omega_{a}} + \frac{1}{|\menu|}\mathbf{1}_{\Omega_{A\setminus\menu}}.
\]

Fix any menu $\menu \in \DMenu$, and for any $\ut\in\UT$ let
\[
\maxval_{\menu}(\ut):=\max_{\scr \in \SCR_{\menu}} \left[ \bbE[\ut \cdot \scr] - \cost(\scr) \right].
\]
Note, 
\[
\maxval_{\menu}(\ut^{\paramb}) \geq \bbE[\ut^{\paramb}\cdot \scrb^{\menu}] - \cost(\scrb^{\menu}) = \paramb \prior(\Omega_{\menu}) - \cost(\scrb^{\menu}).
\]
Because $\cost(\scrb^{\menu})<\infty$ by Assumption~\ref{ass: smoothness}\eqref{ass: smoothness, part finite},  we obtain that
\[
\liminf_{\paramb\rightarrow \infty}\frac{1}{\paramb}\maxval^{\menu}(\ut^{\paramb}) \geq \prior(\Omega_{\menu}).
\]
Hence, for sufficiently large $\paramb$, any SCR $\scr^{\menu}$ that is $\ut^{\paramb}$-rationalizable over $\menu$ has 
\begin{equation*}\label{eq: 02-03-2023a}
\bbE[\mathbf{1}_{\Omega_a}\cdot \scr_{a}] > \frac{1}{2}\prior(\Omega_a) \quad \text{for all }a\in \menu.
\end{equation*}
So fix $\paramb$ large enough that the above inequality holds for all $\menu \in \DMenu$; and for any $\menu \in \DMenu$, set $\dscr^{\menu}$ to be some SCR that is $\ut^{\paramb}$-rationalizable over $\menu$. 

To complete the proof, we argue $\dscr^{\menu}$ is linearly independent for all $\menu$---that is, that $\{\dscr^{\menu}_{a}\}_{a \in \menu}$ consists of $|\menu|$ linearly-independent vectors. For this purpose, define the $|\menu|\times|\menu|$ matrix,
\[
M := (\bbE[\mathbf{1}_{\Omega_{a}} \cdot \dscr^{\menu}_{b}])_{(a,b)\in \menu\times \menu} \in \mathbb{R}^{\menu\times\menu}_{+}.
\]
Notice $M$ is strictly diagonally dominant, and so has full rank. To complete the lemma, observe that if $\paramc:=(\paramc_{a})_{a\in\menu} \in \mathbb{R}^{\menu}$ is such that $\sum_{a}\paramc_{a}\dscr^{\menu}_{a}=\mathbf{0},$ then $M \paramc = \mathbf{0}$, meaning $\paramc = \mathbf{0}$. We conclude $\dscr^{\menu}$ is linearly independent for all $\menu \in \DMenu$. 
\end{proof}

We now proceed to prove Proposition~\ref{prop: testable cycles converse}. For a proof sketch, consider the graph generated by some full-support data set with menu set $\DMenu$, and let $\Cycle$ be a set of testable cycles that does not contain a cycle basis of that graph. The key to the proof is to use the data set generated by Lemma~\ref{lem: a consistent conditional full support data exists} and the utility function $\ut$ that rationalizes it for the construction of a new data set that is inconsistent, but satisfies the cycle equation for all cycles in $\Cycle$. The construction is based on pairing every menu $\menu$ with an SCR that is rationalized over $\menu$ by a slight perturbation of $\ut$, where we choose this perturbation so that it satisfies two properties. First, the perturbation is small enough so that the rationalized SCR is conditionally full support. Second, the perturbation results in a collection of utility functions whose differences across menus satisfy the cycle equation for all cycles in $\Cycle$, but fails that equation for some test cycle outside said set. Using Proposition~\ref{prop: multiplier result for iteratively differentiable case} to connect these utility differences to the cost function's derivative concludes the proof.

\begin{proof}[Proof of Proposition~\ref{prop: testable cycles converse}]
Fix some nonempty $\DMenu \subseteq \Menu$, and let $\dgraph=(\dgraphv,\dgraphE)$ be the graph induced by any full-support data set with menu set $\DMenu$. Take  $\Cycle$ to be a set of testable cycles in $\dgraph$ that does not contain a cycle basis. Our goal is to find a conditionally full-support data set that is inconsistent, but satisfies \eqref{eq: testable cycle equation} for all cycles in $\Cycle$. In what follows, we will invoke various supporting results from our paper, applying them to versions of our model with action set $\menu\in\DMenu$.

We first argue that some $z \in \mathbb{R}^{\dgraphE}$ exists such that $\cyclev^{\cycle}\cdot z = 0$ for all testable cycles $\cycle \in \Cycle$, but $\cyclev^{\check\cycle}\cdot z \neq 0$ for some testable cycle $\check\cycle \notin \Cycle$. To find such a $z$, let $\Cyclev$ be the vector subspace of $\mathbb{R}^{\dgraphE}$ spanned by the set of all cycle vectors, and let $\tilde\Cyclev$ be the span of $\{\cyclev^{\cycle}:\cycle \in \Cycle\}.$ Obviously, $\tilde\Cyclev \subseteq \Cyclev \subseteq \mathbb{R}^{\dgraphE}$. Moreover,  $\Cycle$ not containing a cycle basis implies $\tilde\Cyclev \neq \Cyclev$. Therefore, because a basis for $\tilde\Cyclev$ can be completed to one for $\Cyclev$, a nonzero linear map from $\Cyclev$ to $\mathbb{R}$ exists that is zero on $\tilde\Cyclev$. Then, because a basis for $\Cyclev$ can be completed to a basis for $\mathbb{R}^{\dgraphE}$, this linear map can be extended to a linear map from $\mathbb{R}^{\dgraphE}$ to $\mathbb{R}$. Representing this linear map via $\cyclev \mapsto z \cdot \cyclev$ (because $\dgraphE$ is finite) gives a $z$ as desired. 

By Lemma~\ref{lem: a consistent conditional full support data exists}, a consistent data set $\data=(\DMenu,\dscr)$ such that $\dscr^{\menu}$ is linearly independent (hence full support) for every $\menu \in \DMenu$. Let $\ut^{0}$ be some utility that simultaneously rationalizes $\dscr^{\menu}$ over $\menu$ for all $\menu \in \DMenu$. 

For every $\menu \in \DMenu$ and every $\epsilon>0$, define $\ut^{\menu,\epsilon} \in \UT$ via  

\[
\ut^{\menu,\epsilon}_{a} = 
\begin{cases}
\ut^{0}_{a} + \epsilon z_{a,\menu} & \text{if }a \in \menu, \\
\ut^{0}_{a} & \text{otherwise,}
\end{cases}
\]
and take $\scr^{\menu,\epsilon}$ to be some SCR that is $\ut^{\menu,\epsilon}$-rationalizable over $\menu$. We argue that $\scr^{\menu,\epsilon}$ has conditionally full support (over $\menu$) for sufficiently small $\epsilon$. To that end, define the proper, weak*-lower semicontinuous, and convex (by Lemma~\ref{lemma: Properties of Indirect Cost}) cost function,
\[
\cost_{\menu}(\scr) = \begin{cases}
\cost(\scr) & \text{if }\scr \in \SCR_\menu,\\
\infty & \text{otherwise,}
\end{cases}
\]
and take
\[
\begin{split}
\maxval_{\menu}: \UT & \rightarrow \mathbb{R},\\
\ut & \mapsto \max_{\scr \in \SCR}\bbE [ \ut \cdot \scr ] - \cost_\menu (\scr)
\end{split}
\]
to be the value function with cost function $\cost_{\menu}$. Letting $\partial \maxval_{\menu}$ be the subdifferential of $\maxval_\menu$, Lemma~\ref{lem: subdifferential characterization of optimality} and the definition of $\cost_\menu$ imply
\[
\partial \maxval_\menu (\ut) 
= \argmax_{\scr\in \SCR} \left[ \bbE [ \ut \cdot \scr ] - \cost_\menu (\scr) \right]
= \argmax_{\scr \in \SCR_\menu} \left[ \bbE [ \ut \cdot \scr ] - \cost (\scr) \right].
\]
Thus, $\scr^{\menu,\epsilon} \in \partial \maxval_{\menu}(\ut^{\menu,\epsilon})$. Applying Proposition~\ref{prop: linear independence implies unique rationalizability}, we therefore obtain that $\partial \maxval_\menu (\ut^{0}) = \{\dscr^{\menu}\}$---in particular, $\partial\maxval_\menu$ is singleton valued at $\ut^{0}$. Since $\partial\maxval_\menu$ is also norm-to-weak* upper hemicontinuous (Lemma~\ref{lem: value function continuity}), that $\ut^{\menu,\epsilon} \rightarrow \ut^{0}$ implies $\scr^{\menu,\epsilon}\xrightarrow{w^*}\dscr^{\menu}$. In particular, $\scr^{\menu,\epsilon}$ has full support over $\menu$ for sufficiently small $\epsilon>0$. Applying Lemma~\ref{lem: inada and full support implies conditionally full support}, it follows that $\scr^{\menu,\epsilon}$ is conditionally full support (on $\menu$) for all $\epsilon>0$ small enough.

Pick $\epsilon>0$ so that $\scr^{\menu,\epsilon}$ is conditionally full support on $\menu$ for all $\menu \in \DMenu$. Define the data set $\tilde\data=(\DMenu,\tilde\dscr)$ by setting $\tilde\dscr^{\menu}=\scr^{\menu,\epsilon}$. We now conclude the proof by showing $\tilde\data$ satisfies \eqref{eq: testable cycle equation} for all cycles in $\Cycle$, but violates said equation for the testable cycle $\check\cycle$ identified at the beginning of this proof. For any $\menu \in \Menu$, Proposition~\ref{prop: multiplier result for iteratively differentiable case} delivers a $\nuis_\menu \in \lp{1}$ such that for all $a \in \menu$,
\[
\ut^{\menu,\epsilon}_{a}=\nuis_\menu + \func^{\tilde\data}_{a,\menu},
\]
and so
\[
\func^{\tilde\data}_{a,\menu} = \ut^{0}_a - \nuis_\menu + \epsilon z_{a,\menu}\mathbf{1}.
\]
Thus, given a testable cycle $\cycle=a_0 \menu_1 a_1 \ldots \menu_N a_N$, 
\[
\sum_{n=1}^{N} (\func_{a_n,\menu_n} - \func_{a_{n-1},\menu_n}) = 
\left[\sum_{n=1}^{N} (\ut^{0}_{a_n} - \ut^{0}_{a_{n-1}} + \nuis_{\menu_n} - \nuis_{\menu_n})\right] + 
\epsilon \left[\sum_{n=1}^{N}(z_{a_n,\menu_n} - z_{a_{n-1},\menu_n}) \right] 
= \epsilon (\cyclev^{\cycle}\cdot z)\mathbf{1}. 
\]
By choice of $z$, it follows that \eqref{eq: testable cycle equation} holds for all cycles $\cycle\in\Cycle$ but fails for $\cycle=\check\cycle$, as required.
\end{proof}

\newpage
\section{Supplement to Section~\ref{sec: discussion}}\label{app: discussion proofs}

This appendix provides formal support for any nontrivial claims made in section~\ref{sec: discussion}.

\subsection{Partial Knowledge of Benefits}

In this section, we give a characterization of which SCRs are $\UTB$-rationalizable for certain well-behaved instances of $\UTB\subseteq \UT$. 

Begin with the case in which the relevant set is a singleton, $\UTB=\{\ut\}$. Recall, as explained in the main text, and shown in Lemma~\ref{lem: subdifferential characterization of optimality}, to answer the question of whether or not $\ut$ rationalizes $\scr$ is the same as checking whether or not $\ut$ is a subgradient of $\cost$ at $\scr$. The latter condition admits a test in terms of $\cost$'s directional derivative. Specifically, the SCR $\scr$ is $\ut$-rationalizable if and only if 
\[
\dcost{\scr}(\scr') \geq \bbExp{\ut \cdot (\scr' - \scr)}
\]
for all $\scr' \in \SCR$. One can interpret the above condition as saying the agent does not benefit at the margin from shifting her behavior from $\scr$ in any direction. 

Our next result generalizes the above test for the case in which $\UTB$ is convex and  weakly compact.%
\footnote{
A sequence of utilities $\sequence{\ut^{n}}$ converges weakly to $\ut$ if $\bbExp{\ut^{n} \cdot \scru}\rightarrow \bbExp{\ut \cdot \scru}$ for all $\scru \in \SCRU$. The set of utilities $\UTB$ is weakly compact if it is weakly closed and uniformly integrable; see \cite{bogachev2007measureV1} Definition~4.5.1 and Theorem~4.7.18.
}
Our generalization replaces the change in the agent's expected benefits $\bbExp{\ut \cdot (\scr' - \scr)}$ with the \textbf{support function} of $\UTB$,
\[
\begin{split}
\suppf_{\UTB} : \SCRU & \rightarrow \mathbb{R}\cup\{\pm\infty\},\\
\scru & \mapsto \sup_{\ut \in \UTB} \bbExp{\ut\cdot\scru}.
\end{split}
\]
Given a pair $\scr,\scr'$ of SCRs, the quantity $-\suppf_{\UTB}(\scr - \scr')=\inf_{\ut \in \UTB} \bbExp{\ut\cdot(\scr'-\scr)}$ gives the lowest possible marginal increase in the agent's objective from shifting her behavior from $\scr$ towards $\scr'$.
Geometrically, the support function describes a closed convex set via its supporting hyperplanes.

The next proposition shows one can test whether $\scr$ is $\UTB$-rationalizable by comparing the support function of $\UTB$ to the cost's directional derivative at $\scr$.

\begin{proposition}\label{prop: rationalizable characterization for compact set of utilities}
If $\UTB \subset \UT$ is nonempty, weakly compact, and convex, then $\scr\in\Dom$ is $\UTB$-rationalizable if and only if every $\scr' \in \Dom$ has $$d_{\scr}^{+} \cost (\scr') \geq -\suppf_{\UTB}(\scr - \scr').$$
\end{proposition}
\begin{proof}
First, suppose $\scr$ is $\UTB$-rationalizable, that is, rationalized by some $\ut \in \UTB$. Taking any $\scr' \in \Dom$, 
Lemma~\ref{lem: subdifferential characterization of optimality} then implies $d_{\scr}^{+}\cost (\scr')+\bbExp{\ut\cdot(\scr-\scr')}\geq0$.
But observe $\suppf_{\UTB}(\scr - \scr') \geq \bbExp{\ut\cdot(\scr-\scr')}$ by definition of $\suppf_{\UTB}$. Hence,
$d_{\scr}^{+} \cost (\scr') + \suppf_{\UTB}(\scr - \scr') \geq d_{\scr}^{+} \cost (\scr') + \bbExp{\ut\cdot(\scr-\scr')} \geq 0$.\footnote{As the proof makes clear, convexity and compactness of $\UTB$ play no role in this direction of the equivalence.}

Conversely, suppose every $\scr' \in \Dom$ has $d_{\scr}^{+} \cost (\scr') + \suppf_{\UTB}(\scr - \scr') \geq0$.
Toward showing $\scr$ is $\UTB$-rationalizable, observe convexity of $\cost$ implies every $\epsilon \in [0,1]$ and $\scr' \in \SCR$ have
$\tfrac{1}{\epsilon}\left[ \cost (\scr + \epsilon (\scr' - \scr) ) - \cost (\scr) \right] \leq \cost(\scr') - \cost (\scr)$; and so taking limits yields $d_{\scr}^{+} \cost (\scr') \leq \cost(\scr') - \cost(\scr).$ Using this inequality gives
\begin{equation*}
\begin{split} 
0 & \leq \inf_{\scr' \in \Dom} \left\{d_{\scr}^{+} \cost (\scr') + \suppf_{\UTB}(\scr - \scr') \right\}
\\ & \leq\inf_{\scr' \in \Dom} \left\{\cost(\scr') - \cost(\scr) + \suppf_{\UTB}(\scr - \scr') \right\} 
\\ & =\inf_{\scr' \in \Dom} \left\{\cost(\scr') - \cost(\scr) + \max_{\ut \in \UTB}\bbExp{\ut\cdot(\scr-\scr')} \right\}
\\ & =\inf_{\scr' \in \Dom}\max_{\ut \in \UTB} \left\{\cost(\scr') - \cost(\scr) + \bbExp{\ut\cdot(\scr-\scr')} \right\} 
\\ & = \max_{\ut \in \UTB} \inf_{\scr' \in \Dom} \left\{\cost(\scr') - \cost(\scr) + \bbExp{\ut\cdot(\scr-\scr')} \right\},
\end{split}
\end{equation*}
where the first equality follows from weak compactness of $\UTB$, and the last equality then follows from Sion's minmax theorem \citep{sion1958general}. Therefore, some $\ut \in \UTB$ exists such that $\cost(\scr') - \cost(\scr) + \bbExp{\ut\cdot(\scr-\scr')} \geq 0$ for all $\scr' \in \Dom$---that is, $\scr$ is $\ut$-rationalizable. 
\end{proof}

Proposition~\ref{prop: rationalizable characterization for compact set of utilities} is most useful if one can compute directional derivatives of $\cost$ and the support function of the utility set $\UTB$. Naturally, the tractability of these computations will depend on the form of $\cost$ (hence, $\icost$) and of $\UTB$. 

When $\icost$ is differentiable at the information policy revealed by a given SCR, Lemma~\ref{lem: posterior separable approximation characterizes optimal SCR} tells us we can replace $\cost$ with its posterior-separable approximation, and so one need only compute directional derivatives of the latter:
\begin{corollary}
If $\UTB \subset \UT$ is nonempty, weakly compact, and convex, and $\icostd$ is a derivative of $\icost$ at $\scr\in\Dom$, then $\scr$ is $\UTB$-rationalizable if and only if every $\scr' \in \Dom$ has $$d_{\scr}^{+} \cost_\icostd (\scr') \geq -\suppf_{\UTB}(\scr - \scr').$$
\end{corollary}
\noindent In this case, Lemma~\ref{lem: directional derivative of cost for differentiable case} provides an explicit formula for the relevant directional derivatives, and Lemma~\ref{lem: directional derivative of cost for iteratively differentiable case, and subgradient result for differentiable case} provides an even more explicit formula for the same in the iteratively differentiable case. 

Below, we provide several examples of sets of utilities to which Proposition~\ref{prop: rationalizable characterization for compact set of utilities} can be applied, and explicitly compute the relevant support functions.

\begin{example}[Utility bounds]\label{ex:gross-bound}
Given constant payoff bound $\bar\ut\in\real_{++}$, consider the set\footnote{In a standard abuse, we say an element of $\lp{1}$ is nonnegative if some version of it is nonnegative. We adopt an analogous abuse for other pointwise function inequalities, or (when $\Omega$ is ordered) in calling a member of $\lp{1}$ increasing.} $$\UTB:=\{\ut\in\UT:\ |\ut_a|\leq \bar\ut\mathbf1\ \forall a\in A\}.$$ 
Let us demonstrate that $\UTB$ satisfies the hypotheses of Proposition~\ref{prop: rationalizable characterization for compact set of utilities} and has $\suppf_{\UTB}(\scru)=\bar\ut\Vert\scru \Vert_1$ for every $\scru\in\SCRU$. As our analysis shows, computing support functions can in fact be a useful step in verifying the compactness condition on $\UTB$ required by the theorem. With this work in hand, we can apply Proposition~\ref{prop: rationalizable characterization for compact set of utilities}: A given SCR $\scr\in\Dom$ is rationalized by some objective with payoffs bounded by $\bar\ut$ if and only if 
every $\scr' \in \Dom$ has $d_{\scr}^{+} \cost (\scr') \geq -\bar\ut\Vert\scr - \scr'\Vert_1$.

To see $\UTB$ satisfies the proposition's hypotheses, let $L_+\subseteq\lp{1}$ denote the set of nonnegative functions. 
Because $\UTB=\left[ \left(L_+-\bar\ut\mathbf1\right) \cap \left(\bar\ut\mathbf1-L_+\right)\right]^A$ and
$$
L_+=\left\{\func\in\lp{1}: \ \int \mathbf1_{\hat\Omega}(\omega)\func(\omega) \ \prior(\dd\omega)\geq0 \ \forall\text{ measurable } \hat\Omega\subseteq\Omega \right\},
$$
it follows $\UTB$ is convex and weakly closed. It remains to compute the support function of $\UTB$ and verify $\UTB$ is weakly compact. 
To that end, observe the support function of $\UTB$ evaluated at any $\scru\in\SCRU$ takes the form 
$$\suppf_{{\UTB}}(\scru)= \sup_{\ut \in {\UTB}} \sum_{a\in A} \bbExp{\ut_a\scru_a}= \sum_{a\in A} \bbExp{\bar\ut\ \text{sign}(\scru_a) \scru_a}= \bar\ut\sum_{a\in A} \bbE|\scru_a|=\bar\ut\Vert\scru \Vert_1,$$
and (as noted in the second equality above) the program defining this support function attains a maximum at $(\bar\ut\ \text{sign}(\scru_a))_{a\in A}$. Hence, James' theorem \citep[Theorem 6.36 in][]{aliprantis2006infinite} implies that $\UTB$ is weakly compact.
\end{example}

\begin{example}[Restricted incremental utility]\label{ex:incremental-bound}
Suppose $A=\{a_0,\ldots,a_n\}$ for distinct $a_0, \ldots,a_n$, and let $\Func_i \subseteq \lp{1}$ be nonempty, convex, and weakly compact for each $i\in \{1,\ldots,n\}$. Let $$\UTB:=\left\{\ut\in\UT:\ \ut_{a_i}-\ut_{a_{i-1}}\in\Func_i\ \forall i\in\{1,\ldots,n\}\right\}.$$
Although the unbounded set $\UTB$ is not weakly compact, we can still use Proposition~\ref{prop: rationalizable characterization for compact set of utilities}. To see we can, let $\UTB_0:=\{\ut\in\UTB:\ \ut_{a_0}=0\}$. 
Because every $\ut\in\UTB$ rationalizes the same set of SCRs as utility $\ut^0:=(\ut_a-\ut_{a_0})_{a\in A}\in\UTB_0$, it follows that a given SCR is $\UTB$-rationalizable if and only if it is $\UTB_0$-rationalizable. Moreover, we show below that $\UTB_0$ satisfies Proposition~\ref{prop: rationalizable characterization for compact set of utilities}'s hypotheses and compute the support function of $\UTB_0$ in terms of the support functions of $\{\Func_i\}_{i=1}^n$---demonstrating useful algebraic properties of support functions that make them easier to compute. The upshot is that some $\ut\in\UTB$ rationalizes a given SCR $\scr\in\Dom$ if and only if \begin{equation}\label{ex,eq:incremental-bound}
d_{\scr}^{+} \cost (\scr') \geq - \sum_{i=1}^n \suppf_{\Func_i}\left( \sum_{j=i}^n (\scr_{a_j}-\scr'_{a_j})\right), \ \forall \scr'\in\Dom.
\end{equation}

Let us now verify $\UTB_0$ satisfies Proposition~\ref{prop: rationalizable characterization for compact set of utilities}'s hypotheses.
First, observe $\UTB_0$ is nonempty (containing the zero utility) and obviously convex. To see it is compact, define \begin{eqnarray*}
\Phi: \lp{1}^n &\to& \lp{1}^A = \UT \\
\func=(\func_i)_{i=1}^n &\mapsto& \left( \sum_{i=1}^j \func_i \right)_{a=a_j\in A}.
\end{eqnarray*}
The linear map $\Phi$ is clearly norm-to-norm continuous, hence weak-to-weak continuous. It follows that the image $\UTB_0=\Phi(\prod_{i=1}^n\Func_i)$ is weakly compact, and Proposition~\ref{prop: rationalizable characterization for compact set of utilities} applies.

Finally, let us see how basic algebraic properties of support functions can be used to make the computation of $\suppf_{\UTB_0}$ more explicit. Recall the adjoint (also known as the transpose) $\Phi^*:\SCRU=\lp{\infty}^A\to\lp{\infty}^n$ is the unique linear norm-continuous map such that every $\scru\in\SCRU$ and $\func\in\prod_{i=1}^n \Func_i$ satisfy $\sum_{i=1}^n \bbE\left[ (\Phi^*\scru)_i \func_i \right] = \bbExp{\scru\cdot\Phi \func} $. Then, each $\scru\in\SCRU$ has 
$$
\suppf_{\UTB_0}(\scru)
= \suppf_{\Phi(\prod_{i=1}^n\Func_i)}(\scru)
= \suppf_{\prod_{i=1}^n\Func_i}(\Phi^*\scru)
= \sum_{i=1}^n \suppf_{\Func_i}\left( (\Phi^*\scru)_i\right),
$$
where the second equality follows directly from the definition of the support function and the adjoint, and the third equality follows from additive separability of linear functions on product spaces. 
Meanwhile, by direct computation, the functional form $\Phi^*\scru = \left(\sum_{j=i}^n \scru_{a_j}\right)_{i=1}^n$ satisfies the equation defining the adjoint, and
so
$$\suppf_{\UTB_0}(\scru)= \sum_{i=1}^n \suppf_{\Func_i}\left( \sum_{j=i}^n \scru_{a_j}\right).$$
Therefore, $\scr$ is $\UTB$-rationalizable if and only if \eqref{ex,eq:incremental-bound} holds.
\end{example}

\begin{example}[Bounded increasing differences]\label{ex:increasing-diff}
Suppose $A=\{a_0,\ldots,a_n\}\subset\real$ for $a_0<\cdots<a_n$ and $\Omega\subset\real$. Let $\UTB$ denote the set of utilities with (i) weakly increasing differences, and (ii) payoffs Lipschitz of constant $\xi$ with respect to the chosen action. That is, $\UTB$ is the set of all utilities $\ut\in\UT$ such that every $i\in\{1,\ldots,n\}$ has $\ut_{a_i}-\ut_{a_{i-1}}$ weakly increasing with magnitude bounded by $\bar\ut_i:=\xi({a_i}-{a_{i-1}})$. 
Below, we verify that this example is an instance of Example~\ref{ex:incremental-bound} with some appropriate specification of the $\{\Func_i\}_{i=1}^n$, and the support function $\suppf_{\Func_i}$ takes the form 
\begin{equation}\label{ex,eq:increasing-diff}
\suppf_{\Func_i}(\func)=\bar\ut_i\max_{\omega^*\in\Omega}\left\{ -\int_{(-\infty,\omega^*)} \func(\omega)\ \prior(\dd\omega) + |\func(\omega^*)|\prior(\omega^*)+ \int_{(\omega^*,\infty)} \func(\omega) \ \prior(\dd\omega)
\right\}.
\end{equation}
A key component of the argument, which is useful for applying Proposition~\ref{prop: rationalizable characterization for compact set of utilities} more generally, is that one can compute a support function of a compact convex set by optimizing only over its extreme points. 
Substituting \eqref{ex,eq:increasing-diff} into equation \eqref{ex,eq:incremental-bound} yields an explicit characterization of whether a given SCR is $\UTB$-rationalizable.

For example, consider the further specialization to the case in which $A=\Omega=\{0,1\}$ and $\xi=1$.  In this case, the expression that needs to be nonnegative according to equation \eqref{ex,eq:incremental-bound} is
$$
d_{\scr}^{+} \cost (\scr') + \begin{cases}
\bbE\left|\scr_1-\scr'_1\right| &: \  \scr'_1(0)>\scr_1(0) \text{ and } \scr'_1(1)<\scr_1(1),\\
\left|\ip^\scr_1-\ip^{\scr'}_1\right| &: \  \text{otherwise.}
\end{cases}
$$
In particular, in this case, the magnitude of the ``marginal benefit'' component is a distance between $\scr$ and $\scr'$. If $\scr'$ matches the state less than $\scr$ does in both states, then it is the probability that $\scr$ and $\scr'$ choose different actions; and otherwise it is the difference in wholesale action frequencies between $\scr$ and $\scr'$. 

Let us now formally verify the relevant compactness property and the support function computations. 
First, in Example~\ref{ex:gross-bound}, we established that $\{\func_i\in\lp{1}:\ |\func_i|\leq \bar\ut_i\}$ was weakly compact. Moreover, the set of $\func_i\in\lp{1}$ with weakly increasing version is 
$$\bigcap_{\Omega_L,\Omega_H\subseteq\Omega \text{ measurable}: \ \sup\Omega_L\leq\inf\Omega_H} \left\{ \func_i\in\lp{1}:\ \prior(\Omega_H)\bbE\left[\mathbf1_{\Omega_L}\func_i\right]\leq \prior(\Omega_L)\bbE\left[\mathbf1_{\Omega_H}\func_i\right] \right\},
$$
a weakly closed set. Hence, the nonempty convex set $$\Func_i:=\{\func_i\in\lp{1}:\ \func_i \text{ weakly increasing, } |\func_i|\leq \bar\ut_i\},$$ 
is weakly compact for every $i\in\{1,\ldots,n\}$. Therefore, the computations of Example~\ref{ex:incremental-bound} apply directly.
Toward a more explicit characterization, let us compute $\suppf_{\Func_i}(\func)$ for each $i\in\{1,\ldots,n\}$ and $\func\in\lp{\infty}$. Because $\Func_i$ is convex and compact, the Bauer maximum principle says that every continuous linear function evaluated over $\Func_i$ attains its maximum at the extreme points, $\ext\ \Func_i=\{\func_i\in\lp{1}:\ \func_i \text{ weakly increasing, } |\func_i|= \bar\ut_i\}$.\footnote{The proof that the extreme points of $\Func_i$ takes this form is essentially identical to that of Lemma 2.7 from \cite{borgers2015introduction}.} Each extreme point therefore takes a simple cutoff form, taking value $-\bar\ut_i$ to the left of the cutoff and value $\bar\ut_i$ to the right of the cutoff, and one of these values at the cutoff (if it is a $\prior$-atom). Hence, $$\suppf_{\Func_i}(\func)=\bar\ut_i\max_{\omega^*\in\Omega}\left\{ -\int_{(-\infty,\omega^*)} \func(\omega) \ \prior(\dd\omega) + |\func(\omega^*)|\prior(\omega^*)+ \int_{(\omega^*,\infty)} \func(\omega) \ \prior(\dd\omega)
\right\}.$$

Finally, consider the specialization to the case in which $A=\Omega=\{0,1\}$ and $\xi=1$.  Direct computation then shows the relevant support function takes the form \begin{eqnarray*}
\suppf_{\UTB_0}(\scru)&=& \suppf_{\Func_1}(\scru_1)\\
&=&\max\{-\bbE[\scru_1],\ 
-\prior(0)\scru_1(0)+\prior(1)\scru_1(1),\ 
\bbE[\scru_1]\} \\
&=& \max\left\{\left|\bbE[\scru_1]\right|,\ 
-\prior(0)\scru_1(0)+\prior(1)\scru_1(1)\right\} \\
&=& \begin{cases}
\bbE\left|\scru_1\right| &: \  \scru_1(0)<0<\scru_1(1) \\
\left|\bbE[\scru_1]\right| &: \  \text{otherwise,}
\end{cases}
\end{eqnarray*}
as described above.
\end{example}

\begin{example}[Action payments]\label{ex:action-pay}
Given $\underline\ut\in\UT$ and payment bound $\bar w\geq0$, let $$
\UTB:=\{\underline\ut+(w_a\mathbf1)_{a\in A}:\ w\in[0,\bar w]^A\},
$$
which is (being the convex hull of finitely many points) obviously nonempty, convex, and weakly compact. 
We can interpret this example as settling a question of implementability with transfers: If agent has latent utility $\underline\ut$, and a principal has the ability to offer action-dependent payments (with preferences quasilinear in money) bounded by $\bar w$, can SCR $\scr\in\Dom$ be incentivized? 
As we show below, each $\scr,\scr'\in\SCR$ have\footnote{Observe, the last term is proportional to the total-variation distance between the induced marginal action distributions.} $$\suppf_{\UTB}(\scr-\scr')=\bbExp{\underline\ut\cdot(\scr-\scr')}+ \tfrac{\bar w}2\sum_{a\in A}\left|\ip^\scr_a-\ip^{\scr'}_a \right|.$$ 
So for a given SCR $\scr\in\Dom$, let $$
\bar w^*:= \inf_{\scr'\in\Dom:\ \ip^{\scr'}_a\neq\ip^{\scr}_a \ \exists a\in A} \tfrac2{\sum_{a\in A}\left|\ip^\scr_a-\ip^{\scr'}_a \right|}\left[\bbExp{\underline\ut\cdot(\scr'-\scr)}-d_{\scr}^{+} \cost (\scr')\right].
$$
Then, Proposition~\ref{prop: rationalizable characterization for compact set of utilities} implies $\scr\in\Dom$ is rationalizable with monetary incentives in $[0,\bar w]$ if and only if $\bar w\geq\bar w^*$.

Now, let us verify the described computation. To that end, note each $\scru\in\SCRU$ has
\begin{eqnarray*}
\suppf_{\UTB}(\scru) 
&=& \bbExp{\underline\ut\cdot\scru} +\suppf_{\UTB-\underline\ut}(\scru) \\
&=&\bbExp{\underline\ut\cdot\scru} + \suppf_{\bar w\left[\co\{\mathbf0, \mathbf1\}\right]^A}(\scru)\\
&=&\bbExp{\underline\ut\cdot\scru}+\bar w\sum_{a\in A}\suppf_{\co\{\mathbf0, \mathbf1\}}(\scru_a) \\
&=&\bbExp{\underline\ut\cdot\scru}+ \bar w\sum_{a\in A}\left(\bbExp{\scru_a} \right)_+;
\end{eqnarray*}
whereas any $\scr,\scr'\in \SCR$ have 
\begin{eqnarray*}
\sum_{a\in A}\left(\bbExp{\scr_a-\scr_a'} \right)_+ 
&=& \sum_{a\in A}\left(\ip^\scr_a-\ip^{\scr'}_a \right)_+ \\
&=& \sum_{a\in A}\left[\tfrac12\left(\ip^\scr_a-\ip^{\scr'}_a \right) + \tfrac12\left|\ip^\scr_a-\ip^{\scr'}_a \right|\right] \\
&=& \tfrac12\left[\sum_{a\in A}\left(\ip^\scr_a-\ip^{\scr'}_a \right)\right] 
+ \tfrac12\sum_{a\in A}\left|\ip^\scr_a-\ip^{\scr'}_a \right| \\
&=&\tfrac12\sum_{a\in A}\left|\ip^\scr_a-\ip^{\scr'}_a \right|.
\end{eqnarray*}
The aforementioned form for the support function follows directly.
\end{example}

\subsection{Partial Knowledge of Costs}\label{app: sets of costs}

For the LLR cost function, direct computation shows that any data set data set $\data = (\DMenu,\dscr)$ and state $\omega$ have
\[
\func^{\data}_{a,\menu}(\omega) = \sum_{\omega'} \left\{\frac{\mathbf{\param}_{\omega,\omega'}}{\prior(\omega)}\left[ \posterior_{a}^{\dscr^{\menu}}(\omega')+ \log \frac{\posterior_{a}^{\dscr^{\menu}}(\omega)}{\posterior_{a}^{\dscr^{\menu}}(\omega')}  \right] - \frac{\mathbf{\param}_{\omega',\omega}}{{\prior(\omega')}}  \posterior_{a}^{\dscr^{\menu}}(\omega') \right\}.
\]

We can express Corollary~\ref{cor: testable cycles} more succinctly in the language of vectors and matrices. Specifically, view $\boldsymbol{\param}$ as a matrix in $\real^{\Omega\times\Omega}$, view $\prior$ as a vector in $\real^\Omega$, view $\func^{\data}=[\func^{\data}_{a,\menu}(\omega)]_{\omega,\{a,\menu\}}$ and $\boldsymbol{\posterior}=[\posterior_{a,\menu}(\omega)]_{\omega,\{a,\menu\}}$
as a matrices in $\real^{\Omega\times\dgraphE_\data}$, and let $\log \boldsymbol{\posterior}$ denote the matrix whose entries are the logarithms of the corresponding entries of $\boldsymbol{\posterior}$. Finally, we find it convenient to reparameterize the cost function, by letting\footnote{Here, ${\rm{diag}}:\real^\Omega\to\real^{\Omega\times\Omega}$ is the natural embedding taking each vector to a diagonal matrix.} $$\boldsymbol{\paramb}:=\left[{\rm{diag}}(\prior)\right]^{-1}\left[\boldsymbol{\param}-{\rm{diag}}(\boldsymbol{\param}\mathbf1_\Omega)\right],$$
a matrix each of whose rows sums to zero.\footnote{These matrices are in bijective correspondence with zero-diagonal $\boldsymbol{\param}\in\real^{\Omega\times\Omega}$. 
Indeed, direct computation shows 
 $
\boldsymbol{\param}=\left[{\rm{diag}}(\prior)\right]\left(\boldsymbol{\paramb}-\boldsymbol\iota\odot\boldsymbol{\paramb}\right)$, where $\boldsymbol\iota\in\real^{\Omega\times\Omega}$ is the identity matrix, and $\odot$ is the Hadamard product.} 
It follows readily that $$\func^{\data}=(\boldsymbol{\paramb}-\boldsymbol{\paramb}^\top) \boldsymbol{\posterior}-
\boldsymbol{\paramb}(\log \boldsymbol\posterior).$$
Corollary~\ref{cor: testable cycles} then tells us that $\data$ is consistent if and only if every testable cycle $\cycle$ has
$$(\boldsymbol{\paramb}-\boldsymbol{\paramb}^\top) \boldsymbol{\posterior} \cyclev^\cycle
=
\boldsymbol{\paramb}(\log \boldsymbol\posterior)
 \cyclev^\cycle,$$
a condition amounting to $|\Omega|\xi$ linear restrictions on $\boldsymbol{\paramb}$ (hence on $\boldsymbol{\param}$), where $\xi$ is the cardinality of a cycle basis of $\data$ as calculated in \eqref{eq: cycle basis cardinality}.

Now, let us specialize to the example data set described by~\eqref{eq: example dataset}. The computations that show this data set is consistent with mutual information costs in fact show $(\log \boldsymbol\posterior) \cyclev^\cycle$ is zero for every testable cycle in a cycle basis, and the uniform prior makes $\boldsymbol{\paramb}-\boldsymbol{\paramb}^\top$ proportional to $\boldsymbol{\param}-\boldsymbol{\param}^\top$. Hence, the data set is consistent with LLR costs of parameter $\boldsymbol{\param}$ if and only if  $(\boldsymbol{\param}-\boldsymbol{\param}^\top) \boldsymbol{\posterior} \cyclev^\cycle$ is zero for both testable cycles $\cycle\in\{\cycle_1,\cycle_2\}$ in our cycle basis. This condition clearly holds for every symmetric matrix $\boldsymbol{\param}$, and an easy computation shows no nonzero skew-symmetric matrix has zero product with $\boldsymbol{\posterior} \cyclev^\cycle$ for both of these testable cycles $\cycle$. Because every square matrix is the sum of a symmetric and skew-symmetric matrix, it follows that the data set is consistent if and only if $\boldsymbol{\param}$ is symmetric.

\subsection{Unique Rationalizability and Strict Convexity}

The following lemma shows that, when $\icost$ is strictly convex, the agent's optimal information choice is unique, and the only scope for multiple best responses is the willingness to mix over her action conditional on a realized signal.

\begin{lemma}\label{lem: indirect cost is almost strictly convex if icost is strictly convex}
Suppose $\icost$ is strictly convex, and let $\ut\in\UT$. If $\scr,\scrb\in\SCR$ are $\ut$-rationalizable, then $\ip^{\scr}=\ip^{\scrb}$ and, for every $a\in A$, the functions $\scr_a$ and $\scrb_a$ are proportional. 
\end{lemma}
\begin{proof}
First, we show no two $\ut$-rationalizable SCRs can generate different information policies. 
To that end, suppose $\scr^1,\scr^2\in\Dom$ have $\ip^{\scr^1}\neq\ip^{\scr^2}$; we show they cannot both be $\ut$-rationalizable. 
Letting $\scr:=\tfrac12\sum_{i=1}^2\scr^i$, Lemma~\ref{lem: SCR basic properties from IP basic properties}\eqref{lem: SCR basic properties from IP basic properties-convexity} tells us $\tfrac12\sum_{i=1}^2\ip^{\scr^i} \succeq \ip^{\scr}$, so that 
$$\cost\left(\scr\right)=\icost\left(\ip^{\scr}\right)\leq\icost\left(\tfrac12\sum_{i=1}^2\ip^{\scr^i}\right)<\tfrac12\sum_{i=1}^2\icost\left(\ip^{\scr^i}\right)=\tfrac12\sum_{i=1}^2\cost\left(\scr^i\right).$$
Hence, $\bbExp{\ut\cdot\scr}-\cost(\scr)>\tfrac12\sum_{i=1}^2\left[ \bbExp{\ut\cdot\scr^i}-\cost(\scr^i) \right]\geq\min_{i=1,2}\left[ \bbExp{\ut\cdot\scr^i}-\cost(\scr^i) \right]$. 
Therefore, $\scr$ witnesses that at least one of $\{\scr^1,\scr^2\}$ is not $\ut$-rationalizable.

Now, let $\scr,\scrb\in\Dom$ be $\ut$-rationalizable. Because $\scrc\mapsto\bbExp{\ut\cdot\scrc}-\cost(\scrc)$ is concave (by Lemma~\ref{lemma: Properties of Indirect Cost}), it follows that $\scrc:=\tfrac12(\scr+\scrb)$ is $\ut$-rationalizable as well. By the above analysis,  $\ip^\scr=\ip^\scrb=\ip^\scrc$; in particular, $\ip^\scrc\nprec\tfrac12\ip^\scr+\tfrac12\ip^\scrb.$ 
Lemma~\ref{lem: SCR basic properties from IP basic properties}\eqref{lem: SCR basic properties from IP basic properties-strictness} then implies no $a \in \supp (\scr) \cap \supp (\scrb)$ exists such that $\posterior_{a}^{\scr} \neq \posterior_{a}^{\scrb}$. Said differently, $\scr_a$ and $\scrb_a$ are proportional for every $a\in A$. 
\end{proof}

\begin{proof}[Proof of Proposition~\ref{prop: dense set of uniquely rationalizable SCRs under strict convexity}]
Let $\SCR^1$ denote the set of SCRs $\scr\in\Dom$ with $\supp(\scr)=A$ and the $|A|$ beliefs $\{\posterior^\scr_a\}_{a\in A}$ all distinct. Below, we show $\SCR^1$ is weak* open and norm dense in $\Dom$. Before doing so, let us see this result would deliver the proposition. 
To that end, let $\SCRrat$ denote the set of rationalizable SCRs if $\Omega$ is infinite, and let it denote the relative interior of $\Dom$ if $\Omega$ is finite. Given Proposition~\ref{prop: dense set of rationalizable SCRs}, every SCR in $\SCRrat$ is rationalizable, and $\SCRrat$ is norm dense in $\Dom$. Moreover, by construction, $\SCRrat$ is open in $\Dom$ if $\Omega$ is finite. Hence, the intersection $\SCRrat\cap\SCR^1$ is norm dense in $\Dom$ and open in it if $\Omega$ is finite.
Moreover, if we establish that any $\ut\in\UT$ and $\ut$-rationalizable $\scr\in\SCRrat\cap\SCR^1$ are such that $\scr$ is uniquely $\ut$-rationalizable, Lemma~\ref{lem: open set of utilities for weak* open set of uniquely rationalized SCRs} would tell us the SCRs in $\SCRrat\cap\SCR^1$ are rationalized by an open set of utilities.\footnote{Just as in the proof of Theorem~\ref{thm: dense set of uniquely rationalizable SCRs}, we can apply Lemma~\ref{lem: open set of utilities for weak* open set of uniquely rationalized SCRs} to the  weak*-open set $G$ (resp. $\SCRrat\cap G$) if $\Omega$ is infinite (resp. finite).} 
Finally, let us see that any $\ut\in\UT$ and $\ut$-rationalizable $\scr\in\SCRrat\cap\SCR^1$ are such that $\scr$ is uniquely $\ut$-rationalizable. To do so, take any $\ut$-rationalizable SCR $\scrb$. That $\scr\in\SCR^1$ implies $\ip^\scr$ has support size $|A|$. 
By Lemma~\ref{lem: indirect cost is almost strictly convex if icost is strictly convex}, we know $\ip^\scrb=\ip^\scr$, and so $\ip^\scrb$ has support size $|A|$ too. Hence, $\supp(\scr)=\supp(\scrb)=A$. Lemma~\ref{lem: indirect cost is almost strictly convex if icost is strictly convex} then also tells us $\posterior^\scr_a=\posterior^\scrb_a$ for every $a\in A$. But then, because $\{\posterior^\scr_a\}_{a\in A}$ are $|A|$ distinct beliefs, it follows from $\ip^\scr=\ip^\scrb$ that $\ip_a^\scr=\ip_a^\scrb$ for every $a\in A$. Hence, $\scrb=\scr$, as desired.

So all that remains is to show $\SCR^1$ is weak* open and norm dense in $\Dom$. 
To that end, define the set $$\SCR^+:=\bigcap_{a\in A}\left\{\scr\in\Dom:\ \bbExp{\scr_a}>0 \right\}$$
and, for each distinct $a,b\in A$, the set
$$\SCR^{a,b}:=
\bigcup_{\hat\Omega\subseteq\Omega\text{ Borel}
}
\left\{\scr\in\Dom:\ \bbExp{\scr_a}\bbExp{\mathbf1_{\hat\Omega}\scr_b}\neq\bbExp{\scr_b}\bbExp{\mathbf1_{\hat\Omega}\scr_a}\right\},
$$
By construction, these sets are all weak* open (hence, norm open) in $\Dom$, and so too is the finite intersection $\SCR^1=\SCR^+\cap\bigcap_{a,b\in A \text{ distinct}} \SCR^{a,b}$. Moreover, because a finite intersection of open and dense sets is dense, norm denseness of $\SCR^1$ will follow if we can show $\SCR^+$ is norm dense in $\Dom$ and each pair of distinct $a,b\in A$ has $\SCR^+\cap\SCR^{a,b}$ norm dense in $\SCR^+$. 
To see $\SCR^+$ is norm dense, note $\scr^+:=(\tfrac1{|A|}\mathbf1)_{a\in A}\in\Dom$ because $\ip^{\scr^+}=\delta_\prior$ and $\icost(\delta_{\prior})<\infty$. Because $\Dom$ is convex by Lemma~\ref{lemma: Properties of Indirect Cost}, it follows that every $\scr\in\Dom$ and $\epsilon\in(0,1)$ have $(1-\epsilon)\scr+\epsilon\scr^+\in\Dom$, which witnesses (taking $\epsilon\to0$) the SCR $\scr$ as a norm limit from $\SCR^+$.

Now, fix any pair of distinct $a,b\in A$. It remains to show $\SCR^+\cap\SCR^{a,b}$ is norm dense in $\SCR^+$. To that end, let $\scr\in\SCR^+$ be arbitrary; we want to show $\scr$ is in the norm closure of $\SCR^+\cap\SCR^{a,b}$. If $\scr\in\SCR^{a,b}$, we have nothing to show, so we focus on the complementary case in which  $\posterior^\scr_a=\posterior^\scr_b=:\posterior$. Below, we locate an SCR $\scr^*\in\Dom$ such that any proper convex combination of $\scr$ and $\scr^*$ lies in $\SCR^{a,b}$. Observe such proper convex combinations necessarily live in $\Dom$ (because $\cost$ is convex) and so in $\SCR^+$ (because); and they can approximate $\scr$ arbitrarily well by choosing sufficiently skewed weights. Hence, finding such an $\scr^*$ will yield the required denseness property. To locate such an $\scr^*$, we separately address the case in which $\posterior\neq\prior$ and the case in which $\posterior=\prior$.

First, suppose $\posterior\neq\prior$, and let $\scr^*$ denote the unique SCR with $\scr^{*}_{a}=\mathbf1$. That $\icost(\delta_{\prior})<\infty$ implies $\scr^*\in\Dom$. Moreover, any proper convex combination $\scrb$ of $\scr$ and $\scr^*$ is in $\SCR^{a,b}$ because it has $\scrb_b=\posterior$, whereas $\scrb_a$ is a proper convex combination of $\posterior$ and $\prior$. Thus, $\scr^*$ is as required.

Finally, suppose $\posterior=\prior$. By hypothesis, some $\ip\in\Domicost$ exists such that $\ip\neq\delta_\prior$. Because $\ip$ has barycenter $\prior$, the distribution $\ip$ must be nondegenerate. Pooling different posteriors generated by $\ip$ if necessary (which will weakly reduce costs and so remain in $\Domicost$ because $\icost$ is monotone), we may assume without loss $\ip$ has binary support $\{\posterior^*_a,\posterior^*_b\}$. Note $\prior$ lies strictly between $\posterior^*_a$ and $\posterior^*_b$.  Then, define $\scr^*$ to be the unique SCR with $\posterior^{\scr^*}_a=\posterior^*_a$ and $\posterior^{\scr^*}_b=\posterior^*_b$. By construction, $\cost(\scr^*)=\icost(\ip^{\scr^*})=\icost(\ip)<\infty$; that is, $\scr^*\in\Dom$. Finally, any proper convex combination $\scrb$ of $\scr$ and $\scr^*$ is in $\SCR^{a,b}$ because it has $\scrb_a$ and $\scrb_b$ being proper convex combinations of $\prior$ with $\posterior^*_a$ and $\posterior^*_b$, respectively. Thus, $\scr^*$ is as required.
\end{proof}

\begin{proof}[Proof of Proposition~\ref{prop: cross-choice predictions under strict convexity}]
Assumption~\ref{ass: smoothness}\eqref{ass: smoothness, part finite} implies, given $|\Omega|>1$, that $\Domicost\neq\{\delta_\prior\}$. Hence, Proposition~\ref{prop: dense set of uniquely rationalizable SCRs under strict convexity} delivers a norm-dense subset of $\Dom$ comprising only uniquely rationalizable SCRs. 
The result then follows directly from Lemma~\ref{lem: dense unique subset predictions from dense uniquely rationalizable}.
\end{proof}

\subsection{Subdifferentiability, Rationalizability, and Posterior Separable Costs}

\begin{proposition}\label{SubdiffDerivToCost}
Fix some $\scr \in \SCR$ such that $\ip^{\scr} \in \feas \ \icost$, and suppose $\icostd$ is a derivative of $\icost$ at $\ip^{\scr}$. If $\partial \icostd (\posterior_{a}^{\scr})$ is nonempty for all $a \in \supp \ \scr$, then $\partial \cost(\scr)$ is nonempty, and so $\scr$ is rationalizable.
\end{proposition}
\begin{proof}
By Lemma~\ref{lem: posterior separable approximation characterizes optimal SCR}, it suffices to show $\scr$ is rationalizable according to $\cost_{\icostd},$ which is (by Lemma \ref{lem: subdifferential characterization of optimality}) equivalent to some $\ut\in\UT$ satisfying, for every $\scr'\in\Dom$, the inequality $d_{\scr}^{+}\cost_\icostd (\scr') \geq \bbExp{\ut\cdot(\scr'-\scr)}$.

For each $a\in\supp(\scr)$, take some $\func_a\in\partial \icostd (\posterior_{a}^{\scr})$. Shifting $\func_a$ by a constant if necessary---which clearly preserves $\func_a\in\partial \icostd (\posterior_{a}^{\scr})$---assume without loss that $\int\func_a(\omega)\ \posterior_a^\scr(\dd\omega)=\icostd(\posterior_a^\scr)$. 
Let us see $\ut\in\UT$ given by $$\ut_a=\begin{cases}
\func_a &:\ a\in\supp(\scr) \\
\min\icostd(\DO)\mathbf1& :\ a\in A\setminus\supp(\scr)
\end{cases}$$
yields the desired payoff ranking. Indeed, for any $\scr'\in\SCR$, Lemma~\ref{lem: directional derivative of cost for iteratively differentiable case, and subgradient result for differentiable case} tells us 
\begin{eqnarray*}
d^{+}_{\scr} \cost_{\icostd} (\scr') 
&\geq& \sum_{a\in \supp(\scr)} \bbExp{ (\scr'_{a} - \scr_{a} )\ut_{a}}+\sum_{a\in \supp(\scr')\setminus\supp(\scr)}  \bbExp{(\scr'_{a} - \scr_{a}) \ \icostd(\posterior_{a}^{\scr'})} \\
&\geq& \bbExp{ (\scr' - \scr )\cdot\ut},
\end{eqnarray*}
as required.
\end{proof}

\subsection{Convexity and Monotonicity}

The following lemma constructs a ``convex monotone envelope'' of an information cost function $\hat\icost$ and establishes some regularity properties of the same.

\begin{lemma}\label{lem: convex monotone envelope}
Suppose $\hat\icost: \Info \to \ereal$ is proper and lower semicontinuous. Then, \begin{eqnarray*}
\icost: \Info &\to& \ereal \\
\ip &\mapsto& \min_{\mip\in\Delta\Info:\ \int\ipb\ \mip(\dd\ipb)\succeq\ip} \int\hat\icost(\ipb)\ \mip(\dd\ipb)
\end{eqnarray*}
is well defined (i.e., the defining minimum exists), proper, lower semicontinuous, convex, and monotone.
\end{lemma}
\begin{proof}
Given the HLPBSSC theorem, the relation $\succeq$ is continuous (i.e., a closed subset of $\Info^2$). We can therefore apply a version of the maximum theorem \cite[e.g., Lemma 17.30 from][]{aliprantis2006infinite} because the barycenter map is continuous and $\hat\icost$ is lower semicontinuous: the map $\icost: \Info \to \ereal$ is well defined and lower semicontinuous.\footnote{The map takes values in $\ereal$ rather than $\real$, but the cited lemma can be applied because $\ereal$ is homeomorphic to a subset of $\real$ via a strictly increasing transformation.} Moreover, $\icost$ is monotone because $\succeq$ is transitive, and it is proper because $\hat\icost$ is proper and $\icost\leq\hat\icost$.

Finally, we turn to convexity. The HLPBSSC theorem implies $\succeq$ is a convex subset of $\Info^2$, and so the correspondence $\Info\rightrightarrows\Delta\Info$ given by $\ip\mapsto\{\mip\in\Delta\Info:\ \int\ipb\ \mip(\dd\ipb)\succeq\ip\}$ has a convex graph.  Hence, $\icost$ is a (weakly) convex function. 
\end{proof}

\begin{proof}[Proof of Proposition~\ref{prop: monotonicity and convexity are wlog}]
Let $\icost$ be the function defined in the statement of Lemma~\ref{lem: convex monotone envelope}, which (by that lemma) satisfies all of our standing assumptions on information cost functions.

Take any SCR $\scr$. Lemma~\ref{lem: which policies can induce} implies a given $\mip\in\Delta\Info$ can induce $\scr$ if and only if $\int\ip\ \mip(\dd\ip)\succeq\ip^\scr$. Hence, by Lemma~\ref{lem: convex monotone envelope}, some $\mip\in\Delta\Info$ of minimum average cost $\int\hat\icost(\ip)\ \mip(\dd\ip)$ can induce $\scr$, and this cost is exactly equal to $\icost(\ip^\scr)$. The result follows.
\end{proof}

\subsection{Continuous Choice with Bounded Utilities}\label{app: Continuous Choice with Bounded Utilities}

We begin with an abstract result on $\UTB$-rationalizability for the case in which $\UTB$ is a well-behaved linear subspace of $\UT$.

\begin{lemma}\label{lem: rationalizability by utilities in a subspace}
Suppose $\UTB\subseteq\UT$ is a linear subspace, $(\tilde\SCRU, \Vert\cdot\Vert)$ is a normed space with $\tilde\SCRU\supseteq\SCRU$, and 
$\{\varphi|_\SCR:\ \varphi:\tilde\SCRU\to\real \text{ linear continuous}\}$
is the set of maps $\SCR\to\real$ given by $\scr\mapsto\bbExp{\ut\cdot\scr}$ for $\ut\in\UTB$. 
Then, an SCR $\scr$ is $\UTB$-rationalizable if and only if $$\inf_{\scr'\in\Dom:\ \scr'\neq\scr}\frac{ \cost(\scr')-\cost(\scr) }{\Vert\scr'-\scr\Vert}>-\infty.$$
\end{lemma}

\begin{proof}
Fix $\scr\in\SCR$. Viewing $\cost$ as a function on $\tilde\SCRU$ (by letting it take value $\infty$ on $\tilde\SCRU\setminus\SCR$), the subdifferential of $\cost$ at $\scr$ takes the form
$$\partial_{\tilde\SCRU}\cost(\scr)=\left\{\varphi\in\tilde\SCRU^*:\ \cost(\tilde\scr)\geq \cost(\scr) + \varphi(\tilde\scr-\scr) \ \forall \tilde\scr\in\tilde\SCRU \right\}.$$
By hypothesis, $\{\varphi|_\SCR:\ \varphi\in\tilde\SCRU^*\}$ is the set of functions $\SCR\to\real$ that are given by $\scr\mapsto\bbExp{\ut\cdot\scr}$ for some $\ut\in\UTB$. Hence, 
\begin{eqnarray*}
&&\partial_{\tilde\SCRU}\cost(\scr) \neq\varnothing \\
&\iff& \{\varphi|_\SCR:\ \varphi\in\partial_{\tilde\SCRU}\cost(\scr)\}\neq\varnothing \\
&\iff& \exists\ut\in\UTB \text{ such that } \forall \scr'\in\SCR, \  \cost(\scr')\geq \cost(\scr) + \bbExp{\ut\cdot(\scr'-\scr)}\\
&\iff& \scr \text{ is $\UTB$-rationalizable.}
\end{eqnarray*}
Therefore, the lemma follows immediately from \citeapos{gale1967geometric} duality theorem.
\end{proof}

\begin{proof}[Proof of Proposition~\ref{prop: continous choice with bounded utilities}]
Letting $\tilde\SCRU=\lp{1}^A$, whose dual space is naturally identified with $\UTB=\lp{\infty}^A$, the result follows directly from the first part of Lemma~\ref{lem: rationalizability by utilities in a subspace}.
\end{proof}

Although we have applied Lemma~\ref{lem: rationalizability by utilities in a subspace} to determine which SCRs are rationalized by some bounded utility, one can vary the norm (generating distinct bounded-steepness conditions) to characterize rationalizability with respect to different classes of utilities. For example, $\cost$ exhibits bounded steepness with respect to the $\lp{2}$ norm at exactly the SCRs that can be rationalized by a finite-variance utility. Similarly, bounded steepness of $\cost$ with respect to the Kantorovich-Rubinstein norm characterizes rationalizability by a Lipschitz utility.

In contrast to the previous examples, Lemma~\ref{lem: rationalizability by utilities in a subspace} cannot be applied directly to the case in which $\UTB$ is the space of continuous functions $\Omega\to\real$. The reason is that, whereas the above examples relied on viewing $\SCRU$ as a subset of some $\tilde\SCRU$ whose dual was naturally identified with $\UTB$, the space of continuous functions is typically not a dual. Indeed, if $\Omega$ is infinite with finitely many connected components (e.g., if it is $[0,1]$), one can easily show (via the Banach-Alaoglu theorem) the space of continuous functions is not the dual of any normed space.

\subsection{Costly Stochastic Choice}

Let us state the analogue of Theorem~\ref{thm: cross-choice predictions} that we reported in section~\ref{sec: discussion}.

\begin{samepage}
\begin{proposition}\label{prop: cross-choice predictions for unfounded costs}
Suppose $\tilde\cost$ is strictly convex on its domain $\SCR^{\tilde\cost}$ and is finite and differentiable at every conditionally full-support SCR. Further, suppose every full-support $\scr\in \SCR^{\tilde\cost}$ that does not have conditionally full support admits some $\scr'\in\SCR$ such that $d^+_\scr\tilde\cost(\scr')= -\infty.$ 
Then, SCRs yielding unique subset predictions are weak* dense in $\SCR^{\tilde\cost}$.
\end{proposition}
\end{samepage}

We now build up to the proof of the above proposition. 

\begin{remark}
In the analysis of this section, we apply Lemmas~\ref{lem: unique prediction with known benefits}, \ref{lem: value function continuity},  and~\ref{lem: subdifferential characterization of optimality} and Proposition~\ref{prop: dense set of rationalizable SCRs} to $\tilde\cost$. All of these results were proven under the hypothesis that $\cost$ is proper, convex, and weak* lower semicontinuous (established in Lemma~\ref{lemma: Properties of Indirect Cost} for $\cost$ and directly assumed for $\tilde\cost$). In particular, inspection of the proofs of these results shows they do not use the fact that $\cost$ is derived from $\icost$, and so the results can be applied to $\tilde\cost$ without change.
\end{remark}

The following lemma adapts Lemma~\ref{lem: multiplier result for full-support iteratively differentiable case} to the simpler setting in which $\tilde\cost$ is assumed differentiable.

\begin{lemma}\label{lem: multiplier result for differentiable SCR costs}
Let $\scr\in\SCR^{\tilde\cost}$ have $\supp(\scr)=A$, and suppose $\ut$ is a derivative of $\tilde\cost$ at $\scr$.
The following are equivalent for $\utb\in\UT$:
\begin{enumerate}[(i)]
\item SCR $\scr$ is $\utb$-rationalizable; that is,  $\scr\in\argmax_{\scrb\in\SCR}\left[ \bbExp{\utb\cdot\scrb}-\tilde\cost(\scrb) \right]$.
\item Some $\nuis\in \lp{1}$ and $\mult\in\lp{1}_+^A$ exist such that every $a\in A$ has \begin{eqnarray*}
\utb_a &=& \nuis-\mult_a+\ut_a, \\
\scr_a\mult_a &=& 0.
\end{eqnarray*}
\end{enumerate}
\end{lemma}
\begin{proof}
Let $\ut^\scr:= \utb-\ut \in \UT$. Lemma~\ref{lem: subdifferential characterization of optimality} tells us $\scr$ is $\utb$-rationalizable if and only if every $\scr'\in\Dom$ has $d_{\scr}^{+}\tilde\cost (\scr') \geq \bbExp{\utb\cdot(\scr'-\scr)}$. But, that $\ut$ is a derivative of $\tilde\cost$ at $\scr$ implies $d_{\scr}^{+}\tilde\cost (\scr')=\bbExp{\ut\cdot(\scr'-\scr)}$. We can therefore write the optimality condition as $\scr\in\argmax_{\scr'\in\SCR} \bbExp{\ut^{\scr}\cdot\scr'}$. As in the proof of Lemma~\ref{lem: multiplier result for full-support iteratively differentiable case}, this condition is equivalent to requiring some $\nuis\in\lp{1}$ and $\mult\in\lp{1}_+^A$ to have $\ut^\scr=(\nuis-
\mult_a)_{a\in A}$ and $(\mult_a\scr_a)_{a\in A}=0$.
\end{proof}

The following analogue of Corollary~\ref{cor: multiplier result for conditionally full-support, iteratively differentiable case} 
is an immediate consequence of Lemma~\ref{lem: multiplier result for differentiable SCR costs}.

\begin{corollary}\label{cor: multiplier result for conditionally full-support case with differentiable SCR costs}
Let the SCR $\scr$ have conditionally full support. 
If $\tilde\cost$ is differentiable at $\scr$, and $\ut$ and $\ut'$ both rationalize $\scr$, some $\nuis \in \lp{1}$ exists such that 
\[
\ut_{a} = \ut'_{a} + \nuis \ \forall a\in A.
\]
\end{corollary}

Next, we record an analogue of Lemma~\ref{lem: inada and full support implies conditionally full support}.

\begin{lemma}\label{lem: inada and full support implies conditionally full support for unfounded costs}
If $\scr\in \SCR^{\tilde\cost}$ admits a $\scr'\in\SCR$ such that $d^+_\scr\tilde\cost(\scr')= -\infty,$
then $\scr$ is not rationalizable.
\end{lemma}
\begin{proof}
For every $\epsilon\in (0,1)$, let $\scrb^{\epsilon}= \scr + \epsilon(\scr' - \scr)$. Every $\ut \in \UT$ then has
\[
\frac{1}{\epsilon}\left\{ \bbExp{\ut \cdot \scrb^{\epsilon}} - \tilde\cost(\scrb^{\epsilon}) -\left[\bbExp{\ut \cdot \scr} - \tilde\cost(\scr) \right]  \right\}
= \bbExp{\ut \cdot (\scr' - \scr)} - \frac{1}{\epsilon}\left[\tilde\cost(\scrb^{\epsilon}) - \tilde\cost(\scr) \right] \xrightarrow{\epsilon\searrow 0}-\infty.
\]
Thus, for all sufficiently small $\epsilon\in(0,1)$, the agent's objective must be strictly higher under $\scrb^{\epsilon}$ than under $\scr$. 
\end{proof}

Next, we provide an analogue of Lemma~\ref{lem: dense unique subset predictions from dense uniquely rationalizable}.

\begin{lemma}\label{lem: dense unique subset predictions from dense uniquely rationalizable with unfounded costs}
Suppose $\tilde\cost$ is finite and differentiable at every conditionally full-support SCR, every full-support $\scr\in \SCR^{\tilde\cost}$ that does not have conditionally full support admits some $\scr'\in\SCR$ such that $d^+_\scr\tilde\cost(\scr')= -\infty$, 
and some norm-dense subset $\SCR_{1}$ of $\SCR^{\tilde\cost}$ exists that comprises only uniquely rationalizable SCRs. Then, the set of SCRs yielding unique subset predictions is weak* dense in $\SCR$.
\end{lemma}
\begin{proof}
The proof of Lemma~\ref{lem: dense unique subset predictions from dense uniquely rationalizable} (which invokes Lemmas~\ref{lem: unique prediction with known benefits}, \ref{lem: value function continuity}, and~\ref{lem: subdifferential characterization of optimality} and Proposition~\ref{prop: dense set of rationalizable SCRs}) can be applied nearly verbatim, with three minor differences. First, the step invoking Assumption~\ref{ass: smoothness}\eqref{ass: smoothness, part finite} is replaced with the hypothesis that the set $\SCR^{\tilde\cost}$ contains all conditionally full-support SCRs, because the latter set is dense in $\SCR$. Second, we invoke Lemma~\ref{lem: inada and full support implies conditionally full support for unfounded costs} instead of Lemma~\ref{lem: inada and full support implies conditionally full support}. Third, we invoke Corollary~\ref{cor: multiplier result for conditionally full-support case with differentiable SCR costs} instead of Corollary~\ref{cor: multiplier result for conditionally full-support, iteratively differentiable case}.
\end{proof}

Finally, we can prove the main result of this section.

\begin{proof}[Proof of Proposition~\ref{prop: cross-choice predictions for unfounded costs}]
By Proposition~\ref{prop: dense set of rationalizable SCRs} and strict convexity of $\tilde\cost$, the set $\SCR_{1}$ of uniquely rationalizable SCRs is norm dense in $\SCR^{\tilde\cost}$. The proposition then follows from Lemma~\ref{lem: dense unique subset predictions from dense uniquely rationalizable with unfounded costs}.
\end{proof}

\newpage
\section{Auxiliary Material}\label{app: aux}

\subsection{Omitted Proofs for Auxiliary Results}

This section provides proofs of some additional results that were stated in the previous appendices without proof.

\begin{proof}[Proof of Corollary~\ref{cor: unique information choice for binary actions}]
Say $A=\{0,1\}$ without loss, and note $\scr\in\SCR$ generates state-independent behavior if and only if $\ip^\scr=\delta_{\prior}$.
We therefore require that if $\scr\in\SCR$ is optimal and $\ip^\scr\neq\delta_{\prior}$, then $\scr$ is uniquely optimal. But $\ip^\scr\neq\delta_{\prior}$ for $\ip^\scr\in\Info$ implies both that $\ip^\scr_0,\ip^\scr_1>0$ and that $\posterior^\scr_0,\posterior^\scr_1$ are distinct, and hence affinely independent. Proposition~\ref{prop: linear independence implies unique rationalizability} therefore applies directly. 

Letting $\ip^*:=\ip^{\scr^*}$ for some $\ut$-rationalizable SCR $\scr^*$, the above argument tells us $\ip^\scr=\ip^*$ for every $\ut$-rationalizable SCR $\scr$. Therefore (in light of Lemma~\ref{lem: which policies can induce}), $\icost$ being strictly monotone means every optimal strategy entails information policy exactly $\ip^*$, as required.
\end{proof}

\begin{proof}[Proof of Fact~\ref{fact: convexity of derivative is wlog under continuity}]
Toward establishing convexity of $\icostd$, fix any $\posterior_1,\posterior_2 \in \Delta \Omega$ and $\wt \in (0,1)$, and let $\posterior = (1-\wt)\posterior_1 + \wt\posterior_2.$ Let $(\posterior_1^n,\posterior_2^n)_{n\in \mathbb{N}}$ be a sequence of pairs of simply drawn posteriors converging to $(\posterior_1,\posterior_2)$, which exists by \citeapos{lipnowski2018disclosure} Lemma~2. By definition of a simply drawn posterior, some sequence $(\zeta_1^n,\zeta_2^n)_{n\in \mathbb{N}}$ of pairs from $\real_{++}$ exists such that $\zeta_i^n\posterior_i^n \leq \prior$ for all $i\in\{1,2\}$ and $n\in\mathbb N$. 

Now, consider any $n\in\mathbb N$. Take $\zeta^n := \tfrac12\min \{\zeta_1^n,\zeta_2^n\}$ and $\posterior^n := (1-\wt)\posterior_1^n + \wt \posterior_2^n$ for $n\in\mathbb N$. Observe $2\zeta^n \posterior^n \leq (1-\wt)\zeta_1^n\posterior_1^n + \wt \zeta_2^n \posterior_2^n \leq \prior$, and so $\hat\posterior^n = \tfrac{\prior - \zeta^n\posterior^{n}}{1-\zeta^n}$ is a well-defined simply drawn posterior. Therefore, both 
$$\ip_{n} :=  \zeta^n\delta_{\posterior^n} + (1-\zeta^n)\delta_{\hat\posterior^n} \text{ and } \tilde{\ip}_{n} := \zeta^n \left[ (1-\wt)\delta_{\posterior_1^n}+\wt\delta_{\posterior_2^n} \right] + (1-\zeta^n)\delta_{\hat\posterior^n}$$
are in $\Infof.$ Because $\ip\in\feas \ \icost$, it follows that $\ip+\epsilon(\ip_{n} - \ip)$ and $\ip+\epsilon(\tilde{\ip}_{n}- \ip )$ are both in $\Domicost$ for all sufficiently small $\epsilon>0$. 
Hence, because $\tilde{\ip}_{n}\succeq {\ip}_{n}$ by construction, we have (because $\icost$ is monotone) 
\begin{equation*}
\begin{split}
0 & \leq \liminf_{\epsilon\searrow 0}\frac{1}{\epsilon}\left[ \icost(\ip+\epsilon(\tilde{\ip}_{n} - \ip)) - \icost(\ip+\epsilon(\ip_{n} - \ip)) \right]
\\ & =\liminf_{\epsilon\searrow 0}\frac{1}{\epsilon}\bigg\{ \big[ \icost(\ip+\epsilon(\tilde{\ip}_{n} - \ip)) - \icost(\ip)\big] - \big[ \icost(\ip+\epsilon(\ip_{n} - \ip))-\icost(\ip) \big]  \bigg\}
\\ & = \int \icostd(\posterior) \ (\tilde{\ip}_{n} - \ip)(\dd\posterior) - \int \icostd(\posterior) \ ({\ip}_{n} - \ip)(\dd\posterior)
\\ & = \zeta^n \big[ (1-\wt) \icostd(\posterior_1^n) + \wt \icostd(\posterior_2^n) - \icostd(\posterior^n) \big].  
\end{split}
\end{equation*}

Finally, that $(1-\wt) \icostd(\posterior_1^n) + \wt \icostd(\posterior_2^n) \geq \icostd(\posterior^n)$ for every $n\in\mathbb N$ yields $(1-\wt) \icostd(\posterior_1) + \wt \icostd(\posterior_2) \geq \icostd(\posterior)$ because $\icostd$ is continuous. The lemma follows.
\end{proof}

\begin{proof}[Proof of Fact~\ref{fact: derivative of derivative is unique}]
Let $\func\in\lp{1}$ be the difference of two derivatives of $\icostd\in\icostD$ at the simply drawn $\posterior$.\footnote{Note that the proof relies only on the ``Newton quotient'' property of $\icostdd_{\posterior}$, together with its average value normalization, and so does not require that $\icostd\in\icostD$.} By hypothesis, every simple $\posteriorb \in \DO$ has $\int \func(\omega)\ (\posteriorb - \posterior)(\dd\omega)=0$, and hence, $\int \func(\omega)\ \posteriorb(\dd\omega)=\bar\func:=\int \func(\omega)\ \posterior(\dd\omega)$. 

Because $\int \func(\omega)\ \prior(\dd\omega)=\bar\func$, the event $\hat\Omega:=\{\func\geq\bar\func\}\subseteq\Omega$ has $\prior(\hat\Omega)>0$. Let $\posteriorb:=(\prior|\hat\Omega)$ denote the conditional measure, which is a simply drawn posterior. 
Therefore, 
$$0 = \prior(\hat\Omega)\int \left[\func(\omega)-\bar\func\right]\ \posteriorb (\dd\omega) = \bbExp{\mathbf1_{\hat\Omega} (\func-\bar\func)}.$$
An expectation of a nonnegative random variable can be zero only if the random variable is almost surely zero, yielding two consequences. First, $\mathbf1_{\hat\Omega} (\func-\bar\func)=0\in\lp{1}$, so that $\func\leq\bar\func$. Second, that $\func\leq\bar\func$ and $\bbExp{\func}=\bar\func$ implies $\func=\bar\func\in\lp{1}$, a constant.

Finally, that every derivative of $\icostd$ at $\posterior$ has $\posterior$-expectation $\icostd(\posterior)$ then implies $\bar\func=0$.\footnote{If we did not require the normalization that $\int \icostdd_{\posterior}(\omega) \ \posterior(\dd\omega)=\icostd(\posterior)$---which would not affect the definition of differentiability---the derivative would be unique only up to addition of a constant function.}
\end{proof}

\begin{proof}[Proof of Fact~\ref{fact: equivalent inada conditions}]
Every $\ipb\in\Info$ has $\ipb\succeq\delta_\prior$, and so $\ip + \epsilon(\ipb - \ip) \succeq \ip + \epsilon(\delta_\prior - \ip)$ for any $\epsilon\in(0,1)$. Hence, \begin{eqnarray*}
&\inf_{\ipb\in\Domicost,\ \epsilon\in(0,1)}&\frac{1}{\epsilon} \left[
\icost(\ip + \epsilon(\ipb - \ip)) - \icost(\ip) \right] \\
=&\inf_{\epsilon\in(0,1)}&\frac{1}{\epsilon} \left[ \icost(\ip + \epsilon(\delta_\prior - \ip)) - \icost(\ip) \right] \text{ (by monotonicity)}\\
=&\lim_{\epsilon\searrow0}&\frac{1}{\epsilon} \left[ \icost(\ip + \epsilon(\delta_\prior - \ip)) - \icost(\ip) \right] \text{ (by convexity)},
\end{eqnarray*}
which is equal to $d^{+}_{\ip}\icost(\delta_\prior)$ by definition.
\end{proof}

\begin{proof}[Proof of Fact~\ref{fact: testable cycles counterexamples}]
To see the first numbered fact, note that the proof of the equivalence of conditions~\eqref{cor: testable cycles - all cycles} and~\eqref{cor: testable cycles - basis condition} in Corollary~\ref{cor: testable cycles} makes no use of conditional full support. 

To see the second fact, observe that the proof of Corollary~\ref{cor: testable cycles} applies directly (mutatis mutandis) to a generalized version of a data set in which the set of menus is a \emph{multiset} of menus. 
In particular, it applies to the generalized data set $\data'$ in which each menu $\menu\in\DMenu$ is replaced with the smaller menu $\supp\dscr^\menu$; note that (relabeling vertices in $\DMenu$ in the obvious way) the graph $\dgraph_{\data'}$ is the same  as the graph $\dgraph_{\data}$, giving rise to the same testable cycles. 
Because the data set $\data$ is consistent, so is $\data'$. That $\data$ is fully mixed implies $\data'$ has conditionally full support, and so Corollary~\ref{cor: testable cycles} tells us condition~\eqref{cor: testable cycles - all cycles} is satisfied for $\data'$, hence (because the graphs are the same) for $\data$. 

To see the third fact, note Proposition~\ref{prop: multiplier result for iteratively differentiable case} tells us a sufficient condition for the support-$\menu$ SCR $\dscr^\menu$ to be $\ut$-rationalizable over $\menu\in\DMenu$ is that some $\nuis_\menu\in\lp1$ has $\ut_a=\nuis_\menu+\func^\data_{a,\menu}$ for every $a\in A$. 
Otherwise following verbatim the proof that condition~\eqref{cor: testable cycles - all cycles} implies condition~\eqref{cor: testable cycles - consistency} in Corollary~\ref{cor: testable cycles} delivers the third fact.

Now, we pursue the ``moreover'' part. We will provide examples showing the second and third facts do not generally have converses.

Consider the case in which $A=\{1,2,3\}$, $\DMenu=\{\{1,2\},A\}$, $\Omega=\{0,1\}$, $\dscr^{\{1,2\}}$ has conditionally full support with $\dscr^{\{1,2\}}(0)\neq \dscr^{\{1,2\}}(1)$, $\dscr^{A}_1(0)=\dscr^{A}_1(1)=1$, and $\icost$ is strictly monotone---e.g., $\icost$ could be mutual information as in Example~\ref{ex: Mutual Info Costs}. This data set is fully mixed, and it vacuously satisfies condition~\eqref{cor: testable cycles - all cycles} of Corollary~\ref{cor: testable cycles} because the graph has no nontrivial cycles. Let us observe the data set is not consistent. The reason is that, if it were $\ut$-rationalizable for some $\ut\in\UT$, then $\dscr^{\{1,2\}}$ and $\dscr^{A}$ would both be $\ut$-rationalizable over $\{1,2\}$. This conclusion would contradict Proposition~\ref{prop: linear independence implies unique rationalizability} (applied to the model with menu $\{1,2\}$), which would tell us $\dscr^{\{1,2\}}$ is uniquely $\ut$-rationalizable over $\{1,2\}$.

Consider now the case in which $A=\{0,\tfrac12,1\}$, $\DMenu=\{\{0,\tfrac12\},\{0,1\},\{\tfrac12,1\}\}$, $\Omega=\{0,1\}$,  $\dscr_{\min \menu}^\menu(0)=\dscr_{\max \menu}^\menu(1)=1$ for every $\menu\in\DMenu$, and $\icost$ has derivative $\icostd$ at full information such that $\icostd:\DO\to\real$ is globally differentiable and not affine---e.g., $\icost$ could be given by $\icost(\ip) =\int\posterior(1)^2\ \ip(\dd\posterior)$. The data set clearly has full support. Identifying $\DO$ with $[0,1]$ in the obvious way, direct computation shows every $\posterior\in[0,1]$ has $\icostdd_\posterior(\omega)=\icostd(\posterior)+(\omega-\posterior)\icostd'(\mu)$ for $\omega\in\Omega$. Because $\posterior^\menu_a=\mathbf1_{a=\max\menu}$ for $a\in\menu\in\DMenu$, we can then write $\func^\data_{a,\menu}=\icostd(\mathbf1_{a=\max\menu})+(\omega-\mathbf1_{a=\max\menu})\icostd'(\mathbf1_{a=\max\menu})$, which in particular implies $\func^\data_{a,\menu}(1)-\func^\data_{a,\menu}(0)=\icostd'(\mathbf1_{a=\max\menu})$. The testable cycle $\cycle= 0\ \{0,\tfrac12\}\ \tfrac12\ \{\tfrac12,1\}\ 1\ \{0,1\}\ 0$ therefore has 
$$
\cyclev^\cycle\cdot\func^\data(1)-\cyclev^\cycle\cdot\func^\data(0)
= [\icostd'(1)-\icostd'(0)]+[\icostd'(1)-\icostd'(0)]+[\icostd'(0)-\icostd'(1)] = \icostd'(1)-\icostd'(0) \neq 0.
$$
In particular, $\cyclev^\cycle\cdot\func^\data\neq\mathbf0$, so that condition~\eqref{cor: testable cycles - all cycles} of Corollary~\ref{cor: testable cycles} does not hold. To conclude, we observe that $\data$ is consistent. Indeed, because $\icostd$ is Lipschitz, direct computation shows $\data$ is rationalized by $\ut^\param\in\UT$ given by $\ut^\param_a(\omega):=-\param(a-\omega)^2$ for large enough $\param\in\real$.
\end{proof}

\subsection{On Inducible Belief Distributions}\label{app: sec: bayes plausibility}

In this section, we validate footnote~\ref{footnote: splitting lemma}'s claim that belief distributions and signal structures are equivalent formalisms. To state the equivalence, let us invest in some terminology.

\newcommand{\signal}{m}
\newcommand{\Signal}{M}
\newcommand{\sigmap}{\tau}
\newcommand{\bmap}{\pi}
\newcommand{\Posterior}{B}
\newcommand{\bigP}{\mathbb P}

\begin{definition}
Given a Polish space $\Signal$, a \textbf{signal} is a measurable map $\sigmap:\Omega\to\Delta\Signal$, and a \textbf{belief map} is a measurable map $\bmap:\Signal\to\DO$. For a given pair $(\sigmap,\bmap)$ of such maps, \begin{itemize}
    \item Say the pair $(\sigmap,\bmap)$ is \textbf{Bayes consistent} if every $\hat\Omega\subseteq\Omega$ and $\hat\Signal\subseteq\Signal$ have
        \[
        \int_{\hat\Omega} \sigmap(\hat\Signal|\omega) \ \prior(\dd\omega)
        =\int\int_{\hat\Signal}\bmap(\hat\Omega|\signal) \ \sigmap(\dd\signal|\omega) \ \prior(\dd\omega).
        \]
    \item Say the pair $(\sigmap,\bmap)$ \textbf{generates} $\ip\in\Delta\DO$ if every $\Posterior\subseteq\DO$ has 
    \[\int \sigmap\left(\bmap^{-1}(\Posterior)\mid\omega\right) \ \prior(\dd\omega)=\ip(\Posterior).
    \]
\end{itemize}
\end{definition}

Note every signal admits some Bayes-consistent belief map, and that updated beliefs are almost surely unique, and hence unique in distribution. These observations amount to recording a standard disintegration result in present notation.

\begin{fact}\label{fact: conditional belief existence and uniqueness}
Suppose $\Signal$ is Polish and that $\sigmap$ is a signal. Define $\bigP\in\Delta(\Omega\times\Signal)$ by letting 
\[
\bigP(\hat\Omega\times\hat\Signal):=\int_{\hat\Omega} \sigmap(\hat\Signal|\omega)\ \prior(\dd\omega)
\]
for every measurable $\hat\Omega\subseteq\Omega$ and $\hat\Signal\subseteq\Signal$. Define $\bigP_\Signal:=\marg_\Signal\bigP$. 
\begin{enumerate}[(i)]
    \item Some belief map $\bmap$ exists such that the pair $(\sigmap,\bmap)$ is Bayes consistent.
    \item If belief maps $\bmap,\tilde\bmap$ are such that both $(\sigmap,\bmap)$ and $(\sigmap,\tilde\bmap)$ are Bayes consistent,  $\bmap(\hat\Omega|\signal)=\tilde\bmap(\hat\Omega|\signal)$ for $\bigP_\Signal$-almost every $\signal$. 
    \item Suppose belief maps $\bmap,\tilde\bmap$ are such both $(\sigmap,\bmap)$ and $(\sigmap,\tilde\bmap)$ are Bayes consistent; and suppose $\ip,\tilde\ip\in\Delta\DO$ are such that $(\sigmap,\bmap)$ and $(\sigmap,\tilde\bmap)$ generate $\ip$ and $\tilde\ip$, respectively. Then, $\ip=\tilde\ip$.
\end{enumerate}
\end{fact}
\begin{proof}
Given a belief map $\bmap$, note the pair $(\sigmap,\bmap)$ is Bayes consistent if and only if $\bigP(\hat\Omega\times\hat\Signal)= \int_{\hat\Signal} \bmap(\hat\Omega|\signal) \ \bigP_{\Signal}(\dd\signal)$ for every measurable $\hat\Omega\subseteq\Omega$ and $\hat\Signal\subseteq\Signal$. Hence, (i) and (ii) both follow directly from the disintegration theorem \citep[][Theorem 1.23]{kallenberg2017random}.\footnote{By \cite[][Theorem 19.7]{aliprantis2006infinite}, probability kernels as in \cite{kallenberg2017random} coincide with measurable maps into the space of measures.} Finally, any two ($\DO$-valued) random variables on a probability space that agree almost surely must have the same distribution; hence, (iii) follows directly from (ii).
\end{proof}

The following result shows every signal generates (when pairing it with belief updating to make it Bayes consistent) a unique belief distribution, and that a belief distribution can be generated in this way if and only if it averages to the prior. This classic result---whose proof we include for the sake of completeness---is again a straightforward consequence of the disintegration theorem.

\begin{fact}
Let $\Signal$ be uncountable and Polish, and let $\ip\in\Delta\DO$. The following are equivalent:\footnote{As the proof demonstrates, (i) and (ii) are equivalent and imply (iii), even if $\Signal$ is not assumed uncountable. Moreover, the proof (along with the proof of Fact~\ref{fact: conditional belief existence and uniqueness}) applies verbatim to Polish, non-compact $\Omega$.} \begin{enumerate}[(i)]
    \item Some Bayes-consistent pair $(\sigmap,\bmap)$ generates $\ip$.
    \item Some signal $\sigmap$ is such that every Bayes-consistent $(\sigmap,\bmap)$ generates $\ip$.
    \item The distribution $\ip$ averages to $\prior$, that is, $\ip\in\Info$.\footnote{Recall elements of $\Info$ are those $\ip\in\Delta\DO$ with barycenter $\prior$ or, equivalently, with $\int\posterior(\hat\Omega)\ \ip(\dd\posterior)=\prior(\hat\Omega)$ for every measurable $\hat\Omega\subseteq\Omega$.}
\end{enumerate}
\end{fact}
\begin{proof}

Given any signal $\sigmap$, Fact~\ref{fact: conditional belief existence and uniqueness} tells us some Bayes-consistent pair includes $\sigmap$ and that no two such pairs generate distinct belief distributions. Hence, (i) is equivalent to (ii). In what follows, we show (i) is equivalent to (iii).

To see (i) implies (iii), suppose $(\sigmap,\bmap)$ is Bayes consistent and generates $\ip$. Toward (iii), consider any measurable $\hat\Omega\subseteq\Omega$. Applying Bayes consistency with $\hat\Signal=\Signal$ yields
$\prior(\hat\Omega)=\int\int\bmap(\hat\Omega|\signal) \ \sigmap(\dd\signal|\omega) \ \prior(\dd\omega)$; and because $(\sigmap,\bmap)$ generates $\ip$, any bounded integrable $\func:\DO\to\real$ has $\int\int \func\left( \bmap(\signal)\right) \ \sigmap(\dd\signal|\omega)\ \prior(\dd\omega)=\int \func(\posterior)\ \ip(\dd\posterior).$ Hence, using $\func(\posterior):=\posterior(\hat\Omega)$ gives $\prior(\hat\Omega)=\int \posterior(\hat\Omega)\ \ip(\dd\posterior)$, delivering (iii).

Conversely, suppose (iii) holds, from which we establish (i). By the Borel isomorphism theorem \citep[][Theorem 3.3.13]{srivastava2008course}, some bimeasurable surjection $\bmap:\Signal\to\DO$ exists.\footnote{Moreover, in the special case in which $\Signal\supseteq\DO$, we can take $\bmap$ with $\bmap(\posterior)=\posterior$ for each $\posterior\in\DO$, generating a witnessing $(\sigmap,\bmap)$ in which each signal realization from $\DO$ leads to itself as the updated belief.} Now, define the measure $\bigP\in\Delta(\Omega\times\Signal)$ by letting $\bigP(\hat\Omega\times\hat\Signal):= \int_{\bmap(\hat M)} \posterior(\hat\Omega) \ \ip(\dd\posterior)$
for every measurable $\hat\Omega\subseteq\Omega$ and $\hat\Signal\subseteq\Signal$. By (iii) and because $\bmap$ is surjective, the marginal of $\bigP$ on its first coordinate is $\prior$. Hence, the disintegration theorem \citep[][Theorem 1.23]{kallenberg2017random} delivers some measurable $\sigmap:\Omega\to\Delta\Signal$ such that $\bigP(\hat\Omega\times\hat\Signal)= \int_{\hat\Omega} \sigmap(\hat\Signal|\omega) \ \prior(\dd\omega)$ for every measurable $\hat\Omega\subseteq\Omega$ and $\hat\Signal\subseteq\Signal$. Let us see that $(\sigmap,\bmap)$ witnesses (i). Indeed, for any measurable $\hat\Omega\subseteq\Omega$ and $\hat\Signal\subseteq\Signal$, combining the disintegration property defining $\sigmap$ with the definition of $\bigP$ implies \begin{equation}\label{eqn: formula for constructed signal}
    \int_{\hat\Omega} \sigmap(\hat\Signal|\omega) \ \prior(\dd\omega)
        =\int_{\bmap(\hat M)} \posterior(\hat\Omega) \ \ip(\dd\posterior).
\end{equation}
Specializing equation \eqref{eqn: formula for constructed signal} to the case of $\hat\Omega=\Omega$ and $\hat\Signal=\bmap^{-1}(\Posterior)$ for some measurable $\Posterior\subseteq\DO$ tells us, because $\bmap$ is surjective,  that $(\sigmap,\bmap)$ generates $\ip$. All that remains now is to verify $(\sigmap,\bmap)$ is Bayes consistent. To show it is, define the map $\bigP_\signal\in\Delta\Signal$ by letting $\bigP_\Signal\in\Delta\Signal:=\int\sigmap(\hat\Signal|\omega) \ \prior(\dd\omega)$ for every measurable $\hat\Signal\subseteq\Signal$. We know $(\sigmap,\bmap)$ generates $\ip$; that is, $\bigP_\Signal\circ\bmap^{-1}=\ip$. Hence, every measurable $\hat\Omega\subseteq\Omega$ and $\hat\Signal\subseteq\Signal$ have 
\[
\int\int_{\hat\Signal}\bmap(\hat\Omega|\signal) \ \sigmap(\dd\signal|\omega) \ \prior(\dd\omega)
=\int_{\hat\Signal}\bmap(\hat\Omega|\signal) \ \bigP_\Signal(\dd\signal)
=\int_{\bmap(\hat\Signal)} \posterior(\hat\Omega) \ \ip(\dd\posterior).
\]
Bayes consistency then follows from equation \eqref{eqn: formula for constructed signal}.
\end{proof}

\end{spacing}

\end{document}